\documentclass[reqno]{amsart}
\usepackage{fullpage,amsfonts,amssymb,amsthm,mathrsfs,graphicx,color,framed,esint,stackrel,amsmath}
\usepackage{cancel,bm}
\usepackage{tikz}
\usetikzlibrary{decorations.markings}
\usetikzlibrary{decorations.markings,decorations.pathmorphing}

\def\yr{1}

\usepackage{booktabs,longtable}
\usepackage[initials]{amsrefs}

\DeclareMathAlphabet\mathbfcal{OMS}{cmsy}{b}{n}

\usepackage[colorlinks=true, pdfstartview=FitV, linkcolor=blue, citecolor=blue, urlcolor=blue]{hyperref}
\newcommand{\be}{\begin{equation}}
\newcommand{\ee}{\end{equation}}
\newcommand{\bea}{\begin{eqnarray}}
\newcommand{\eea}{\end{eqnarray}}
\newcommand{\bes}{\begin{equation*}}
\newcommand{\ees}{\end{equation*}}
\newcommand{\beas}{\begin{eqnarray*}}
\newcommand{\eeas}{\end{eqnarray*}}

\renewcommand{\d}{{\mathrm d}}

\newcommand{\im}{\mathrm{i}}
\newcommand{\e}{\mathrm{e}}

\newtheorem{theo}{Theorem}[section]

\newtheorem{lem}[theo]{Lemma}
\newtheorem{rem}[theo]{Remark}
\newtheorem{problem}[theo]{Riemann-Hilbert Problem}
\newtheorem{remark}[theo]{Remark}
\newtheorem{prop}[theo]{Proposition}
\newtheorem{cor}[theo]{Corollary}
\newtheorem{definition}[theo]{Definition}
\newtheorem{conj}[theo]{Conjecture}

\DeclareFontFamily{U}{mathx}{}
\DeclareFontShape{U}{mathx}{m}{n}{<-> mathx10}{}
\DeclareSymbolFont{mathx}{U}{mathx}{m}{n}
\DeclareMathAccent{\widecheck}{0}{mathx}{"71}

\begin{document}
\title{Asymptotic analysis of a Family of Painlevé Functions with Applications to CUE Derivative Moments}
\author{Thomas Bothner}
\address{School of Mathematics, University of Bristol, Fry Building, Woodland Road, Bristol, BS8 1UG, United Kingdom}
\email{thomas.bothner@bristol.ac.uk}

\author{Fei Wei}
\address{Department of Mathematics, University of Sussex, Brighton, BN1 9RH, United Kingdom}
\email{weif0831@gmail.com}

\date{\today}


\begin{abstract}
The Riemann–Hilbert method is employed to carry out an asymptotic analysis of a family of $\sigma$-Painlev\'e V functions associated with Hankel determinants involving the confluent hypergeometric function of the second kind. In the large-matrix limit, this family degenerates to a family of $\sigma$-Painlev\'e III$'$ functions, whose precise asymptotic behavior is also obtained. Both families of Painlev\'e functions arise in the study of the joint moments of the derivative of the characteristic polynomial of a CUE random matrix and the polynomial itself, whose asymptotics are closely related to the moments of the Riemann zeta function and the Hardy~$\mathsf{Z}$-function on the critical line. One of our main results establishes a representation of the leading coefficients of these joint moments in terms of $\sigma$-Painlevé~III$'$ functions for general real exponents. The other main result resolves a question of Assiotis et al. in \cite{ABGS} 
concerning the existence of a probability density for a random variable arising in the ergodic decomposition of Hua-Pickrell measures. 
\end{abstract}

\maketitle

\setcounter{tocdepth}{1}
\tableofcontents

\section{Introduction}

There are deep and far‑reaching connections between Painlev\'e equations and random matrix theory. In many classical random matrix ensembles, quantities such as gap probabilities, correlation functions, and partition functions can be expressed as Fredholm determinants of integrable integral operators, which in turn satisfy nonlinear differential equations of Painlev\'e type. For instance, Tracy and Widom \cite{tracy1994fredholm} showed that the distribution of the largest eigenvalue in the Gaussian Unitary Ensemble (GUE) is governed by a solution of the Painlev\'e II equation. Forrester and Witte \cite{ForresterWittePainleve1,forrester2002application} applied Okamoto's $\tau$-function theory to relate Painlev\'e equations to certain random matrix theory averages taken over various ensembles, including the Laguerre Unitary Ensemble (LUE), Jacobi Unitary Ensemble (JUE), and Circular Unitary Ensemble (CUE).  For a further discussion on the roles of Painlev\'e functions in other aspects of random matrix theory, we refer readers to, e.g., \cite{ItsAlexander}. The isomonodromic deformation method, developed by Jimbo, Miwa, Ueno \cite{JM} and Flaschka, Newell \cite{FN}, provides a fundamental link between  Painlev\'e equations and Riemann–Hilbert problems. Building on this connection, a systematic theoretical framework was developed by Deift, Its, and their collaborators, who established that many random matrix quantities can be characterised by Riemann–Hilbert problems whose solutions link to Painlev\'e-type equations, thereby providing a unified approach for both exact and asymptotic analyses (e.g., \cite{Deift1999, DZ, ItsN}). In this paper, we use the Riemann–Hilbert method to study the connection between Painlev\'e functions and the joint moments of the derivative of the characteristic polynomial of a CUE random matrix and itself.\bigskip 

Let $\mathbb{U}(N)$ denote the group of $N\times N$ unitary matrices, equipped with the normalised Haar measure $\mu_N$. For any $A\in\mathbb{U}(N)$, we denote its characteristic polynomial by
\begin{equation*}
	Z_A(\theta):=\prod_{j=1}^N\Big(1-\e^{\im(\theta_j-\theta)}\Big),\ \ \ \theta\in[0,2\pi),
\end{equation*}
where $\{\theta_j\}_{j=1}^N\subset[0,2\pi)$ are the eigenangles of $A$. Let
\begin{equation*}
	V_A(\theta):=\exp\Bigg[\frac{\im N}{2}(\theta+\pi)-\frac{\im}{2}\sum_{j=1}^N\theta_j\Bigg]Z_A(\theta),
\end{equation*}
so that $V_A:[0,2\pi)\rightarrow\mathbb{R}$ is real-valued and $|V_A(\theta)|=|Z_A(\theta)|$. Our primary interest lies on the joint moments of $V_{A}(\theta)$ and its derivative, namely on
\begin{equation}\label{i1}
	F_N(s,h):=\int_{\mathbb{U}(N)}|V_A(0)|^{2(s-h)}|V_A'(0)|^{2h}\d\mu_N(A),\ \ \ \ \ \textnormal{Re}(h)>-\frac{1}{2},\ \textnormal{Re}(s-h)>-\frac{1}{2}.
\end{equation}
The study of (\ref{i1}) has attracted a lot of attention over the past twenty-five years. Interest in \eqref{i1} originated from the work of Keating and Snaith \cite{keating2000random} on the moments of characteristic polynomials of Haar-distributed unitary random matrices, which is \eqref{i1} with $h=0$. In that work, they derived an explicit formula for (\ref{i1}) (with $h=0$) and proposed the conjecture that the leading coefficients in this case, as $N\rightarrow\infty$, coincide with those appearing in the asymptotic expansion of the mean value of the $(2s)^{\textnormal{th}}$ moment of the Riemann zeta function $\zeta$ on the critical line $\tfrac{1}{2}+\im t, t\in [0,T]$, as $T\rightarrow \infty$. The latter is known as the Keating-Snaith conjecture. Prior to \cite{keating2000random}, no viable conjecture existed for the mean-value moment asymptotics of $\zeta$ on the critical line when $s\in\mathbb{Z}_{\geq 5}$, even under the assumption of the Riemann hypothesis. Several special cases of the Keating-Snaith conjecture have been confirmed by known results in analytic number theory for small values of $s$. We also refer the reader to \cite{gonek2007hybrid} for theoretical evidence supporting this conjecture.\smallskip 

Inspired by \cite{keating2000random}, Hughes \cite{Hughes} subsequently proposed investigating the asymptotic expression of \eqref{i1}, as $N\rightarrow\infty$, for general real $s$ and real $h$ with $-\frac{1}{2}<h<s+\frac{1}{2}$ as follows\footnote{We use $f(x)\sim g(x)$ as $x\rightarrow\infty$ to denote asymptotic equivalence, i.e. $f(x)/g(x)\rightarrow 1$ as $x\rightarrow\infty$},
\bea\label{object}
F_{N}(s,h)\sim F(s,h)N^{s^2+2h}, \quad \textnormal{as} \quad N \rightarrow \infty,
\eea
aiming to obtain the conjectural leading coefficients of the mean value of the joint moments of the Hardy $\mathsf{Z}$-function, $\mathsf{Z}(t),t\in[0,T]$, as $T \to \infty$. Here, $\mathsf{Z}(t)$ is a real-valued counterpart of the Riemann zeta function on the critical line, obtained by multiplying $\zeta(\frac{1}{2}+\im t)$ by a phase factor to make it real for real $t$, compare \eqref{Hardy Z-function}. 

\subsection{A brief history of $F(s,h)$}\label{history} We now give a brief review of the progress on (\ref{object}). For $(s,h)\in \mathbb{Z}_{\geq 1}^2$ with $2h < 2s + 1$, Hughes \cite{Hughes} provided an explicit formula for $F(s,h)$. Conrey, Rubinstein, and Snaith \cite{conreyetal} gave an alternative proof using representations of these moments in terms of the multiple contour integrals introduced in \cite{CFKRS1,CFKRS2}. They found another expression for the leading coefficient $F(s,s)$, in terms of a Hankel determinant of size $s\times s$. The Conrey-Rubinstein-Snaith result was later extended by Bailey et al. \cite{Bailey_2019} to $F(s,h)$, still with $s,h\in\mathbb{Z}_{\geq 1}$, and shown, based on previous results by Forrester and Witte \cite[Section $5$]{Forrester_2006}, to relate to $\sigma$-Painlev\'e III$'$ functions.
%
Independently, Basor et al. \cite[(1.16)]{Basor_2019} obtained a similar result for $F(s,h),(s,h)\in\mathbb{Z}_{\geq 1}^2$ using a Riemann-Hilbert analysis to derive the large $N$ limit of a certain $s \times s$ Hankel determinant of Laguerre polynomials, which relates to $F_N(s,h)$ as established in \cite{W}. Additionlly, Dehaye \cite{Dehaye2008} gave an alternative proof of \eqref{object} for $s,h\in\mathbb{Z}_{\geq 1}$, and obtained another representation of $F(s,h)$ as a certain combinatorial sum. Moreover, Simm and Wei \cite[Theorem 1.6]{simm2024moments} gave an expression for $F(s,h)$, with $s,h\in \mathbb{Z}_{\geq 1}$, in terms of the Bessel kernel. \smallskip

For $s \in \mathbb{Z}_{\geq 1}$ and $2h \in \mathbb{Z}_{\geq 1}$ odd with $2h < 2s + 1$, Winn \cite{W} proved \eqref{object} and derived an explicit expression for $F(s,h)$ in terms of sums over partitions. Later, Assiotis et al. \cite{ABGS} used a probabilistic approach to obtain an alternative expression for $F(s,h)$ in this case. For any real $2s > -1$ and $0<2 h <2s +1$, Assiotis, Keating, and Warren \cite{AKW} proved that 
 (\ref{object}) holds, and they provided a probabilistic representation of $F(s,h)$ in terms of a real-valued random variable appearing in the ergodic decomposition of the Hua–Pickrell measure. Furthermore, Assiotis et al. \cite[Remark 2.8]{ABGS} gave an alternative expression for $F(s,h)$ in terms of an infinite series. Moreover, in Remark 1.3 of \cite{ABGS}, they also expressed $F(s,h)$ in terms of a $\sigma$-Painlev\'e III$'$ function for $2s > -1$ and $h \in \mathbb{Z}_{\geq 1}$. 
 For exponents with $s \notin \mathbb{Z}_{\geq 1}$ and $h \notin \mathbb{Z}_{\geq 1}$, to the best of our knowledge, no explicit results exist concerning the connection between $F(s,h)$ and Painlev\'e functions. In this work, we address this previously unexplored question by providing an explicit representation of $F(s,h)$ in terms of a $\sigma$-Painlev\'e III$'$ function for real $s > -\tfrac{1}{2}$ and complex $h$, as follows.
 \subsection{$F(s, h)$ in terms of a Painlev\'e function for general exponents $s$ and $h$.}
\begin{theo}\label{maintheorem1}
Let $s > -\tfrac{1}{2}$ and $h \in \mathbb{C}$ with $0 \leq \mathrm{Re}(h) < \tfrac{1}{2} + s$. Then \eqref{object} holds with 
\begin{equation*}
F(s,h)
= \frac{G^2(1+s)}{G(1+2s)} \times
\begin{cases}
(-1)^h \displaystyle\frac{\mathrm{d}^{2h}}{\mathrm{d}t^{2h}}
\exp\left[
\int_{0}^{t} v(z; s)\, \frac{\mathrm{d}z}{z}
\right]\bigg|_{t=0}, & h \in \mathbb{Z}_{\ge 0}\bigskip\\[1.0ex]
-\dfrac{2}{\pi} \sin(\pi h)\, \Gamma(2h - 2M)
\displaystyle\int_{0}^{\infty} \frac{\mathrm{d}^{2M+1}}{\mathrm{d}t^{2M+1}}
\exp\left[
\int_{0}^{t} v(z; s)\, \frac{\mathrm{d}z}{z}
\right]
\frac{\d t}{t^{2h-2M}}, & h \notin \mathbb{Z}_{\ge 0}
\end{cases},
\end{equation*}
where $G(z)$ is the Barnes' $G$-function, cf. \cite[$\S 5.17$]{NIST}, $\Gamma(z)$ Euler's Gamma function, cf. \cite[$\S 5.2.1$]{NIST}, and $M\in\mathbb{Z}_{\geq 0}$ 
such that $\mathrm{Re}(h) \in (M, M+1)$.
The function $v(z)=v(z; s)$ satisfies the $\sigma$-Painlev\'e \textnormal{III}$'$ equation,
\begin{equation}\label{e42}
\left( z \frac{\mathrm{d}^2 v}{\mathrm{d}z^2} \right)^{\!2}
= \left( 4\left( \frac{\mathrm{d}v}{\mathrm{d}z} \right)^{\!2} - 1 \right)
\left( v + s^2 - z \frac{\mathrm{d}v}{\mathrm{d}z} \right) + s^2,\ \ \ \ \ \ z>0,\ \ s>-\frac{1}{2},
\end{equation}
cf. \cite{O}, and it has the explicit form
\begin{equation}\label{e73repeat}
v(z; s)
= -\im \sqrt{\frac{z}{\beta}}\,
Q_1^{21}\left( 2\sqrt{\frac{z}{\beta}}, s \right)
+ \frac{1}{16} - s^2 - \frac{z}{2},
\end{equation}
where \(Q_1^{21}(x,s)\) is the \((2,1)\)-entry of \(Q_1(x,s)\) appearing in the unique solution of RHP \ref{PIIImodel}, see \eqref{A4}.
\end{theo}

We remark that the integral representation of \(F(s,h)\) in Theorem \ref{maintheorem1} also remains well-defined  for \(-1 < \mathrm{Re}(2h) < 0\) by taking \(M = -1/2\). The reason this case is not included in the statement is that we rely on the convergence result in \cite[Theorem~1.2]{AKW}, which asserts that the limit \eqref{object} exists for \(0 \leq 2h < 2s + 1\). This result can be extended to \(0 \leq \mathrm{Re}(2h) < 2s + 1\), see Proposition \ref{exofh} below. This, in turn, determines the validity range of \((\ref{object})\). However, we show that the expression of $F(s,h)$ in Theorem \ref{maintheorem1} is an analytic function in $h$ on
$
\{ h \in \mathbb{C} : -1 < \mathrm{Re}(2h) < 2s + 1 \},
$
in Remark \ref{analyticfunction}. So if we temporarily disregard the convergence issue for exponents $h\in (-1/2,0)$, Theorem \ref{maintheorem1} actually provides a representation of $F(s,h)$ in terms of a solution to the $\sigma$-Painlev\'e III$'$ equation \eqref{e42},\eqref{e73repeat} for the full range of the real exponent $h$, and for almost the full range of the complex exponent $h$ for fixed $s$, except along the sets where $\mathrm{Re}(h) \in \mathbb{Z}_{\geq 0} $ and $\textnormal{Im}(h)\neq 0$. Nevertheless, these exceptional lines can be analytically continued away.

\subsubsection{Comparison with previous results}

For \( s \in \mathbb{Z}_{\geq 1} \) and \( h \in \mathbb{Z}_{\geq 1} \), Theorem~\ref{maintheorem1} agrees with the results obtained in \cite[Section~5]{Forrester_2006} and \cite[(1.16)]{Basor_2019}, as discussed in Subsection \ref{history}. In \cite{Forrester_2006,Basor_2019}, the Painlev\'e function \( v(z;s) \) was expressed in terms of the modified \( I \)-Bessel function, cf. \cite[$\S10.25.2$]{NIST}, as  
\begin{align}\label{1104e1}
v(z;s)=z\frac{\d }{\d z}\ln \left(
\e^{-\frac{z}{2}} z^{-\frac{1}{2}s^2}
\det\!\Big[ I_{j+k+1}\!\left( 2\sqrt{z} \right) \Big]_{j,k=0}^{s-1}
\right), \quad z>0.
\end{align}
For \( s>-1/2 \) and \( h \in \mathbb{Z}_{\geq 1} \), Theorem~\ref{maintheorem1} also coincides with \cite[Remark~2.8]{ABGS}, where \( v(z;s) \) is represented in terms of the characteristic function of a random variable. Furthermore, for \( s \in \mathbb{Z}_{\geq 1} \) and \( 0<2h<2s+1 \), combining (\ref{1104e1}) with Theorem~\ref{maintheorem1} re-obtains the computation given in \cite[Remark~2.7]{ABGS}. Moreover, in this case, if we apply Theorem~\ref{maintheorem1} with \(M = -1/2\) and combine it with the power series expansion of the function inside the logarithm on the right-hand side of (\ref{1104e1}) with $z=2t$, followed by term-by-term integration, the resulting expression coincides with the formula for \(F(s,h)\) given in equation~(26) of \cite[Theorem~1.5]{ABGS}.
Finally, for \( s=0 \), the Riemann–Hilbert approach used in the present work yields \( v(z;0)=-z/2 \). Together with the expression provided in Theorem~\ref{maintheorem1}, this leads to an explicit evaluation of \( F(0,h) \) for \( -1<\mathrm{Re}(2h)<1 \), in agreement with \cite[Theorem~1.4]{ABGS}.

\vspace{.2cm}

Following the motivation for studying (\ref{object}) explained in the paragraph preceding Subsection~\ref{history}, we recall that \(F(s,h)\) is expected to capture the leading order coefficients of the integral mean of the moments of the Hardy \(\mathsf{Z}\)-function \(\mathsf{Z}(t)\) and its derivative on \([0,T]\) as \(T \to \infty\). It is worth mentioning that all computations of \(F(s,h)\) from the list of results in Subsection~\ref{history} are consistent with the corresponding results for the Hardy \(\mathsf{Z}\)-function at certain values of \(s\) and \(h\), which have been established in number theory. Theorem~\ref{maintheorem1} provides a representation of \(F(s,h)\) in terms of a distinguished solution to the \(\sigma\)-Painlev\'e~III$'$ equation \eqref{e42}. This may be viewed as heuristic evidence suggesting the existence of a connection between the leading coefficients of the integral mean of the joint moments of Hardy’s \(\mathsf{Z}\)-function and integrable systems, an intriguing direction for further investigation. To make this conjectural connection more explicit, we propose the following conjecture to conclude the current subsection.

\begin{conj}\label{conjectureinnumbertheory}
Let $s>-\frac{1}{2}$ and $-\frac{1}{2} <h<s+\frac{1}{2}$. Let $v(z;s)$ be as given in \eqref{maintheorem1}. Then as $T\rightarrow\infty$,
\begin{align*}
&\frac{1}{T} \int_0^T |\mathsf{Z}'(t)|^{2h} |\mathsf{Z}(t)|^{2s-2h}dt\\
\sim\, &a_s(\log T)^{s^2+2h}\times
\begin{cases}
(-1)^h\displaystyle\frac{\mathrm{d}^{2h}}{\mathrm{d}t^{2h}}
\exp\left[
\int_{0}^{t} v(z; s)\, \frac{\mathrm{d}z}{z}
\right]\bigg|_{t=0}, & h \in \mathbb{Z}_{\ge 0}\bigskip, \\[1.0ex]
-\dfrac{2}{\pi} \sin(\pi h)\, \Gamma(2h - 2M)
\displaystyle\int_{0}^{\infty}\frac{\mathrm{d}^{2M+1}}{\mathrm{d}t^{2M+1}}
\exp\left[
\int_{0}^{t} v(z; s)\, \frac{\mathrm{d}z}{z}
\right]
\frac{\mathrm{d}t}{t^{2h-2M}}, & h \notin \mathbb{Z}_{\ge 0},
\end{cases}
\end{align*}
where $M\in\mathbb{Z}_{\geq 0}$ is such that $h\in (M, M+1)$ and we set $M=-\frac{1}{2}$ when $h\in (-\frac{1}{2},0)$.
The arithmetic factor $a_s$ is defined by
\[
a_s = \prod_{p} \left(1 - \frac{1}{p} \right)^{s^2} \sum_{m=0}^\infty \left( \frac{\Gamma(m+s)}{\Gamma(m+1) \cdot \Gamma(s)} \right)^2  p^{-m},
\]
and Hardy's $\mathsf{Z}$-function by
\begin{equation}
\label{Hardy Z-function}
\mathsf{Z}(t) = \pi^{-\im t/2}\frac{\Gamma(\frac{1}{4} + \frac{\im}{2}t)}{|\Gamma(\frac{1}{4}+\frac{\im}{2}t)|}\,\zeta\Big(\frac{1}{2}+\im t\Big),\ \ t\in\mathbb{R}.
\end{equation}
\end{conj}

\subsection{Riemann-Hilbert approach}

To apply the Riemann–Hilbert method in the proof of Theorem~\ref{maintheorem1}, we begin by expressing (\ref{i1}) as an integral involving the Hankel determinant.
\bea\label{NtimesNHankel}
\mathcal{J}_{N}(t;s):=\det\!\Big[U(1-s-N,\,2-2s-2N+j+k,2t)\Big]_{j,k=0}^{N-1},
\eea
where $t>0$ and $U(a,b,z)$ denotes the confluent hypergeometric function of the second kind, cf. \cite[$\S 13.2.6$]{NIST}.
This integral involves the function, with $\epsilon>0$,
\begin{equation}\label{i5}
	K_h^{\epsilon}(t)
	:=\frac{1}{2\pi}\int_{-\infty}^{\infty}|x|^{h}\e^{-\epsilon|x|-\im t x} \d x
	=\frac{\Gamma(1+h)}{2\pi}
	\left[
	\frac{1}{(\epsilon+\im t)^{h+1}}
	+\frac{1}{(\epsilon-\im t)^{h+1}}
	\right],
	\quad t\in\mathbb{R},
\end{equation}
which is well-defined for all $\operatorname{Re}(h)>-1$. Here the branch of $z^h$, $z\in\mathbb{C}\setminus(-\infty,0]$, is chosen to be the principal one.

\subsubsection{A representation of $F_{N}(s,h)$ in terms of a Hankel determinant}

\begin{prop}\label{1102prop1}
For any $\operatorname{Re}(s)>-\tfrac{1}{2}$ and $0<\operatorname{Re}(h)<\operatorname{Re}(s)+\tfrac{1}{2}$, we have
\begin{equation}\label{1104i7}
F_N(s,h)
=\frac{(-1)^{N(N-1)/2}\,2^{1-2h}}{(\Gamma(s+N))^{N}}
\lim_{\epsilon\downarrow 0}\int_0^{\infty}
K_{2h}^{\epsilon}(t)\e^{-Nt}\mathcal{J}_{N}(t;s)\,\d t.
\end{equation}
\end{prop}
In Remark \ref{rangeofFN}, we explain that when $s$ is real, the validity range of \eqref{1104i7} can be extended to 
$-1<\operatorname{Re}(2h)<2s+1$.
Proposition \ref{1102prop1} is motivated by Winn’s work \cite[$(3.9),(4.1),(4.22),(4.29)$]{W}, where a representation of $F_{N}(s,h)$ was obtained for integer parameters $s\in\mathbb{Z}_{\ge1}$ and $2h\in\mathbb{Z}_{\ge1}$. Namely,
\begin{equation}\label{1102i7}
F_N(s,h)
=(-1)^{s(s-1)/2}2^{1-2h}
\lim_{\epsilon\downarrow 0}\int_0^{\infty}
\frac{(-1)^{2h}}{\pi}\frac{\partial^{2h}}{\partial\epsilon^{2h}}
\left(\frac{\epsilon}{\epsilon^2+t^2}\right)
\e^{-Nt}
\det\Big[L_{N+s-1-(j+k)}^{(2s-1)}(-2t)\Big]_{j,k=0}^{s-1}\d t,
\end{equation}
where $L_{n}^{(\alpha)}$ denotes the classical Laguerre polynomial. From \eqref{i5}, when $2h\in\mathbb{Z}_{\ge0}$, the function $K_{2h}^{\epsilon}(t)$ reduces to
\begin{equation}\label{e38}
K_{2h}^{\epsilon}(t)
=\frac{(-1)^{2h}}{\pi}\frac{\partial^{2h}}{\partial\epsilon^{2h}}
\!\left(\frac{\epsilon}{\epsilon^2+t^2}\right),\ \ \ t\in\mathbb{R},
\end{equation}
which is precisely the Poisson-type expression appearing in \eqref{1102i7}. Thus, our $K_{2h}^{\epsilon}(t)$ in \eqref{i5} is a generalisation of \eqref{e38} used in \cite{W}. Since we deal with general complex parameters $s$, we no longer have a Hankel determinant of size $s\times s$, as the one in \eqref{1102i7}, which requires $s$ to be a positive integer.
Motivated by the work of Assiotis, Gunnes, Keating, and Wei~\cite[Proposition 4.2]{AGKW}, which connects the joint moments of higher derivatives of CUE characteristic polynomials to Hankel determinants shifted by integer partitions for general moment exponents $s$, we instead employ the Fourier transform of a Cauchy-type weight to obtain the Hankel determinant $\mathcal{J}_{N}(t;s)$ appearing in (\ref{1104i7}).
Moreover, as presented in Proposition~\ref{relationtoaHankel}, when $s\in\mathbb{Z}_{\ge1}$, the determinant $\mathcal{J}_{N}(t;s)$ can be transformed into the Hankel determinant appearing in (\ref{1102i7}).\smallskip

From Propositions \ref{expressionofJ} and \ref{integralrepresentation}, we see that $t\mapsto \mathcal{J}_{N}(t;s)$ has no zeros on $(0,\infty)\subset\mathbb{R}$ for real $s>-1/2$, though zeros may appear on $(0,\infty)$ for complex $s$. Our derivation of Theorem \ref{maintheorem1} relies on the scaled large-$N$ limit of $\mathcal{J}_{N}(t;s)$. Exceptional zeros of $t\mapsto\mathcal{J}_N(t;s)$ at finite $N$ could accumulate as $N$ grows, potentially leading to poles of the Painlev\'e function $v(z;s)$ in \eqref{e42} for $\textnormal{Re}(z)>0$. Since the locations of these accumulations are unclear, choosing integration paths in Theorem \ref{maintheorem1} that avoid such poles becomes problematic. Therefore, in this paper, we restrict our attention solely to real $s>-1/2$, when $v(z;s)$ will be pole-free on $(0,\infty)\ni z$.\smallskip

Moving ahead, when $s>-1/2$, it is known from Painlev\'e theory (see, e.g., \cite[Proposition 3.2]{forrester2002application}) that $\mathcal{J}_{N}(t;s)$ is the $\tau$-function of a $\sigma$-Painlev\'e V equation. More precisely, for $t>0$, we have
\begin{equation}\label{1102defu}
\mathcal{J}_N(t;s) = \mathrm{C}_N(s) \exp\left[\int_0^{2t} u_N(z;s) \frac{\mathrm{d}z}{z} \right],
\end{equation}
where $u_{N}(z;s)$ solves the $\sigma$-Painlev\'e V equation
\begin{align}\label{1102PainleveV}
\bigg(z\frac{\d^2u_N}{\d z^2}\bigg)^2=\bigg(u_N-Ns\,-\,&z\frac{\d u_N}{\d z}+2\Big(\frac{\d u_N}{\d z}\Big)^2-2N\frac{\d u_N}{\d z}\bigg)^2\nonumber\\
	&-4\frac{\d u_N}{\d z}\bigg(s+\frac{\d u_N}{\d z}\bigg)\bigg(N+s-\frac{\d u_N}{\d z}\bigg)\bigg(N-\frac{\d u_N}{\d z}\bigg).
\end{align}
Here, $\mathrm{C}_{N}(s)$ is a non-zero constant whose explicit expression can be obtained from Proposition \ref{expressionofJ} and (\ref{Selb}). What is needed in \eqref{1104i7} to get to \eqref{object} is the following scaled large $N$-asymptotic formula for $u_N(z/N;s)$, which we derive in Section \ref{sectiononlargeN}.

\subsubsection{Large $N$-asymptotic of $u_{N}(z;s)$}

\begin{prop}\label{asympformulaforlargeN} Let $u_{N}(z;s)$ be as in \eqref{1102defu}. 
As $N\rightarrow\infty$, uniformly for $z>0$ and $s>-\frac{1}{2}$ chosen on compact subsets,
\begin{equation}\label{e72}
	u_N(z/N;s)=-\im\sqrt{\frac{z}{\beta}}\,Q_1^{21}\Big(2\sqrt{\frac{z}{\beta}},s\Big)+\frac{1}{16}-s^2+\mathcal{O}\big(N^{-1}\big).
\end{equation}
Here the error term is twice $z$-differentiable on $(0,\infty)$, and $Q_1^{21}(\cdot,s)$ is as in Theorem \ref{maintheorem1}.
\end{prop}
It will follow from \eqref{e72} that the large $N$-limit of $u_N(z/N;s)-z/2$, pointwise in $z>0$ and $s>-1/2$, is $v(z;s)$ in \eqref{e42}, see Corollary \ref{cheat}. To obtain Proposition \ref{asympformulaforlargeN} we will exploit a Riemann-Hilbert representation of $u_N$ that stems from \eqref{i13} and \eqref{e43}, compare our workings in Section \ref{sectiononlargeN}. In detail, it will be a representation of $u_N$ in terms of the Riemann-Hilbert problem \ref{RHP3}, see \eqref{e49}, which can be resolved asymptotically, as $N\rightarrow\infty$, by Deift-Zhou nonlinear steepest descent techniques \cite{DZ,DKMVZ}. What results is the leading order \eqref{e72} that makes use of the model problem \ref{PIIImodel}. Equipped with \eqref{e72}, the reader will notice that Theorem \ref{maintheorem1} also requires the following asymptotic formul\ae\, of $z\mapsto v(z;s)$. The same function is smooth on any compact subset of $(0,\infty)\subset\mathbb{R}$, see Proposition \ref{Fredana} and Remark \ref{mero}.


\subsubsection{Large $z$-asymptotic of $v(z;s)$}
\begin{cor}\label{cor5555} 
With $v=v(z;s)$ as in \eqref{e73}, 
\begin{equation}\label{f1000}
	v(z;s)=-\frac{z}{2}+s\sqrt{z}-\frac{3}{4}s^2+\mathcal{O}\big(z^{-\frac{1}{2}}\big),\ \ \ z\rightarrow+\infty,
\end{equation}
uniformly in $s>-\frac{1}{2}$ on compact sets.
\end{cor}
Our approach to prove Corollary \ref{cor5555} is as follows. With \eqref{e72} one simply subjects RHP \ref{PIIImodel} to a Deift-Zhou nonlinear steepest descent analysis, as done in Section \ref{sectiononlargex}. The same approach allows us to derive also subleading terms in \eqref{f1000}, see Remark \ref{remarkaftercorollary49}. Once done, in Section  \ref{sectiononsmallx}, we obtain the asymptotic formula for $v(z;s)$ when $z\downarrow 0$. As a direct consequence of Proposition \ref{loopy}, we have the following conclusion.

\subsubsection{Small $z$-asymptotic of $v(z;s)$}
\begin{prop}\label{1102asymoforsmallz}
Let $s>-\frac{1}{2}$. Then there exists $z_{0}=z_{0}(s)\in(0,1)$ such that for $0<z\leq z_{0}<1$, we have the following asymptotic expansions: for fixed odd $2s\in\mathbb{Z}_{\geq 1}$,
\begin{align*}
v(z;s) =\sum_{k=1}^{s+\frac{1}{2}}\mathsf{d}_{2k}(s)z^{2k}+\mathsf{d}_{1,1}(s)z^{2s+1}\ln z+
\mathcal{O}\big(z^{2s+2}\ln z\big),
\end{align*}
and for fixed $2s\notin\mathbb{Z}_{\geq 0}$, provided that $s\in(m,m+1],m\in\mathbb{Z}_{\geq 0}$ or $s\in(-\frac{1}{2},0)$, 
\begin{align*}
v(z;s)=\sum_{k=1}^{m+1}\mathsf{d}_{2k}(s)z^{2k}+\mathsf{d}_{1,1}(s)z^{2s+1}+\mathcal{O}\big(z^{2s+2}\big),
\end{align*}
where we set $m=0$ when $s\in (-\frac{1}{2},0)$. Moreover, if $2s\in\mathbb{Z}_{\geq 1}$ is even and fixed, then
\begin{equation}\label{powerexpansionforintegers}
v(z;s)= \sum_{k=1}^s\mathsf{d}_{2k}(s)z^{2k}+\mathsf{d}_{2s+1}(s)z^{2s+1}+\mathcal{O}\big(z^{2s+2}\big).
\end{equation}
\end{prop}
A detailed derivation of \eqref{powerexpansionforintegers} is given in  \cite[Section 4.3]{Basor_2019}. Also, we emphasize that $v(z;0)=-z/2$ for any $z>0$ which corresponds to $s=0$ not addressed in Proposition \ref{1102asymoforsmallz}. What's important to notice in Proposition \ref{1102asymoforsmallz} is the gap structure ($\mathsf{d}_{2j-1}(s)=0$ for $1\leq j<s+1$), obtained from the $\sigma$-Painlev\'e III$'$ equation (\ref{e42}), for it will be crucial in the proof of Theorem \ref{maintheorem1}. Our derivation of Proposition \ref{1102asymoforsmallz} is based on another Deift-Zhou nonlinear steepst descent analysis of RHP \ref{PIIImodel}, worked out in Section \ref{sectiononsmallx}. Using this approach, we can expand  $v(z;s)$ into an asymptotic series as $z\downarrow 0$, and all coefficients in the same series can, in principle, be obtained iteratively. For example, as computed in Corollary \ref{nzero}, 
\begin{align*}
&\mathsf{d}_{1,1}(s)=-\frac{1}{\Gamma(2+2s)}\bigg(\frac{\Gamma(1+s)}{\Gamma(1+2s)}\bigg)^2\frac{1}{2\cos(\pi s)}, \,\,\,\,2s\notin\mathbb{Z}_{\geq 0},\\
&\mathsf{d}_{1,1}(s)=-\frac{1}{\Gamma(2+2s)}\bigg(\frac{\Gamma(1+s)}{\Gamma(1+2s)}\bigg)^2\frac{\e^{2\pi\im s}}{\pi}\sin(\pi s), \,\,\,\,2s\in\mathbb{Z}_{\geq 1}~ \textnormal{odd}.
\end{align*}

\subsection{An integral representation of $F(s,h)$ for general exponents $s$ and $h$}

Another ingredient in our proof of Theorem \ref{maintheorem1} is the following representation of the large-$N$ limit of $F_N(s,h)/N^{s^2+2h}$ in terms of the Painlev\'e function $v(z;s)$ and the function $K_{2h}^{\epsilon}(t)$.

\begin{prop}\label{1103prop}
Let $s>-\frac{1}{2}$. Let $h\in \mathbb{C}$ with $0\leq\textnormal{Re}(h)<\frac{1}{2}+s$. Then \eqref{object} holds and
\begin{equation}
F(s,h)=\frac{G(s+1)^2}{G(2s+1)}2^{1-2h}\lim_{\epsilon\downarrow 0}\int_{0}^{\infty}K_{2h}^{\epsilon}(t)\exp\left[\int_0^{2t}v(z;s)\frac{\d z}{z}\right]\d t
\end{equation}
with $K_{h}^{\epsilon}(t)$ and $v(z;s)$ defined in \eqref{i5} and Theorem \ref{maintheorem1}, respectively.
\end{prop}
Our approach to proving Proposition~\ref{1103prop} is based on the probabilistic method introduced in \cite{AKW}, in which the authors used a family of random variables \(\mathsf{X}(s)\) to study the convergence of \( F_N(s,h)/N^{s^2+2h} \) as \( N \to \infty \).

\begin{lem}\cite[Theorem 1.2]{AKW}\label{convergenceat0}
Let $s>-\frac{1}{2}$ and $0<h<s+\frac{1}{2}$. Then \eqref{object} holds and
\begin{equation*}
F(s,h)=\frac{G(s+1)^2}{G(2s+1)}2^{-2h}\mathbb{E}\left[\left|\mathsf{X}(s)\right|^{2h}\right]<\infty.
\end{equation*}
\end{lem}
We extend this convergence result to the range  $0 \leq \mathrm{Re}(2h) < 2s + 1$ in Proposition~\ref{exofh}. Before presenting further results, we introduce a precise definition of the random variable $\mathsf{X}(s)$. 

\begin{definition}\label{defofxs}
Let \( s \in \mathbb{R} \) with \( s > -\tfrac{1}{2} \).  
Let \( \mathbf{C}^{(s)} \) be the determinantal point process on 
\((-\infty, 0) \cup (0, \infty)\) with correlation kernel
\begin{equation*}
\mathsf{K}^{(s)}(x, y) = 
\frac{1}{2\pi}
\frac{(\Gamma(s + 1))^2}{\Gamma(2s + 1)\Gamma(2s + 2)}
\cdot
\frac{\mathsf{G}^{(s)}(x) \mathsf{H}^{(s)}(y) - \mathsf{G}^{(s)}(y) \mathsf{H}^{(s)}(x)}{x-y},
\end{equation*}
where \( \mathsf{G}^{(s)}(x) \) and \( \mathsf{H}^{(s)}(x) \) are given by the formulae
\begin{align*}
\mathsf{G}^{(s)}(x)&= 2^{2s-\frac{1}{2}} \Gamma \left(s + \tfrac{1}{2}\right)
\, |x|^{-\frac{1}{2}} 
J_{s - \frac{1}{2}} \left(\frac{1}{|x|}\right), \\
\mathsf{H}^{(s)}(x) &= \textnormal{sgn}(x)2^{2s + \frac{1}{2}} \Gamma \left(s + \tfrac{3}{2}\right)
\, |x|^{-\frac{1}{2}} 
J_{s + \frac{1}{2}}\!\left(\frac{1}{|x|}\right),
\end{align*}
and \( J_\nu \) denotes the Bessel function of the first kind with parameter \( \nu \), cf. \cite[$\S 10.2.2$]{NIST}.
Then the random variable \( \mathsf{X}(s) \) is defined as the following principal value sum over the points of \( \mathbf{C}^{(s)} \),
\begin{equation}\label{defofran}
\mathsf{X}(s) = \lim_{m \to \infty} 
\left[
\sum_{x\in \mathbf{C}^{(s)}}x \mathbf{1}_{|x|> m^{-2}}
\right].
\end{equation}
\end{definition}
It is shown in Qiu \cite[Theorem 1.2]{Qiu} that (\ref{defofran}) is well-defined. The appearance of the random variable $\mathsf{X}(s)$ can be traced back to the ergodic decomposition of the Hua–Pickrell measures on the space of infinite Hermitian matrices, where it arises in a natural way \cite{Borodin_2001,Qiu}. In this section, we do not provide further details about the background of \(\mathsf{X}(s)\) and only state specific facts about it that are needed in Section~\ref{sectiononexpf}. For readers interested in more background on \(\mathsf{X}(s)\), we refer to \cite{Borodin_2001, Qiu, AKW, ABGS}.

\subsection{A question of Assiotis, Bedert, Gunnes and Soor posed in \cite[Remark 2.5]{ABGS}} Returning to Proposition~\ref{1103prop}, a key step in its proof is to show that the probability density of \(\mathsf{X}(s)\), with respect to the Lebesgue measure on \(\mathbb{R}\), exists for any \(s > -1/2\). This provides, en route, a positive answer to the question posed in \cite[Remark 2.5]{ABGS}, where the authors inquired about the probability density function for the law of \(\mathsf{X}(s),s>0\). The existence of its density had previously been established for \(s \in \left(-1/2,0\right) \cup \mathbb{Z}_{\geq 0}\) in \cite{ABGS}.

\begin{theo}\label{expf}
For any $s>-\frac{1}{2}$,
the distribution of $\mathsf{X}(s)$ admits a bounded, smooth probability density function with respect to the
Lebesgue measure on $\mathbb{R}$.
\end{theo}

By Fourier inversion, existence of the probability density of \(\mathsf{X}(s)\) with respect to the Lebesgue measure follows from the exponential decay of its characteristic function,
\begin{equation*}
\mathbb{E}\big[\mathrm{e}^{\im t \mathsf{X}(s)}\big] = \mathcal{O}\big(\mathrm{e}^{-\frac{t}{2}(1-\epsilon)}\big),\ \ \ t>0,
\end{equation*}
for any sufficiently small \(\epsilon\), as was anticipated in \cite[Remark 2.5]{ABGS}. More precisely, Theorem~\ref{expf} is a consequence of the asymptotic analysis of \(v(z;s)\) as \(z \downarrow 0\) and \(z \to \infty\), obtained in Proposition~\ref{1102asymoforsmallz} and Corollary~\ref{cor5555}, respectively, together with the following representation formula. 

\begin{prop}\label{1103addprop}
Let  $s>-\frac{1}{2}$. Let $\rho^{(s)}(x)$ be the probability density  of  the distribution of $\mathsf{X}(s)$. Then for any $x\in \mathbb{R}$, 
\bea\label{densityformula}
\rho^{(s)}(x)=\frac{1}{\pi}\textnormal{Re}\left(\int_{0}^{\infty}\e^{\im xt}\exp\left[\int_0^{t}v(z;s)\frac{\d z}{z}\right]\d x\right),
\eea
with $v(z;s)$ as in Theorem \ref{maintheorem1}.
\end{prop}

Since \(\mathsf{X}(s)\) is a continuous random variable, its density function can be used to compute its moments,
\[
\mathbb{E}\Big[|\mathsf{X}(s)|^{2h}\Big] = \int_{-\infty}^{\infty} |x|^{2h} \rho^{(s)}(x) \, \mathrm{d}x.
\]  
See \cite[Theorem 1.5]{ABGS} for explicit expressions in case \(s \in \mathbb{Z}_{\geq 0}, -1<\textnormal{Re}(2h)<2s+1\). By substituting \eqref{densityformula}, and using the asymptotic analysis of \(v(z;s)\) to justify the exchange of certain limits and integrations, we obtain the integral representation of \(F(s,h)\) stated in Proposition \ref{1103prop}. We refer the reader to Section \ref{sectiononexpf} for further details.\bigskip

Finally, we discuss an application of Theorem~\ref{maintheorem1}, which provides a necessary and sufficient condition for the validity of a question posed in \cite[Remark~2.8]{ABGS} in the study of the rationality of the coefficients in the series expansion of \(F(s,h)\) for \(s \in \mathbb{Z}_{\geq 1}\). In that Remark the authors conjectured that, for \(m \in \mathbb{Z}_{\ge 1}\),
\begin{equation}\label{remark2.8}
\lim_{h \to m+\frac{1}{2}} \lim_{N \to \infty} \frac{F_{N}(m,h)}{N^{m^{2}+2h}} = +\infty,
\end{equation}
and verified this for \(m = 1,2,3,4\), while noting that the case \(m \ge 5\) remains open.

\begin{prop}\label{rule}
Let \(m \in \mathbb{Z}_{\ge 1}\). Then \eqref{remark2.8} holds if and only if $\mathsf{d}_{2m+1}(m)\neq 0$. Here, $\mathsf{d}_{2m+1}(m)$ appears in \eqref{powerexpansionforintegers}.
\end{prop}

\begin{remark}
The expression \eqref{1104e1} for $v(z; s)$ with $s=m$ and the recursive formula for the coefficients in the asymptotic expansion of  
\[
\det\!\left[ I_{j+k+1}\!\left( 2\sqrt{|z|} \right) \right]_{j,k=0}^{m-1}
\quad \text{as } z \to 0,
\]
see~\cite[Theorem 5]{keating-fei} and \cite[Proposition 4.3]{forrester2025higher}, provides an efficient way to check whether \(\mathsf{d}_{2m+1}(m) \ne 0\).  
Hence, by Proposition~\ref{rule}, we can verify that \eqref{remark2.8} holds for many values of \(m\). For general $m$, since $v(z;s)$ satisfies the $\sigma$-Painlev\'e \textnormal{III}$'$ equation \eqref{e42} with $s=m$, the equation admits a family of solutions depending on the coefficient of $z^{2m+1}$ in the small-$z$ asymptotic expansion. Hence, if $\mathsf{d}_{2m+1}(m) = 0$, we would have the following gap structure:
\begin{equation}\label{1028e3}
v(z;m) = \sum_{k=1}^{\infty} \mathsf{d}_{2k}(m)z^{2k},\ \ \ z\downarrow 0.
\end{equation}
However, for our specific solution $v(z;s)$ in \eqref{1104e1}, the gap structure \eqref{1028e3} should not occur. In other words, we believe that $\mathsf{d}_{2m+1}(m) \neq 0$, and therefore \eqref{remark2.8} holds by Proposition~\ref{rule}.
\end{remark}

\vspace{.6cm}

\textbf{Notation}. Throughout this paper, the notation $f(x) = \mathcal{O}(g(x))$ means that there is a constant $C>0$ such that 
$|f(x)| \le C\,|g(x)|$ for all $x$ in a neighborhood of the point of interest (or for sufficiently large $x$). The implied constant in the notation  $\mathcal{O}$ depends only on $s$. If it depends on other parameters, these will be indicated explicitly in the subscript of $\mathcal{O}$, such as $\mathcal{O}_{s,a}$. The notation $x \downarrow 0$ means that $x >0$ and $x \to 0$. The Euclidean norm is adopted for \(2\times2\) matrices. Accordingly, for any matrix-valued function \(f:\Sigma \to \mathbb{C}^{2\times2}\) with \(f(\zeta) = (f_{ij}(\xi))_{i,j=1}^{2}\), on a reasonable contour $\Sigma\subset\mathbb{C}$, we define  
\begin{equation*}
	\|f(\cdot)\|_{L^2(\Sigma)}^2 := \sum_{i,j=1}^{2} \int_{\Sigma} |f_{ij}(\xi)|^2\,|\d\xi|,
	\qquad
	\|f(\cdot)\|_{L^{\infty}(\Sigma)} := \sup_{\xi\in \Sigma} \|f(\xi)\|.
\end{equation*}
Since all norms on \(\mathbb{C}^{n\times n}\) are equivalent, the results presented in this paper are independent of the particular choice of matrix norm.

\vspace{.6cm}

\textbf{Acknowledgements.}
This research was initiated while FW visited the School of Mathematics at the University of Bristol in June 2024. FW would like to thank the School for the enjoyable academic environment and warm hospitality during her visit. FW is supported by the Royal Society, grant URF$\backslash$R$\backslash$231028.

\section{Preliminary results} 

In this section, we prove some preliminary results that will be used in the subsequent analysis. By applying the Weyl integration formula to the average over $\mathbb{U}(N)$ with respect to the Haar measure, we have the following multi-dimensional integral representation of $F_{N}(s,h) $. 

\begin{prop} \cite[Proposition 3.4]{W}
We have that, for any $\textnormal{Re}(h)>-\frac{1}{2}$ and $\textnormal{Re}(s-h)>-\frac{1}{2}$,
\begin{equation}\label{i2}
	F_N(s,h)=\frac{2^{N^2+2sN-2h}}{(2\pi)^NN!}\int_{\mathbb{R}^N}\Bigg(\prod_{j=1}^N\frac{1}{(1+x_j^2)^{s+N}}\Bigg)\Bigg|\sum_{j=1}^Nx_j\Bigg|^{2h}\prod_{1\leq j<k\leq N}|x_k-x_j|^2\prod_{\ell=1}^N\d x_{\ell}.
\end{equation}
\end{prop}
Next, to deal with $|\sum_{j=1}^Nx_j|^{2h}$ in \eqref{i2}, we use the the Fourier integral $K_{h}^{\epsilon}(t)$ defined as in (\ref{i5}). By a residue computation, see also \cite[$3.382$]{GR}, for $\textnormal{Re}(\nu)>0$ and $\epsilon>0$, using principal branches,
\begin{equation*}
	\int_{-\infty}^{\infty}\frac{\e^{\im x\xi}}{(\epsilon+\im\xi)^{\nu}}\d\xi=\frac{2\pi}{\Gamma(\nu)}x^{\nu-1}\e^{-\epsilon x}\begin{cases}1,&x>0\\ 0,&x<0\end{cases};\hspace{0.5cm}\int_{-\infty}^{\infty}\frac{\e^{\im x\xi}}{(\epsilon-\im\xi)^{\nu}}\d\xi=\frac{2\pi}{\Gamma(\nu)}(-x)^{\nu-1}\e^{\epsilon x}\begin{cases}0,&x>0\\ 1,&x<0\end{cases},
\end{equation*}
which leads to the following inversion identity for \eqref{i5}.
\begin{lem} Suppose $\textnormal{Re}(h)>-1$ and $\epsilon>0$. Then for all $x\in\mathbb{R}\setminus\{0\}$,
\begin{equation}\label{i6}
	\int_{-\infty}^{\infty}K_h^{\epsilon}(\xi)\e^{\im\xi x}\d\xi=|x|^h\e^{-\epsilon|x|}.
\end{equation}
\end{lem}
\begin{proof} By simply adding the aforementioned integral identities.
\end{proof}

Moving forward, we use \eqref{i6} in the manipulation of \eqref{i2}:
\begin{prop}\label{1104inprop} We have that, for any $\textnormal{Re}(h)>0$ and $\textnormal{Re}(s-h)>-\frac{1}{2}$,
\begin{equation}\label{i7}
F_N(s,h)=\lim_{\epsilon\downarrow 0}\frac{2^{N^2+2sN-2h}}{(2\pi)^NN!}\int_{-\infty}^{\infty}K_{2h}^{\epsilon}(\xi)\Bigg[\int_{\mathbb{R}^N}\Bigg(\prod_{j=1}^N\frac{\e^{\im\xi x_j}}{(1+x_j^2)^{s+N}}\Bigg)\prod_{1\leq j<k\leq N}|x_k-x_j|^2\prod_{l=1}^N\d x_{l}\Bigg]\d\xi.
\end{equation}
\end{prop}
\begin{proof} By \eqref{i6}, for any $\epsilon>0$,
\begin{align}
	\int_{\mathbb{R}^N}&\Bigg(\prod_{j=1}^N\frac{1}{(1+x_j^2)^{s+N}}\Bigg)\Bigg|\sum_{j=1}^Nx_j\Bigg|^{2h}\exp\bigg(-\epsilon\bigg|\sum_{j=1}^Nx_j\bigg|\bigg)\prod_{1\leq j<k\leq N}|x_k-x_j|^2\prod_{l=1}^N\d x_{l}\nonumber\\
	=&\,\,\int_{\mathbb{R}^N}\int_{-\infty}^{\infty}K_{2h}^{\epsilon}(\xi)\Bigg(\prod_{j=1}^N\frac{\e^{\im\xi x_j}}{(1+x_j^2)^{s+N}}\Bigg)\prod_{1\leq j<k\leq N}|x_k-x_j|^2\,\d\xi\prod_{l=1}^N\d x_{l},\label{i8}
\end{align}
where, see also \cite[$(3.6),(3.12),(3.22)$]{W}, for any $\textnormal{Re}(h)>-1,\xi\in\mathbb{R},\epsilon>0$ and $x_j\in\mathbb{R}$,
\begin{equation*}
	\big|K_{h}^{\epsilon}(\xi)\big|\leq\frac{|\Gamma(1+h)|}{\pi(\epsilon^2+\xi^2)^{(\textnormal{Re}(h)+1)/2}},\ \ \ \prod_{1\leq j<k\leq N}|x_k-x_j|\leq N!\prod_{j=1}^N\big(1+x_j^2\big)^{(N-1)/2},
\end{equation*}
followed by, for any $x_j\in\mathbb{R}$ and $\textnormal{Re}(h)\geq 0$,
\begin{equation*}
	\Bigg|\sum_{j=1}^Nx_j\Bigg|^{\textnormal{Re}(h)}\leq N^{\textnormal{Re}(h)}\prod_{j=1}^N\big(1+x_j^2\big)^{\textnormal{Re}(h)/2}.
\end{equation*}
Consequently, Fubini's theorem allows us to swap the order of integration in the right hand side of \eqref{i8}, provided $\textnormal{Re}(h)>0$ and $\textnormal{Re}(s)>-\frac{1}{2}$,
\begin{align}\label{i9}
	\textnormal{LHS in}\ \eqref{i8}=\int_{-\infty}^{\infty}K_{2h}^{\epsilon}(\xi)\Bigg[\int_{\mathbb{R}^N}\Bigg(\prod_{j=1}^N\frac{\e^{\im\xi x_j}}{(1+x_j^2)^{s+N}}\Bigg)\prod_{1\leq j<k\leq N}|x_k-x_j|^2\prod_{l=1}^N\d x_{l}\Bigg]\d\xi.
\end{align}
On the other hand, for $\textnormal{Re}(h)\geq 0$ and any $x_j\in\mathbb{R}$,
\begin{align*}
	\prod_{j=1}^N\frac{1}{(1+x_j^2)^{\textnormal{Re}(s)+N}}\Bigg|\sum_{j=1}^Nx_j\Bigg|^{2\textnormal{Re}(h)}\prod_{1\leq j<k\leq N}|x_k-x_j|^2\leq(N!)^2N^{2\textnormal{Re}(h)}\prod_{j=1}^N\frac{1}{(1+x_j^2)^{\textnormal{Re}(s-h)+1}},
\end{align*}
and so the dominated convergence theorem says, provided $\textnormal{Re}(h)>0$ and $\textnormal{Re}(s-h)>-\frac{1}{2}$,
\begin{align}
	\int_{\mathbb{R}^N}&\Bigg(\prod_{j=1}^N\frac{1}{(1+x_j^2)^{s+N}}\Bigg)\Bigg|\sum_{j=1}^Nx_j\Bigg|^{2h}\prod_{1\leq j<k\leq N}|x_k-x_j|^2\prod_{l=1}^N\d x_{l}\nonumber\\
	=&\,\lim_{\epsilon\downarrow 0}\int_{\mathbb{R}^N}\Bigg(\prod_{j=1}^N\frac{1}{(1+x_j^2)^{s+N}}\Bigg)\Bigg|\sum_{j=1}^Nx_j\Bigg|^{2h}\exp\bigg(-\epsilon\bigg|\sum_{j=1}^Nx_j\bigg|\bigg)\prod_{1\leq j<k\leq N}|x_k-x_j|^2\prod_{l=1}^N\d x_{l}.\label{i10}
\end{align}
Combining \eqref{i10},\eqref{i8} and \eqref{i9} back in \eqref{i2} results in \eqref{i7}, for the indicated values of $(s,h)$.
\end{proof}

We note that our \eqref{i7} is a complex-valued generalization of \cite[$(3.9)$]{W}, based on \eqref{i5} rather than \cite[$(3.8)$]{W}. Moving ahead, the innermost integrals over $(x_1,\ldots,x_N)$ in \eqref{i7} constitute a Hankel determinant.
Since we consider, for now, general complex parameters \(s\), we no longer adopt the approach in \cite[Proposition~4.2]{W} to obtain a Hankel determinant of size \(s \times s\) as in \cite[Proposition~4.3]{W}, which requires \(s\) to be a positive integer. Instead, motivated by the work of Assiotis, Gunnes, Keating, and Wei~\cite[Proposition 4.2]{AGKW}, which connects the joint moments of higher derivatives of CUE characteristic polynomials to Hankel determinants shifted by integer partitions for general moment exponents \(s\), we employ the Fourier transform of a Cauchy-type weight to obtain a Hankel determinant as in (\ref{NtimesNHankel}) of size \(N \times N\), with entries involving the confluent hypergeometric function of the second kind depending on the parameters \(s\). 
For notational convenience, we introduce the following 
abbreviation. 
\begin{equation}\label{defofJ}
J_N(\xi,s):=\int_{\mathbb{R}^N}\Bigg(\prod_{j=1}^N\frac{\e^{\im\xi x_j}}{(1+x_j^2)^{s+N}}\Bigg)\prod_{1\leq j<k\leq N}|x_k-x_j|^2\prod_{\ell=1}^N\d x_{\ell},\ \ \ \ \textnormal{Re}(s)>-\frac{1}{2},\ \ \xi\in\mathbb{R},
\end{equation}
and realize the following special value of $J_N$ as a Selberg integral, for $\textnormal{Re}(s)>-\frac{1}{2}$,
\begin{equation}\label{Selb}
	J_N(0,s)=\frac{(2\pi)^NN!}{2^{N^2+2sN}}\prod_{j=1}^N\frac{\Gamma(j)\Gamma(2s+j)}{\Gamma^2(s+j)}=\frac{(2\pi)^NN!}{2^{N^2+2sN}}\frac{G(1+N)G(1+N+2s)G^2(1+s)}{G(1+2s)G^2(1+N+s)},
\end{equation}
with $G(z)$ the Barnes $G$-function.
The sought-after Hankel determinant formula reads as follows:
\begin{lem} For any $\textnormal{Re}(s)>-\frac{1}{2}$ and $\xi\in\mathbb{R}\setminus\{0\}$,
\begin{equation}\label{i11}
	J_N(\xi,s)=(-1)^{N(N-1)/2}N!\,\e^{N|\xi|}\det\big[\phi^{(j+k)}(\xi,s)\big]_{j,k=0}^{N-1},\ \ \ \ \ \ \phi(\xi,s):=\int_{-\infty}^{\infty}\frac{\e^{-\im\xi(x-\im\delta)}}{(1+x^2)^{s+N}}\d x,
\end{equation}
using the shorthand $\phi^{(k)}=\frac{\partial^k}{\partial\xi^k}\phi$ with $k\in\mathbb{Z}_{\geq 0}$ and $\delta=\textnormal{sgn}(\xi)\in\{\pm 1\}$.
\end{lem}
\begin{proof} Observe that
\begin{equation*}
	\bigg|\frac{\partial^k}{\partial\xi^k}\frac{\e^{-\im\xi(x-\im\delta)}}{(1+x^2)^{s+N}}\bigg|\leq |x|^{k-2\textnormal{Re}(s)-2N}\ \ \ \forall\,|x|\geq 1,\ \ \ k\in\mathbb{Z}_{\geq 0},\ \ \ \xi\in\mathbb{R}\setminus\{0\},
\end{equation*}
and so, by the dominated convergence theorem, with $\textnormal{Re}(s)>-\frac{1}{2}$ and $\xi\in\mathbb{R}\setminus\{0\}$ arbitrary,
\begin{equation*}
	\frac{\partial^k}{\partial\xi^k}\phi(\xi,s)=(-\im)^k\int_{-\infty}^{\infty}(x-\im\delta)^k\frac{\e^{-\im\xi(x-\im\delta)}}{(1+x^2)^{s+N}}\d x\ \ \ \ k\in\{0,1,2,\ldots,2N-2\}.
\end{equation*}
Consequently, multilinearity of determinants asserts
\begin{align*}
	\det\big[\phi^{(j+k)}(\xi,s)\big]_{j,k=0}^{N-1}=(-1)^{N(N-1)/2}\int_{\mathbb{R}^N}\prod_{j=1}^N(x_j-\im\delta)^{j-1}\det\big[(x_j-\im\delta)^{k-1}\big]_{j,k=1}^N\prod_{\ell=1}^N\frac{\e^{-\im\xi(x_{\ell}-\im\delta)}\d x_{\ell}}{(1+x_{\ell}^2)^{s+N}},
\end{align*}
and thus, after symmetrization,
\begin{align*}
	N!\det\big[\phi^{(j+k)}(\xi,s)\big]_{j,k=0}^{N-1}=&\,\,(-1)^{N(N-1)/2}\int_{\mathbb{R}^N}\Big(\det\big[(x_j-\im\delta)^{k-1}\big]_{j,k=1}^N\Big)^2\prod_{l=1}^N\frac{\e^{-\im\xi(x_{l}-\im\delta)}\d x_{l}}{(1+x_{l}^2)^{s+N}}\\
	=&\,\,(-1)^{N(N-1)/2}\e^{-N|\xi|}\int_{\mathbb{R}^N}\Bigg(\prod_{j=1}^N\frac{\e^{\im\xi x_j}}{(1+x_j^2)^{s+N}}\Bigg)\prod_{1\leq j<k\leq N}|x_k-x_j|^2\prod_{\ell=1}^N\d x_{\ell},
\end{align*}
as claimed in \eqref{i11}. The proof is complete.
\end{proof}
In order to simplify the Hankel determinant in \eqref{i11}, we utilize the following special case of \cite[$3.384$]{GR},
\begin{equation*}
	\frac{1}{\pi}\int_{-\infty}^{\infty}\frac{\e^{-2\im\xi x}}{(1+x^2)^{s+N}}\d x=\frac{|\xi|^{s+N-1}}{\Gamma(s+N)}W_{0,\frac{1}{2}-s-N}(4|\xi|),\ \ \xi\in\mathbb{R}\setminus\{0\},\ \ \textnormal{Re}(s)>-\frac{1}{2},\ \ N\in\mathbb{Z}_{\geq 1},
\end{equation*}
that invokes the Whittaker function $W_{\kappa,\mu}(z)$, cf. \cite[$13.14.3$]{NIST}, and thus naturally $U(a,b,z)$,
\begin{equation}\label{i12}
	\phi(\xi,s)=\frac{\pi}{\Gamma(s+N)}2^{2(1-s-N)}\e^{-2|\xi|}U(1-s-N,2-2s-2N,2|\xi|),\ \ \xi\in\mathbb{R}\setminus\{0\}.
\end{equation}
In turn, by \cite[$13.3.27$]{NIST},
\begin{equation*}
	\phi^{(k)}(\xi,s)=\frac{\pi}{\Gamma(s+N)}2^{2(1-s-N)}(-2\delta)^k\e^{-2|\xi|}U(1-s-N,2-2s-2N+k,2|\xi|),\ \ \ k\in\mathbb{Z}_{\geq 0},
\end{equation*}
and we have thus arrived at the following intermediate result.

\begin{prop}\label{expressionofJ} For any $\textnormal{Re}(s)>-\frac{1}{2}$ and $\xi\in\mathbb{R}\setminus\{0\}$,
\begin{equation}\label{i13}
	J_N(\xi,s)=(-1)^{N(N-1)/2}\bigg(\frac{2\pi\e^{-|\xi|}}{\Gamma(s+N)}\bigg)^N2^{-N^2-2sN}N!\,\mathcal{J}_N(|\xi|;s),\end{equation}
with the Hankel determinant $\mathcal{J}_N(t;s)$ defined as in \eqref{NtimesNHankel}. 
\end{prop}

We now give the proof of Proposition~\ref{1102prop1}.
\begin{proof}[Proof of Proposition~\ref{1102prop1}]
The result follows directly from Proposition~\ref{1104inprop}, (\ref{defofJ}), and Proposition~\ref{expressionofJ}, utilising evenness of $\xi\mapsto J_n(\xi,s)$ and $\xi\mapsto K_h^{\epsilon}(\xi)$.
\end{proof}

We now perform a transformation on the Hankel determinant \( \mathcal{J}_{N}(\xi;s) \) in (\ref{i13}) to obtain an alternative form of the integral representation of \( J_{N}(\xi,s) \) as in (\ref{e43}). This serves as the starting point for the Riemann–Hilbert method used in this paper.

\begin{lem}\label{transformation}
Let $N\geq 1$ be an integer. Let $\mu,\nu\in \mathbb{C}$ with $\textnormal{Re}(2\mu)>N-1$. Let $U(a,b,z)$ be the confluent hypergeometric function of the second kind, cf. \cite[$\S 13.2.6$]{NIST}. Then for all $z\in\mathbb{C}$ with $\textnormal{Re}(z)>0$,
\begin{align*}
\det\Big[U(1-2\nu,2-2\nu-2\mu+i+j,z)\Big]_{i,j=0}^{N-1}
=&\left(\prod_{\ell=0}^{N-2}(1-2\nu+\ell)^{N-1-\ell}\right)z^{N(2\mu+2\nu-N)}\left(\prod_{j=0}^{N-1}\frac{1}{\Gamma(2\mu-j)}\right)\\
&\times\,\det\left[\int_{0}^{\infty}\mathrm{e}^{-zx}x^{2\mu-j-1}(1+x)^{2\nu-i-1}dx\right]_{i,j=0}^{N-1}.
\end{align*}
\end{lem}
\begin{proof}
Note that by \cite[$\S 13.3.9$]{NIST},
\[
U(a,b,z) = a U(a+1,b,z) + U(a,b-1,z).
\]
So, exploiting multilinearity of determinants and the above recursion, we obtain
\beas
\det\Big[U(a,b+i+j,z)\Big]_{i,j=0}^{N-1}&=&\left(\prod_{\ell=0}^{N-2}(a+\ell)^{N-1-\ell}\right)\det\Big[U(a+j,b+i+j,z)\Big]_{i,j=0}^{N-1}\\
&=&\left(\prod_{j=0}^{N-2}(a+\ell)^{N-1-\ell}\right)\det\Big[U(a+i,b+i+j,z)\Big]_{i,j=0}^{N-1}.
\eeas
Next, with $U(a,b,z) = z^{1-b} U(1+a-b,2-b,z)$, see \cite[$\S 13.2.40$]{NIST}, we obtain furthermore
\begin{align*}
\det\Big[U(a,b+i+j,z)\Big]_{i,j=0}^{N-1}=&\left(\prod_{\ell=0}^{N-2}(a+\ell)^{N-1-\ell}\right)z^{N(1-b)-N(N-1)}\nonumber\\
&\times \, \det\Big[U(1+a-b-j,2-b-i-j,z)\Big]_{i,j=0}^{N-1}.
\end{align*}
Note that for $c$ with $\textnormal{Re}(c)>0$ and any $\textnormal{Re}(z)>0$,
\[
U(c,d,z) = \frac{1}{\Gamma(c)} \int_0^\infty \mathrm{e}^{-zx} x^{c-1} (1+x)^{d-c-1} dx,
\]
and so the Lemma follows.
\end{proof}
Applying \eqref{i13}, Lemma \ref{transformation} and the Andr\'{e}ief identity, we have the following representation of $J_{N}(\xi,s)$.

\begin{prop}\label{integralrepresentation}
Let $N\geq 1$ be an integer. Let $s\in \mathbb{C}$ with $\textnormal{Re}(s)>-1/2$. Let $\xi> 0$. Then
\begin{equation}\label{e43}
J_N(\xi,s)=\frac{(2\pi)^N\e^{-N\xi}}{2^{N^2+2sN}}\prod_{j=1}^N\frac{1}{\Gamma^2(s+j)}\int_{\mathbb{R}_+^N}\Bigg(\prod_{j=1}^N(y_j+2\xi)^sy_j^s\,\e^{-y_j}\Bigg)\prod_{1\leq j<k\leq N}|y_k-y_j|^2\prod_{\ell=1}^N\d y_{\ell}.
\end{equation}
\end{prop}

Equation~(\ref{e43}) first appeared in \cite[Proposition 3]{W}, where it was derived by a different approach using the homogeneity of \( \Delta(y_{1},\ldots,y_{N})\), the Vandermonde determinant. Here we present an alternative derivation to illustrate how the general form 
\bea\label{withgeneralparameters}
\det\!\Big[U(1-2\nu,\,2-2\nu-2\mu+i+j,\,z)\Big]_{i,j=0}^{N-1}
\eea
can be rewritten as an integral of the type \eqref{e43}, which is of independent interest. Moreover, this determinant has been actively used in recent research on joint moments arising from various random matrix ensembles. In \cite{BW2}, the determinant (\ref{withgeneralparameters}) with $\mu=(N+s)/2$ and $\nu=(N+\bar{s})/2$ ($s\in \mathbb{C}$) is used in the study of joint moments of characteristic polynomials for the Circular Jacobi Ensemble and in the analysis of the associated Painlev\'e equation. In \cite{assiotis2025joint}, the determinant (\ref{withgeneralparameters}) with $\mu=(a+b)/2+N$ and $\nu=-b/2$ (for some $a,b\in \mathbb{R}$) was used in the study of joint moments of characteristic polynomials from the orthogonal and unitary symplectic groups.

Returning to the joint moments for the CUE studied in the current paper with $\mu=\nu=(N+s)/2$, $s\in \mathbb{R}$ in \eqref{withgeneralparameters}, we note the following. Since it was shown in \cite[Propositions~4.2 and~4.3]{W} that the integral on the right-hand side of (\ref{e43}) can be expressed as a Hankel determinant involving Laguerre polynomials when $s \in \mathbb{Z}_{\geq 1}$, we make the below observation.

\begin{prop}\label{relationtoaHankel}
Let $N,s\geq 1$ be integers. Let $\mathcal{J}_{N}(t;s)$ be as in \eqref{NtimesNHankel} with $t>0$. Then
\beas
\frac{(-1)^{N(N-1)/2}}{(\Gamma(s+N))^N}\mathcal{J}_{N}(t;s)=(-1)^{s(s-1)/2}\det\Big[L_{N+s-1-(j+k)}^{(2s-1)}(-2t)\Big]_{j,k=0}^{s-1};\  L_n^{(\alpha)}(x):=\frac{\e^x}{n!}x^{-\alpha}\frac{\d^n}{\d x^n}\Big(\e^{-x}x^{n+\alpha}\Big).
\eeas
\end{prop}

At the end of this section, we note that when \( s \in \mathbb{R} \), the Hankel determinant \( \mathcal{J}_{N}(t;s) \) in \eqref{NtimesNHankel} is known to be related to \(\sigma\)-Painlev\'e~V, see (\ref{1102defu}). This can be verified, for instance, from \cite[Proposition~3.2]{forrester2002application}, or alternatively from (\ref{e43}) together with \cite{chen2012coulumb}.

%
%
%
%

\section{Large $N$-asymptotics for $J_N(\xi,s)$}\label{sectiononlargeN}

As explained in the introduction, the poles of the Painlev\'e function $u_N$ in \eqref{1102defu} arise from the zeros of \( J_{N}(\xi,s) \) when the exponent \( s \) is complex. Therefore, the present work focuses on the case \( s \in \mathbb{R} \). Our starting point is the integral representation of \( J_N(\xi,s) \) for \( \xi > 0 \) given in (\ref{e43}). Denote 
\begin{equation}\label{e44}
W_N(\xi,s):=\int_{\mathbb{R}_+^N}\Bigg(\prod_{j=1}^N(y_j+\xi)^sy_j^s\,\e^{-y_j}\Bigg)\prod_{1\leq j<k\leq N}|y_k-y_j|^2\prod_{\ell=1}^N\d y_{\ell}.
\end{equation}
Evidently, $W_N$ constitutes a Hankel moment determinant by Andr\'eief's identity,
\begin{equation*}
	W_N(\xi,s)=N!\,\det\bigg[\int_0^{\infty}y^{j+k}\omega(y;\xi,s)\d y\bigg]_{j,k=0}^{N-1},\ \ \ \omega(y;\xi,s):=(y+\xi)^sy^s\e^{-y}, \ y>0,
\end{equation*}
which is well-defined for $\xi>0$ and $s>-1/2$. It is straightforward to verify from the integral representation on the right-hand side of (\ref{e43}) that, for any \( s > -1/2 \), \( J_N(\xi,s) \) is strictly positive for \( \xi > 0 \). So we have the following conclusion. 

\begin{lem}\label{Mor2} The function
\begin{equation*}
	z\mapsto W_N(z,s)=N!\,\det\bigg[\int_0^{\infty}y^{j+k}\omega(y;z,s)\d y\bigg]_{j,k=0}^{N-1}
\end{equation*}
is smooth for $z>0$ for any $N\in\mathbb{N},s>-\frac{1}{2}$, and does not vanish on $(0,\infty)\subset\mathbb{R}$.
\end{lem}
By Propositions \ref{expressionofJ} and \ref{integralrepresentation}, combining with (\ref{1102defu}), 
\begin{equation*}
u_N(\xi;s)=\xi\frac{\partial}{\partial\xi}\ln W_N(\xi,s),\ \ \ \xi>0.
\end{equation*} 
As $(0,\infty)\ni \xi\mapsto W_{N}(\xi,s)$ is zero-free for any $N\in\mathbb{N}$ and $2s>-1$, the function $(0,\infty)\ni z\mapsto u_N(\frac{z}{N};s)$ is pole free and we can carry out a bona-fide large $N$-asymptotic analysis of the same quantity. To evaluate $W_N$, and relate it to a RHP, we use
%
%
%
%
%
\begin{equation}\label{e46}
	\mathbb{C}[y]\ni p_j(y)=\chi_j\Big(y^j+\gamma_jy^{j-1}+\mathcal{O}\big(y^{j-2}\big)\Big),\ \ \ \ \chi_j=\chi_j(z,s)>0;\ \ \ \ \gamma_j=\gamma_j(z,s)\in\mathbb{R}
\end{equation}
to denote the degree $j\in\mathbb{Z}_{\geq 0}$ orthonormal polynomial with respect to $\omega(y;z,s)$ on $(0,\infty)$. That is we have
\begin{equation*}
	\int_0^{\infty}p_j(y)p_k(y)\omega(y;z,s)\d y=\delta_{jk},\ \ \ \ 0\leq k\leq j,
\end{equation*}
and $p_j\in\mathbb{C}[y]$ exists for any $z>0$ and $s>-\frac{1}{2}$ by Lemma \ref{Mor2}. Moreover,
\begin{equation}\label{e47}
	W_N(z,s)=N!\,\prod_{j=0}^{N-1}\chi_j^{-2}(z,s),\ \ \ \ z>0;\hspace{1.5cm}\frac{W_{N+1}(z,s)}{W_N(z,s)}=(N+1)\chi_N^{-2}(z,s),\ \ \ \ z>0,
\end{equation}
and $p_N$ admits the following RHP characterisation.
\begin{problem}\label{RHP3} Let $N\in\mathbb{N},s>-\frac{1}{2}$ and $z>0$. Now determine $X(\lambda)=X(\lambda;z,s,N)\in\mathbb{C}^{2\times 2}$ so that
\begin{enumerate}
	\item[(1)] $X(\lambda)$ is analytic for $\lambda\in\mathbb{C}\setminus[0,\infty)$ and extends continuously to $\mathbb{C}\setminus\{0\}$ from either side of $(0,\infty)$.
	\item[(2)] The continuous limiting values $X_{\pm}(\lambda)=\lim_{\epsilon\downarrow 0}X(\lambda\pm\im\epsilon),\lambda\in(0,\infty)$ satisfy
	\begin{equation*}
		X_+(\lambda)=X_-(\lambda)\begin{bmatrix}1&\omega(\lambda;z,s)\\ 0&1\end{bmatrix}.
	\end{equation*}
	\item[(3)] Near $\lambda=0$, $X(\lambda)$ is weakly singular in the following entry-wise sense,
	\begin{equation*}
		X(\lambda)=\mathcal{O}\begin{bmatrix}1 & |\lambda|^{s}\\ 1 & |\lambda|^{s}\\ \end{bmatrix},\ -\frac{1}{2}<s<0;\ \ \ X(\lambda)=\mathcal{O}\begin{bmatrix}1 & \ln|\lambda| \\ 1 & \ln|\lambda|\end{bmatrix},\ s=0;\ \ \ X(\lambda)=\mathcal{O}\begin{bmatrix}1&1\\ 1&1\end{bmatrix},\ s>0.
	\end{equation*}
	\item[(4)] As $\lambda\rightarrow\infty$ and $\lambda\notin(0,\infty)$,
	\begin{equation*}
		X(\lambda)=\bigg\{I+\sum_{\ell=1}^2\frac{o_{\ell}}{\lambda^{\ell}}+\mathcal{O}\big(\lambda^{-3}\big)\bigg\}\lambda^{N\sigma_3},\ \ \ o_{\ell}=o_{\ell}(z,s,N)=\Big[o_{\ell}^{jk}(z,s,N)\Big]_{j,k=1}^2,\ \ \ I=\big[\delta_{jk}\big]_{j,k=1}^2.
	\end{equation*}
\end{enumerate}
\end{problem}
By \cite{FIK2}, RHP \ref{RHP3} is uniquely solvable for $N\in\mathbb{N},s>-\frac{1}{2}$ and $z>0$ with its solution being of the form
\begin{equation*}
	X(\lambda)=\begin{bmatrix}\chi_N^{-1}p_N(\lambda) & \frac{\chi_N^{-1}}{2\pi\im}\int_0^{\infty}p_N(y)\omega(y;z,s)\frac{\d y}{y-\lambda}\smallskip\\ -2\pi\im \chi_{N-1}p_{N-1}(\lambda) & -\chi_{N-1}\int_0^{\infty}p_{N-1}(y)\omega(y;z,s)\frac{\d y}{y-\lambda}\end{bmatrix},\ \ \ \lambda\in\mathbb{C}\setminus[0,\infty),
\end{equation*}
expressed in terms of the polynomials in \eqref{e46} and thus relating back to $W_N$ via the identity
\begin{equation}\label{e48}
	\frac{W_{N+1}(z,s)}{W_N(z,s)}=(N+1)\chi_N^{-2}(z,s)=-2\pi\im(N+1)\,o_1^{12}(z,s,N).
\end{equation}
In order to study the large $N$-asymptotic of the previous $v_N$, we first combine \eqref{e44},\eqref{e47} and \eqref{e48}:
\begin{lem} Let $N\in\mathbb{N},s>-\frac{1}{2}$ and $z>0$. Then
\begin{equation}\label{e49}
	u_N(z;s)=N(N+2s)+o_1^{11}(z,s,N),
\end{equation}
with $o_1(z,s,N)$ as in condition $(4)$ of RHP \ref{RHP3} and $u_N(z;s)$ as in \eqref{1102defu}.
\end{lem}
\begin{proof} For $\ell\in\mathbb{N},s>-\frac{1}{2}$ and $z>0$, we look at
\begin{equation}\label{e50}
	\Phi(\lambda)\equiv \Phi(\lambda;z,s,\ell):=X(\lambda;z,s,\ell)(\lambda+z)^{\frac{s}{2}\sigma_3}\lambda^{\frac{s}{2}\sigma_3}\e^{-\frac{\lambda}{2}\sigma_3},\ \ \ \lambda\in\mathbb{C}\setminus\mathbb{R}
\end{equation}
with principal branches chosen for $\mathbb{C}\setminus(-\infty,-z]\ni\lambda\mapsto(\lambda+z)^{\alpha}$ as well as $\mathbb{C}\setminus(-\infty,0]\ni\lambda\mapsto\lambda^{\alpha}$.
\begin{figure}[tbh]
	\begin{tikzpicture}[xscale=0.9,yscale=0.9]
        \draw [->] (0,0) -- (3.5,0) node [right] {$\footnotesize{\textnormal{Re}(\lambda)}$};
        \draw [thick,red,decoration={markings,mark= at position 0.5 with {\arrow{>}}} ,postaction={decorate}] (0,0) -- (3,0);
  	\draw [->] (0,-\yr) -- (0,1.5) node [above] {$\footnotesize{\textnormal{Im}(\lambda)}$};
	\draw [thick, color=red, decorate,decoration={zigzag,segment length=4,amplitude=1,post=lineto,post length=0}] (-3.5,0) -- (0,0);
  	\draw [fill, color=blue!60!black] (-2,0) circle [radius=0.04];
  	\draw [fill, color=blue!60!black] (0,0) circle [radius=0.04];

\end{tikzpicture}
\caption{The oriented jump contour $(-\infty,-z)\cup(-z,0)\cup(0,\infty)$ used in the RHP for $\Phi(\lambda)$, in the complex $\lambda$-plane drawn in \textcolor{red}{red}. The (end)points $\lambda=0,-z$ are colored in \textcolor{blue!60!black}{blue}, for one particular choice $z>0$, and the branch cut is oriented to the right.}
\label{fig5}
\end{figure}
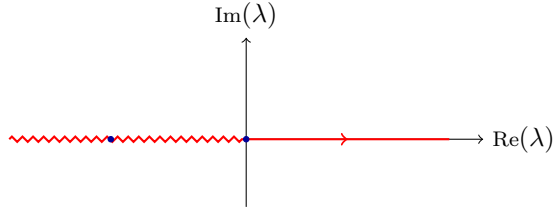

What results from RHP \ref{RHP3} is a RHP for $\Phi(\lambda)$ with $(\lambda,z)$-independent jumps, and consequently the following Lax system, 
\begin{equation}\label{zz1}
	\frac{\partial\Phi}{\partial\lambda}(\lambda)=\bigg\{C_{\infty}+\frac{C_0}{\lambda}+\frac{C_1}{\lambda+z}\bigg\}\Phi(\lambda),\hspace{2cm}  \frac{\partial\Phi}{\partial z}(\lambda)=\bigg\{\frac{D_1}{\lambda+z}\bigg\}\Phi(\lambda),
\end{equation}
with $\lambda$-independent coefficients
\begin{equation*}
	C_{\infty}:=-\frac{1}{2}\sigma_3,\ \ \ \ \ C_1=D_1:=\frac{s}{2}\widehat{\Phi}_{z}(-z)\sigma_3\widehat{\Phi}_{z}(-z)^{-1},\ \ \ \ \ C_0:=\begin{cases}\frac{s}{2}\widehat{\Phi}_0(0)\sigma_3\widehat{\Phi}_0(0)^{-1},&s\neq 0\bigskip\\
	\frac{\im}{2\pi}\widehat{\Phi}_0(0)\bigl[\begin{smallmatrix}0&1\\ 0&0\end{smallmatrix}\bigr]\widehat{\Phi}_0(0)^{-1},&s=0
	\end{cases},
\end{equation*}
that are expressed in terms of locally analytic $\widehat{\Phi}_0(\lambda)$, resp. $\widehat{\Phi}_z(\lambda)$, near $\lambda=0$, resp. $\lambda=-z$. In addition, the coefficients satisfy the constraint
\begin{equation}\label{zz3}
	C_0+C_1=(\ell+s)\sigma_3-\frac{1}{2}\big[o_1(z,s,\ell),\sigma_3\big],
\end{equation}
What's more, the RHP for $\Phi_{\ell}(\lambda)\equiv\Phi(\lambda;z,s,\ell)$ has jumps that are also $\ell$-independent, and so
\begin{equation}\label{zz2}
	\Phi_{\ell+1}(\lambda)=\Big\{\lambda U_{\infty}+E_0^{(\ell)}\Big\}\Phi_{\ell}(\lambda),
\end{equation}
with coefficients
\begin{equation*}
	U_{\infty}=\begin{bmatrix}1&0\\ 0&0\end{bmatrix},\ \ \ E_0^{(\ell)}:=\begin{bmatrix}o_1^{11}(z,s,\ell+1)-o_1^{11}(z,s,\ell) & -o_1^{12}(z,s,\ell)\smallskip\\ o_1^{21}(z,s,\ell+1) & 0\end{bmatrix}.
\end{equation*}
Define the transformation $\Upsilon(\lambda)=\Upsilon(\lambda;z,s,\ell):=(-z)^{-(s+\ell)\sigma_3}\Phi(-\lambda z;z,s,\ell)$, then 
\begin{equation*}
	\frac{\partial\Upsilon}{\partial\lambda}(\lambda)=\bigg\{zA_{\infty}+\frac{\widehat{A}_0}{\lambda}+\frac{\widehat{A}_1}{\lambda-1}\bigg\}\Upsilon(\lambda),\hspace{0.5cm}\frac{\partial\Upsilon}{\partial z}(\lambda)=\big\{\lambda A_{\infty}+\widehat{B}_0\big\}\Upsilon(\lambda),\hspace{0.5cm}\Upsilon_{\ell+1}(\lambda)=\Big\{\lambda U_{\infty}+\widehat{U}_0^{(\ell)}\Big\}\Upsilon_{\ell}(\lambda).
\end{equation*}
Here with the shorthands
\begin{equation*}
	A_{\infty}=-C_{\infty}=\frac{1}{2}\sigma_3,\ \ \ \ \ \ \widehat{A}_0:=(-z)^{-(s+\ell)\sigma_3}C_0(-z)^{(s+\ell)\sigma_3},\ \ \ \ \ \ \widehat{A}_1:=(-z)^{-(s+\ell)\sigma_3}C_1(-z)^{(s+\ell)\sigma_3},
\end{equation*}
\begin{equation*}
	\widehat{B}_0:=\frac{1}{z}\big(\widehat{A}_0+\widehat{A}_1-(\ell+s)\sigma_3\big)\stackrel{\eqref{zz3}}{=}-\frac{1}{2z}(-z)^{-(s+\ell)\sigma_3}\big[o_1(z,s,\ell),\sigma_3\big](-z)^{(s+\ell)\sigma_3},
\end{equation*}
for the $(\lambda,z)$-systems, where $\widehat{B}_0$ is traceless, and
\begin{equation*}
	\widehat{U}_0^{(\ell)}:=(-z)^{-\sigma_3}(-z)^{-(s+\ell)\sigma_3}E_0^{(\ell)}(-z)^{(s+\ell)\sigma_3},
\end{equation*}
for the difference system. Considering both, the $z$-system and the difference system for $\Upsilon$, compatibility says, in particular,
\begin{equation}\label{e51}
	\frac{\partial}{\partial z}\bigg[\frac{o_1^{12}(z,s,\ell)}{-z}(-z)^{-2(s+\ell)}\bigg]=-\bigg[\frac{o_1^{12}(z,s,\ell)}{-z}(-z)^{-2(s+\ell)}\bigg]\bigg[\frac{o_1^{11}(z,s,\ell+1)-o_1^{11}(z,s,\ell)}{-z}\bigg],
\end{equation}
and consequently
\begin{align}
	\frac{\partial}{\partial z}\ln W_N(z,s)\stackrel{\eqref{e47}}{=}\frac{\partial}{\partial z}\ln&\, W_1(z,s)+\sum_{\ell=1}^{N-1}\frac{\partial}{\partial z}\ln\chi_{\ell}^{-2}(z,s)\stackrel{\eqref{e48}}{=}\frac{\partial}{\partial z}\ln W_1(z,s)+\sum_{\ell=1}^{N-1}\frac{\partial}{\partial z}\ln o_1^{12}(z,s,\ell)\nonumber\\
	&\stackrel{\eqref{e51}}{=}\frac{\partial}{\partial z}\ln W_1(z,s)+\frac{1}{z}(N-1)(N+2s+1)+\frac{1}{z}\big(o_1^{11}(z,s,N)-o_1^{11}(z,s,1)\big).\label{e52}
\end{align}
It thus remains to relate $\frac{\partial}{\partial z}\ln W_1(z,s)$ and $o_1^{11}(z,s,1)$. First,
\begin{equation*}
	W_1(z,s)=\chi_0^{-2}(z,s)=\int_0^{\infty}\omega(y;z,s)\d y=z^{1+2s}\int_0^{\infty}(y+1)^sy^s\e^{-yz}\d y,\ \ \ z>0,
\end{equation*}
and so
\begin{align}
	\frac{\partial}{\partial z}\ln W_1(z,s)=&\,\chi_0^2(z,s)\bigg[\frac{1+2s}{z}\chi_0^{-2}(z,s)-z^{1+2s}\int_0^{\infty}(y+1)^sy^{s+1}\e^{-yz}\d y\bigg]\nonumber\\
	=&\,\frac{1+2s}{z}-\frac{1}{z}\chi_0^2(z,s)\int_0^{\infty}\omega(y;z,s)y\,\d y=\frac{1+2s}{z}+\frac{1}{z}o_1^{11}(z,s,1),\label{e53}
\end{align}
by utilizing
\begin{equation*}
	-o_1^{11}(z,s,1)=o_1^{22}(z,s,1)=\chi_0^2(z,s)\int_0^{\infty}\omega(y;z,s)y\,\d y.
\end{equation*}
Substituting \eqref{e53} into \eqref{e52} yields \eqref{e49} and thus completes our proof by \eqref{e44} and Lemma \ref{Mor2}.
\end{proof}

\begin{remark} So far our workings extend to $s\in\mathbb{C}:\textnormal{Re}(s)>-1/2$ and admissible $\textnormal{Re}(z)>0$, at the expense of using discrete exceptional sets. The very subtle impact of $s$ being complex-valued manifests itself in the upcoming asymptotic analysis, specifically in the context of the local analysis of RHP \ref{RHP3} near $\lambda=0$.
\end{remark}

Equipped with \eqref{e49} we now commence the large $N$-asymptotic analysis of RHP \ref{RHP3}. The necessary steps are standard, cf. \cite{DKMVZ,V}, aside from the local analysis near $\lambda=0$. To begin with we introduce the $g$-function
\begin{equation*}
	g(\lambda):=\frac{2}{\pi}\int_0^1\ln(\lambda-x)\sqrt{\frac{1-x}{x}}\d x,\ \ \ \lambda\in\mathbb{C}\setminus(-\infty,1],
\end{equation*}
using the principal branch of the logarithm with $\textnormal{arg}(\lambda-x)\in(-\pi,\pi)$. Evidently
\begin{equation*}
	g_+(\lambda)-g_-(\lambda)=\begin{cases}\displaystyle0,&\lambda>1\\ \displaystyle4\im\int_{\lambda}^1\sqrt{\frac{1-x}{x}}\d x,&\lambda\in(0,1)\\ \displaystyle2\pi\im,&\lambda<0\end{cases};\ \ \ \ \ \ \ \ g(\lambda)=\ln\lambda-\frac{1}{4\lambda}+\mathcal{O}\big(\lambda^{-2}\big),\ \ \ \lambda\rightarrow\infty,
\end{equation*}
and $g_+(\lambda)+g_-(\lambda)-4\lambda+\ell=0$ for $\lambda\in(0,1)$ followed by
\begin{equation*}
	g_+(\lambda)+g_-(\lambda)-4\lambda+\ell=-4\int_1^{\lambda}\sqrt{\frac{x-1}{x}}\d x<0,\ \ \,\lambda\in(1,\infty),
\end{equation*}
where $\ell:=2+4\ln 2$ denotes the relevant Lagrange multiplier.
In turn, the transformation
\begin{equation}\label{e54}
	S(\lambda;z,s,N):=(4N)^{-(N+s)\sigma_3}\e^{\frac{N}{2}\ell\sigma_3}X\Big(4\lambda N;\frac{z}{N},s,N\Big)\e^{-N(g(\lambda)+\frac{\ell}{2})\sigma_3}(4N)^{s\sigma_3},\ \ \ \ \lambda\in\mathbb{C}\setminus[0,\infty),
\end{equation}
leads us to the below, normalised at $\lambda=\infty$, RHP:
\begin{problem}\label{RHP4} Let $N\in\mathbb{N},s>-\frac{1}{2}$ and $z>0$. The function $S(\lambda)=S(\lambda;z,s,N)\in\mathbb{C}^{2\times 2}$ defined in \eqref{e54} has the following properties:
\begin{enumerate}
	\item[(1)] $S(\lambda)$ is analytic for $\lambda\in\mathbb{C}\setminus[0,\infty)$ and extends continuously to $\mathbb{C}\setminus\{0\}$ from either side of $(0,\infty)$.
	\item[(2)] The continuous limiting values $S_{\pm}(\lambda)=\lim_{\epsilon\downarrow 0}S(\lambda\pm\im\epsilon),\lambda\in(0,\infty)$ satisfy
	\begin{equation*}
		S_+(\lambda)=S_-(\lambda)\begin{bmatrix}\e^{N\pi(\lambda)} & (\lambda+z_N)^s\lambda^s\smallskip\\ 0 & \e^{-N\pi(\lambda)}\end{bmatrix},\ \ \lambda\in(0,1);
	\end{equation*}
	with $(0,1)\ni\lambda\mapsto\pi(\lambda):=g_-(\lambda)-g_+(\lambda)=-4\im\int_{\lambda}^1\sqrt{\frac{1-x}{x}}\d x\in\im\mathbb{R}$ and
	\begin{equation*}
		S_+(\lambda)=S_-(\lambda)\begin{bmatrix}1 & (\lambda+z_N)^s\lambda^s\e^{N\eta(\lambda)}\smallskip\\ 0 & 1\end{bmatrix},\ \ \lambda>1;
	\end{equation*}
	where we use $(1,\infty)\ni\lambda\mapsto\eta(\lambda):=2g(\lambda)-4\lambda+\ell\in\mathbb{R}_{<0}$. Also, $z_N:=\frac{z}{4N^2}>0$.
	\item[(3)] Near $\lambda=1$, $S(\lambda)$ is bounded and near $\lambda=0$, $S(\lambda)$ is weakly singular in the following entry-wise sense,
	\begin{equation*}
		S(\lambda)=\mathcal{O}\begin{bmatrix}1 & |\lambda|^{s}\\ 1 & |\lambda|^{s}\\ \end{bmatrix},\ -\frac{1}{2}<s<0;\ \ \ \ S(\lambda)=\mathcal{O}\begin{bmatrix}1 & \ln|\lambda| \\ 1 & \ln|\lambda|\end{bmatrix},\ s=0.
	\end{equation*}
	If $s>0$, then $S(\lambda)$ is bounded near $\lambda=0$ as well.
	\item[(4)] As $\lambda\rightarrow\infty$ and $\lambda\notin(0,\infty)$,
	\begin{equation*}
		S(\lambda)=I+\frac{r_1}{\lambda}+\mathcal{O}\big(\lambda^{-2}\big),
	\end{equation*}
	where $r_1=r_1(z,s,N)$ relates to $o_1(z,s,N)$ in condition $(4)$ of RHP \ref{RHP3} as follows,
	\begin{equation*}
		r_1(z,s,N)=\frac{1}{4}\bigg[N\sigma_3+\frac{1}{N}(4N)^{-(N+s)\sigma_3}\e^{\frac{N}{2}\ell\sigma_3}o_1\Big(\frac{z}{N},s,N\Big)\e^{-\frac{N}{2}\ell\sigma_3}(4N)^{(N+s)\sigma_3}\bigg].
	\end{equation*}
\end{enumerate}
\end{problem}
After the transformation \eqref{e54}, the jump for $S(\lambda)$ is, as $N\rightarrow\infty$, localised on $(0,1)\subset\mathbb{R}$ and near the vicinities of $\lambda=0$ and $\lambda=1$. Still, given the fast oscillations on $(0,1)$, we first introduce
\begin{equation}\label{e55}
	\mathbb{C}\setminus(-\infty,1]\ni\lambda\mapsto\xi(\lambda):=-2\int_1^{\lambda}\bigg(\frac{\zeta-1}{\zeta}\bigg)^{\frac{1}{2}}\d\zeta,
\end{equation}
with path of integration in \eqref{e55} not crossing $[0,1]\subset\mathbb{R}$ and the root in the integrand analytic for $\zeta\in\mathbb{C}\setminus[0,1]$ so that it is positive for $\zeta>1$. Then,
\begin{equation*}
	2\xi_+(\lambda)=-2\xi_-(\lambda)=4\im\int_{\lambda}^1\sqrt{\frac{1-x}{x}}\d x=-\pi(\lambda),\ \ \lambda\in(0,1),
\end{equation*}
followed by $\xi_+(\lambda)-\xi_-(\lambda)=2\pi\im$ for $\lambda\in(-\infty,0)$ as well as $2\xi(\lambda)=2g(\lambda)-4\lambda+\ell=\eta(\lambda)$ for $\lambda\in(1,\infty)$. Lastly,
\begin{equation*}
	\frac{\d}{\d y}\xi(x\pm\im y)\Big|_{y=0}=2\sqrt{\frac{1-x}{x}}>0,\ \ \ x\in(0,1),
\end{equation*}
so that $\textnormal{Re}(\xi(\lambda))>0$ for $0<|\textnormal{Im}(\lambda)|<\varepsilon$ and $0<\textnormal{Re}(\lambda)<1$ with some $\varepsilon>0$. What results from the properties of $\xi$ is the factorisation
\begin{equation}\label{e56}
	\begin{bmatrix}\e^{N\pi(\lambda)} & \phi(\lambda)\smallskip\\ 0 & \e^{-N\pi(\lambda)}\end{bmatrix}=\begin{bmatrix}1 & 0\smallskip\\ \e^{-2N\xi_-(\lambda)}\phi(\lambda)^{-1} & 1\end{bmatrix}\begin{bmatrix}0&\phi(\lambda)\smallskip\\ -\phi(\lambda)^{-1} & 0\end{bmatrix}\begin{bmatrix}1 & 0\smallskip\\ \e^{-2N\xi_+(\lambda)}\phi(\lambda)^{-1}&1\end{bmatrix},\ \ \lambda\in(0,1),
\end{equation}
with $\phi(\lambda)=\phi(\lambda;z_N,s):=(\lambda+z_N)^s\lambda^s$ that we define in the complex plane $\mathbb{C}\ni\lambda$ with cut on the ray $(-\infty,0]\subset\mathbb{R}$ such that $\phi(\lambda)\sim\lambda^{2s}$ as $\lambda\rightarrow+\infty$.
Then, denoting with $\Omega_{\pm}$ the lens-shaped domains shown in Figure \ref{fig6}, both of which lie in the domain $0<|\textnormal{Im}(\lambda)|<\varepsilon$ and $0<\textnormal{Re}(\lambda)<1$ and which do not intersect with $\Gamma$, we define
\begin{equation}\label{e57}
	T(\lambda;z,s,N):=S(\lambda;z,s,N)\begin{cases}\begin{bmatrix}1&0\\ \e^{-2N\xi(\lambda)}\phi(\lambda)^{-1}&1\end{bmatrix}^{-1},&\lambda\in\Omega_+\smallskip\\
	\begin{bmatrix}1&0\\ \e^{-2N\xi(\lambda)}\phi(\lambda)^{-1}&1\end{bmatrix},&\lambda\in\Omega_-\smallskip\\ I,&\lambda\in\mathbb{C}\setminus\overline{\Omega_+\cup\Omega_-}
	\end{cases},
\end{equation}
and are naturally led to the below RHP, compare \eqref{e56}.
\begin{problem}\label{RHP5} Let $N\in\mathbb{N},s>-\frac{1}{2}$ and $z>0$. The function $T(\lambda)=T(\lambda;z,s,N)\in\mathbb{C}^{2\times 2}$ defined in \eqref{e57} has the following properties:
\begin{enumerate}
	\item[(1)] $T(\lambda)$ is analytic for $\lambda\in\mathbb{C}\setminus\Sigma_T$ and extends continuously to $\mathbb{C}\setminus\{0\}$ from either side of $\Sigma_T$. The oriented contour $\Sigma_T$ is shown in Figure \ref{fig6} and consists of $(0,\infty)\subset\mathbb{R}$ as well as the two lens boundaries $\gamma_{\pm}$ connecting $\lambda=0$ and $\lambda=1$, but excluding both endpoints.
	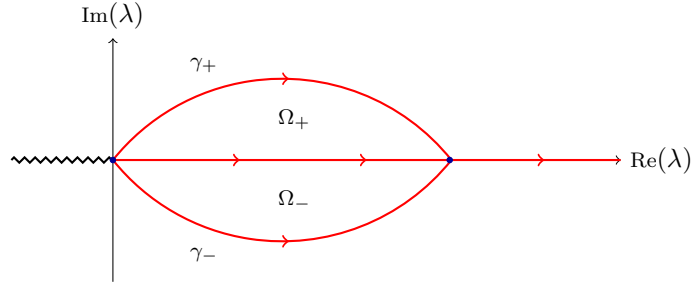
\begin{figure}[tbh]
	\begin{tikzpicture}[xscale=0.9,yscale=0.9]
	        \draw [->] (-2.5,0) -- (5,0) node [right] {$\footnotesize{\textnormal{Re}(\lambda)}$};
	\draw [thick, color=red, decoration={markings, mark=at position 0.25 with {\arrow{>}}},decoration={markings, mark=at position 0.5 with {\arrow{>}}},  decoration={markings, mark=at position 0.85 with {\arrow{>}}}, postaction={decorate}] (-2.5,0) -- (5,0);
	  \draw [->] (-2.5,-1.8) -- (-2.5,1.8) node [above] {$\footnotesize{\textnormal{Im}(\lambda)}$};
		\draw[thick,color=red, decoration={markings, mark=at position 0.5 with {\arrow{<}}}, postaction={decorate} ] (2.5,0) arc (38.65980825:141.3401918:3.201562119);
		\draw[thick,color=red, decoration={markings, mark=at position 0.5 with {\arrow{<}}}, postaction={decorate} ] (2.5,0) arc (-38.65980825:-141.3401918:3.201562119);
	\node [right] at (-1.5,1.4) {{\footnotesize $\gamma_+$}};
	\node [right] at (-0.2,0.6) {{\footnotesize $\Omega_+$}};
	\node [right] at (-1.5,-1.4) {{\footnotesize $\gamma_-$}};
	\node [right] at (-0.2,-0.6) {{\footnotesize $\Omega_-$}};
	\draw [thick,decorate,decoration={zigzag,segment length=4,amplitude=1,post=lineto,post length=0}] (-4,0) -- (-2.5,0);
	\draw[fill, color=blue!60!black] (2.475,0) circle [radius=0.04];
	\draw[fill, color=blue!60!black](-2.5,0) circle [radius=0.04];
\end{tikzpicture}
\caption{The oriented jump contour $\Sigma_T$, in the complex $\lambda$-plane drawn in \textcolor{red}{red}, for $T(\lambda)$. The (end)points $\lambda=0,1$ are colored in \textcolor{blue!60!black}{blue}, and the branch cut for $\lambda\mapsto\phi(\lambda)$ is shown on $(-\infty,0]$.}
\label{fig6}
\end{figure}

	\item[(2)] The continuous limiting values $T_{\pm}(\lambda),\lambda\in\Sigma_T\setminus\{0\}$ satisfy
	\begin{align*}
		T_+(\lambda)=&\,\,T_-(\lambda)\begin{bmatrix}1&0\\ \e^{-2N\xi(\lambda)}\phi(\lambda)^{-1}&1\end{bmatrix},\ \ \lambda\in\gamma_{\pm}\setminus\{0\};\\
		T_+(\lambda)=&\,\,T_-(\lambda)\begin{bmatrix}1&(\lambda+z_N)^s\lambda^s\e^{2N\xi(\lambda)}\\ 0&1\end{bmatrix},\ \lambda>1;
	\end{align*}
	as well as
	\begin{equation}\label{e58}
		T_+(\lambda)=T_-(\lambda)\begin{bmatrix}0&\phi(\lambda)\\ -\phi(\lambda)^{-1}&0\end{bmatrix},\ \ \lambda\in(0,1);\ \ \ \phi(\lambda)=\lambda^s(\lambda+z_N)^s,\ \lambda\in\mathbb{C}\setminus(-\infty,0].
	\end{equation}
	\item[(3)] Near $\lambda=1$, $T(\lambda)$ is bounded and near $\lambda=0$, $T(\lambda)$ is weakly singular in that the behaviour of $S(\lambda)$ as written in condition $(3)$ of RHP \ref{RHP4} is modified according to \eqref{e57}.
	\item[(4)] As $\lambda\rightarrow\infty,\lambda\notin(0,\infty)$,
	\begin{equation*}
		T(\lambda)=I+\frac{r_1}{\lambda}+\mathcal{O}\big(\lambda^{-2}\big),
	\end{equation*}
	with $r_1=r_1(z,s,N)$ as in condition $(4)$ of RHP \ref{RHP4}.
\end{enumerate}
\end{problem}
With $G_T(\lambda)=G_T(\lambda;z,s,N)\in\mathbb{C}^{2\times 2}$ denoting the piecewise defined jump matrix in condition $(2)$ of RHP \ref{RHP5}, the reason for the transformation \eqref{e57} transpires from the following small norm estimate.
\begin{prop}\label{norm1} Suppose $\mathbb{D}_r(z_0)=\{z\in\mathbb{C}:\,|z-z_0|<r\}$ denotes the open disk with center $z_0\in\mathbb{C}$ and radius $r>0$. Then for any $\epsilon\in(0,\frac{1}{2}),s>-\frac{1}{2}$ and $z>0$ there exist $N_0=N_0(\epsilon,s,z)\in\mathbb{N}$ and $c_j=c_j(\epsilon,s,z)>0$ so that

\begin{equation*}
	\|G_T(\cdot;z,s,N)-I\|_{L^2\cap L^{\infty}(\Sigma_T\setminus(\mathbb{D}_{\epsilon}(0)\cup\mathbb{D}_{\epsilon}(1)\cup(\epsilon,1-\epsilon)))}\leq c_1\e^{-c_2N}\ \ \ \ \ \ \forall\,N\geq N_0.
\end{equation*}
\end{prop}
So, the jump matrix $G_T(\lambda)$ away from $\lambda\in(0,1)\subset\mathbb{R}$ and away from two small vicinities of $\lambda=0,\lambda=1$ is close to the identity matrix as $N\rightarrow\infty$, for fixed $(s,z)$, and we can thus focus on the necessary local analysis.

\subsection{The outer parametrix} This one will be modelled according to the jump \eqref{e58} and thus amounts to the use of the Szeg\H{o} function
\begin{equation*}
	\mathcal{D}(\lambda)=\mathcal{D}(\lambda;z_N,s):=\exp\left[-\frac{1}{2\pi}\big(\lambda(\lambda-1)\big)^{\frac{1}{2}}\int_0^1\frac{\ln\phi(x;z_N,s)}{\sqrt{x(1-x)}}\frac{\d x}{x-\lambda}\right],\ \ \lambda\in\mathbb{C}\setminus[0,1],
\end{equation*}
which can be evaluated through explicit quadrature, and where the fractional power is the principal one with cut on $[0,1]\subset\mathbb{R}$. Namely, with principal branches throughout,
\begin{equation}\label{Sze}
	\mathcal{D}(\lambda)=\bigg[\frac{\lambda}{\upsilon(2\lambda-1)}\bigg]^{\frac{s}{2}}\bigg[\frac{\lambda+\vartheta(2\lambda-1;z_N)}{\upsilon(2\lambda-1)}\bigg]^{\frac{s}{2}},\ \ \vartheta(\lambda;z_N):=\lambda z_N+\big(\lambda^2-1\big)^{\frac{1}{2}}\sqrt{z_N(z_N+1)},
\end{equation}
where $\upsilon(\lambda):=\lambda+(\lambda^2-1)^{\frac{1}{2}},\lambda\in\mathbb{C}\setminus[-1,1]$ maps to the exterior of the unit disk. What we need in particular is knowledge of the expansion
\begin{equation*}
	\mathcal{D}(\lambda;z_N,s)=2^{-2s}\bigg(\sqrt{1+z_N}+\sqrt{z_N}\bigg)^s\bigg\{1+\frac{s}{2\lambda}\frac{\sqrt{1+z_N}}{\sqrt{1+z_N}+\sqrt{z_N}}+\mathcal{O}\big(\lambda^{-2}\big)\bigg\},\ \ \lambda\rightarrow\infty,
\end{equation*}
and knowledge of the jump relation $\mathcal{D}_+(\lambda)\mathcal{D}_-(\lambda)=\phi(\lambda),\lambda\in(0,1)$. The Szeg\H{o} function then leads to
\begin{equation}\label{e59}
	P(\lambda;z_N,s):=\mathcal{D}(\infty)^{\sigma_3}\frac{1}{2}\begin{bmatrix}\delta(\lambda)+\delta(\lambda)^{-1} & -\im\big(\delta(\lambda)-\delta(\lambda)^{-1}\big)\\ \im\big(\delta(\lambda)-\delta(\lambda)^{-1}\big) & \delta(\lambda)+\delta(\lambda)^{-1}\end{bmatrix}\mathcal{D}(\lambda)^{-\sigma_3},\ \ \ \lambda\in\mathbb{C}\setminus[0,1],
\end{equation}
with $\delta(\lambda):=(1-1/\lambda)^{\frac{1}{4}}$ analytic off $[0,1]\subset\mathbb{R}$ such that $\delta(\lambda)>0$ for $\lambda>1$, with $\mathcal{D}(\infty)=\lim_{\lambda\rightarrow\infty}\mathcal{D}(\lambda;z_N,s)$, and in turn the below model RHP.
\begin{problem}[Outer parametrix]\label{outerRHP} Let $s>-\frac{1}{2},z>0$ and $N\in\mathbb{N}$. The model function $P(\lambda)=P(\lambda;z_N,s)\in\mathbb{C}^{2\times 2}$ defined in \eqref{e59} has the following properties:
\begin{enumerate}
	\item[(1)] $P(\lambda)$ is analytic for $\lambda\in\mathbb{C}\setminus[0,1]$ and extends continuously to $\mathbb{C}\setminus\{0,1\}$.
	\item[(2)] On $(0,1)\ni\lambda$ the boundary values $P_{\pm}(\lambda)=\lim_{\epsilon\downarrow 0}P(\lambda\pm\im\epsilon)$ obey
	\begin{equation*}
		P_+(\lambda)=P_-(\lambda)\begin{bmatrix}0 & \phi(\lambda)\\ -\phi(\lambda)^{-1} & 0\end{bmatrix};\ \ \ \textnormal{with} \ \ \phi(\lambda)=(\lambda+z_N)^s\lambda^s,\ \ \lambda\in\mathbb{C}\setminus(-\infty,0].
	\end{equation*}
	\item[(3)] The mapping $(0,1)\ni\lambda\mapsto P_{\pm}(\lambda)$ is entry-wise in $L^2(0,1)$.
	\item[(4)] As $\lambda\rightarrow\infty$,
	\begin{equation*}
		P(\lambda)=I+\frac{P_1}{\lambda}+\mathcal{O}\big(\lambda^{-2}\big),
	\end{equation*}
	with leading order coefficient $P_1=P_1(z_N,s)$ equal to
	\begin{equation*}
		P_1(z_N,s)=\frac{\im}{4}\mathcal{D}(\infty)^{\sigma_3}\begin{bmatrix}0&1\\ -1&0\end{bmatrix}\mathcal{D}(\infty)^{-\sigma_3}-\frac{s}{2}\frac{\sqrt{1+z_N}\,\sigma_3}{\sqrt{1+z_N}+\sqrt{z_N}}.
	\end{equation*}
\end{enumerate}
\end{problem}
We expect that $T(\lambda)$ in RHP \ref{RHP5} is well approximated by $P(\lambda)$, as $N\rightarrow\infty$, outside $\mathbb{D}_{\epsilon}(0)\cup\mathbb{D}_{\epsilon}(1)$, for fixed $(s,z)$.
\subsection{The parametrix at $\lambda=0$} The constraint $s\in(-\frac{1}{2},\infty)\subset\mathbb{R}$ will be essential for we use Painlev\'e functions in our construction. The basic idea is to view
\begin{equation*}
	\lambda\mapsto \lambda^s\e^{-N\lambda}(\lambda+z_N)^s
\end{equation*}
as a Fisher-Hartwig perturbation of the Laguerre weight function, where the root singularity at $\lambda=-z_N<0$ falls into the hard edge $\lambda=0$ at the hard-edge rate $\mathcal{O}(N^{-2})$. So, instead of using Bessel functions in the local analysis, Painlev\'e-III transcendents will appear, and those will turn out to be pole-free for $z>0$ and $s>-\frac{1}{2}$, with $s>-\frac{1}{2}$ being the crucial assumption. To begin with we note that 
\begin{equation*}
	\mathbb{D}_{\epsilon}(0)\ni\lambda\mapsto\tau(\lambda):=\big(\im\pi\,\textnormal{sgn}(\textnormal{Im}(\lambda))-\xi(\lambda)\big)^2\stackrel{\eqref{e55}}{=}-16\lambda\bigg(1-\frac{\lambda}{3}+\mathcal{O}(\lambda^2)\bigg),\ \ \ \lambda\rightarrow 0,
\end{equation*}
is locally analytic and so the change of variables, where $\beta>0$ is arbitrary for now,
\begin{equation}\label{e60}
	\mathbb{D}_{\epsilon}(0)\ni\lambda\mapsto\zeta(\lambda)=\zeta(\lambda;z_N):=\beta\e^{-\im\pi}\frac{\tau(\lambda)}{\tau(-z_N)}
\end{equation}
is locally conformal with $\zeta(0)=0$ and $\zeta(-z_N)=-\beta$, furthermore
\begin{equation*}
	\tau(-z_N)=\frac{4z}{N^2}\Big(1+\mathcal{O}\big(N^{-2}\big)\Big)\ \ \ \textnormal{as}\ \ N\rightarrow\infty,
\end{equation*}
and so $\tau(-z_N)>0$ for $z>0$, as $N\rightarrow\infty$. Consequently, we fix $0<\epsilon<\frac{1}{2}$ and consider the model function
\begin{equation}\label{e61}
	U(\lambda;z,s,N):=E(\lambda)Q(\zeta;x,s)\e^{\im\varpi(\zeta;x)\sigma_3}\varphi(\lambda;z_N,s)^{-\frac{1}{2}\sigma_3}\big|_{\zeta=\zeta(\lambda;z_N),\,x=N\sqrt{\tau(-z_N)/\beta}},\ \ \lambda\in\mathbb{D}_{\epsilon}(0)\setminus\Sigma_T,
\end{equation}
defined in terms of the solution $Q(\zeta;x,s)$ of RHP \ref{PIIImodel}, where $\zeta\mapsto\varpi(\zeta;x)=x\zeta^{\frac{1}{2}}$ is cut on $[0,\infty)\subset\mathbb{R}$ with branch fixed by the choice $\textnormal{arg}\,\zeta\in(0,2\pi)$, where
\begin{equation*}
	\varphi(\lambda)=\varphi(\lambda;z_N,s):=(-\lambda)^s(-\lambda-z_N)^s,\ \ \ \ \lambda\in\mathbb{C}\setminus[-z_N,\infty)\ \ \ \textnormal{so that}\ \ \ \varphi(\lambda)>0\ \ \textnormal{for}\ \ \lambda<-z_N,
\end{equation*}
and $\mathbb{D}_{\epsilon}(0)\ni\lambda\mapsto E(\lambda)$ denotes the analytic multiplier
\begin{equation*}
	E(\lambda)=E(\lambda;z_N,s):=P(\lambda;z_N,s)\varphi(\lambda;z_N,s)^{\frac{1}{2}\sigma_3}E_0\big(\zeta(\lambda;z_N);s\big)^{-1}\begin{bmatrix}1 & \im\sqrt{\beta} s\\ 0 & 1\end{bmatrix},\ \ \lambda\in\mathbb{D}_{\epsilon}(0),
\end{equation*}
with
\begin{equation*}
	E_0(\zeta;s):=\zeta^{\frac{1}{4}\sigma_3}\frac{1}{\sqrt{2}}\begin{bmatrix}1&1\\ -1&1\end{bmatrix}\e^{-\im\frac{\pi}{4}\sigma_3}\exp\left[\im\zeta^{\frac{1}{2}}\int_{-\beta}^0\frac{s\sigma_3}{2\sqrt{-x}}\frac{\d x}{x-\zeta}\right],\ \ \zeta\in\mathbb{C}\setminus[-\beta,\infty),\ \ \ \textnormal{arg}\,\zeta\in(0,2\pi).
\end{equation*}
\begin{remark} Analyticity of $\lambda\mapsto E(\lambda)$ in $\mathbb{D}_{\epsilon}(0)$ follows from the observation that $\lambda\mapsto\zeta(\lambda;z_N)$ is locally conformal and so
\begin{equation*}
	E_{0+}(\zeta)=E_{0-}(\zeta)\begin{bmatrix}0&1\\ -1&0\end{bmatrix},\ \zeta\in(0,\epsilon);\ \ \ \ \ \ E_{0+}(\zeta)=E_{0-}(\zeta)\begin{cases}\displaystyle \e^{-\im\pi s\sigma_3},&\zeta\in(-\beta,0)\\
	 I,& \zeta\in(-\epsilon,-\beta)
	 \end{cases},
\end{equation*}
tell us that $E(\lambda)$ is analytic in $\mathbb{D}_{\epsilon}(0)\setminus\{-z_N,0\}$, where we assume $N\geq N_0$ so $z_N\in\mathbb{D}_{\epsilon}(0)$. On the other hand, our explicit formul\ae\,yield the rough estimates
\begin{equation*}
	E(\lambda)=\mathcal{O}(1),\ \ \ \lambda\rightarrow -z_N;\ \ \ \ \ \ \ \ E(\lambda)=\mathcal{O}\big(|\lambda|^{-\frac{1}{2}}\big),\ \ \ \lambda\rightarrow 0,
\end{equation*}
and thus $\lambda\mapsto E(\lambda)$ must be analytic in $\mathbb{D}_{\epsilon}(0)$.
\end{remark}

\begin{remark} By RHP \ref{PIIImodel}, condition $(4)$,
\begin{align*}
	Q(\zeta;x,s)\e^{\im\varpi(\zeta;x)\sigma_3}E_0(\zeta;s)^{-1}=\begin{bmatrix}1&-\im\sqrt{\beta} s\\ 0&1\end{bmatrix}&\Bigg\{I+\frac{1}{\zeta}\begin{bmatrix}\beta s^2/2 & \im\beta^{\frac{3}{2}}(s-s^3)/3\smallskip\\ -\im\sqrt{\beta} s & -\beta s^2/2\end{bmatrix}+\frac{1}{\zeta}\begin{bmatrix}1&\im\sqrt{\beta} s\\ 0 &1\end{bmatrix}\\
	&\hspace{1cm}\times Q_1(x,s)\begin{bmatrix}1 & -\im\sqrt{\beta} s\\ 0&1\end{bmatrix}+\mathcal{O}\big(\zeta^{-2}\big)\Bigg\},
\end{align*}
as $\zeta\rightarrow\infty$ and $\zeta\notin\Sigma_Q$, uniformly for $x>0$ and $s>-\frac{1}{2}$, both chosen from compact subsets.
\end{remark}

What results from RHP \ref{PIIImodel} and the various branch conventions are the following analytic and asymptotic, as $N\rightarrow\infty$, properties of $U(\lambda)$.
\begin{problem}[Origin parametrix]\label{orip} Let $s>-\frac{1}{2},z>0$ and $N\geq N_0$ so that $z_N\in\mathbb{D}_{\epsilon}(0)$. The model function $U(\lambda)=U(\lambda;z,s,N)\in\mathbb{C}^{2\times 2}$ defined in \eqref{e61} has the following properties:
\begin{enumerate}
	\item[(1)] $\lambda\mapsto U(\lambda)$ is analytic for $\lambda\in\mathbb{D}_{\epsilon}(0)\setminus\Sigma_T$ where we assume that $\gamma_{\pm}$ in Figure \ref{fig6} fall into $\lambda=0$ at angle $\frac{\pi}{3}$ and $\frac{5\pi}{3}$. Moreover, $U(\lambda)$ attains continuous limiting values $U_{\pm}(\lambda),\lambda\in(\Sigma_T\cap\mathbb{D}_{\epsilon}(0))\setminus\{0\}$.
	\item[(2)] The limiting values $U_{\pm}(\lambda),\lambda\in(\Sigma_T\cap\mathbb{D}_{\epsilon}(0))\setminus\{0\}$ satisfy the jump conditions
	\begin{align*}
		U_+(\lambda)=&\,\,U_-(\lambda)\begin{bmatrix}0 & \phi(\lambda)\\ -\phi(\lambda)^{-1}&0\end{bmatrix},\ \ \lambda\in(0,\epsilon);\\
		U_+(\lambda)=&\,\,U_-(\lambda)\begin{bmatrix}1&0\\ \e^{-2N\xi(\lambda)}\phi(\lambda)^{-1}&1\end{bmatrix},\ \ \lambda\in\big(\gamma_{\pm}\cap\mathbb{D}_{\epsilon}(0)\big)\setminus\{0\}.
	\end{align*}
	\item[(3)] Near $\lambda=0$, $U(\lambda)$ is weakly singular and of the form
	\begin{equation*}
		U(\lambda)=\widehat{U}(\lambda)\begin{cases}\begin{bmatrix}1 & 0\\ \e^{-2N\xi(\lambda)}\phi(\lambda)^{-1} & 1\end{bmatrix}^{-1},&\lambda\in\Omega_+\cap\mathbb{D}_{\epsilon}(0)\smallskip\\
		\begin{bmatrix}1 & 0\\ \e^{-2N\xi(\lambda)}\phi(\lambda)^{-1} & 1\end{bmatrix},&\lambda\in\Omega_-\cap\mathbb{D}_{\epsilon}(0)\\ I,&\lambda\in\mathbb{D}_{\epsilon}(0)\setminus\big(\overline{\Omega_+}\cup\overline{\Omega_-}\big)\end{cases},
	\end{equation*}
	with $\Omega_{\pm}$ shown in Figure \ref{fig6} where
	\begin{equation*}
		\widehat{U}(\lambda)=\mathcal{O}\begin{bmatrix}1 & |\lambda|^s\\ 1&|\lambda|^s\end{bmatrix},\ \ -\frac{1}{2}<s<0;\ \ \ \ \ \ \ \widehat{U}(\lambda)=\mathcal{O}\begin{bmatrix}1 & \ln|\lambda|\\ 1&\ln|\lambda|\end{bmatrix},\ \ s=0,
	\end{equation*}
	and $\widehat{U}(\lambda)$ bounded at $\lambda=0$ for all other $s>0$.
	\item[(4)] As $N\rightarrow\infty$, $U(\lambda)$ in \eqref{e61} compares to $P(\lambda)$ in \eqref{e59} as follows:
	\begin{equation}\label{e62}
		U(\lambda)=\bigg\{I+U_N(\lambda;z_N,s)+\mathcal{O}\big(N^{-2}\big)\bigg\}P(\lambda),
	\end{equation}
	uniformly for $0<\delta_1\leq|\lambda|\leq\delta_2<\frac{1}{2}$, for any fixed $z>0$ and $s>-\frac{1}{2}$. Here,
	\begin{equation*}
		U_N(\lambda;z_N,s)=\frac{N}{\zeta(\lambda;z_N)}\Big(Q_1^{21}\big(N\sqrt{\tau(-z_N)/\beta},s\big)-\im\sqrt{\beta} s\Big)\widehat{E}_N(\lambda)\begin{bmatrix}0&0\\ 1&0\end{bmatrix}\widehat{E}_N(\lambda)^{-1}=\mathcal{O}\big(N^{-1}\big),
	\end{equation*}
	that we express in terms of the locally analytic
	\begin{equation*}
		\widehat{E}_N(\lambda)=\widehat{E}_N(\lambda;z_N,s):=E(\lambda;z_N,s)\begin{bmatrix}1&-\im\sqrt{\beta} s\\ 0 & 1\end{bmatrix}N^{\frac{1}{2}\sigma_3},\ \ \ \lambda\in\mathbb{D}_{\epsilon}(0),
	\end{equation*}
	and we have used that $\widehat{E}_N(\lambda)=\mathcal{O}(1)$ as $N\rightarrow\infty$ for $0<\delta_1\leq|\lambda|\leq\delta_2<\frac{1}{2}$ and similarly $\zeta(\lambda;z_N)=\mathcal{O}(N^2)$ in the same limit by \eqref{e60}. Note also that
	\begin{equation*}
		x=N\sqrt{\tau(-z_N)/\beta}=2\sqrt{\frac{z}{\beta}}\Big(1+\mathcal{O}\big(N^{-2}\big)\Big),\ \ \ N\rightarrow\infty,
	\end{equation*}
	uniformly for $z>0$ chosen on compact subsets.
\end{enumerate}
\end{problem}
The above concludes our explicit construction near $\lambda=0$, we now move to a neighbourhood of $\lambda=1$.
\begin{remark} By conditions $(2)$ and $(3)$ in RHP \ref{orip} we conclude that
\begin{equation*}
	\lambda\mapsto T(\lambda;z,s,N)U(\lambda;z,s,N)^{-1}
\end{equation*}
is analytic at $\lambda=0$.
\end{remark}

\subsection{The parametrix at $\lambda=1$} Seeing that
\begin{equation*}
	\xi(\lambda)\stackrel{\eqref{e55}}{=}-\frac{4}{3}(\lambda-1)^{\frac{3}{2}}\bigg(1-\frac{3}{10}(\lambda-1)+\mathcal{O}\big((\lambda-1)^2\big)\bigg),\ \ \ \lambda\rightarrow 1,\ \ \lambda\notin[0,1],\ \ \textnormal{arg}(\lambda-1)\in(-\pi,\pi),
\end{equation*}
the relevant construction is standard, explicit, compare \cite[page $150$]{V}, and invokes Airy functions. We begin by introducing the locally conformal change of coordinates
\begin{equation}\label{e63}
	\mathbb{D}_{\epsilon}(1)\ni\lambda\mapsto \zeta(\lambda)=\zeta(\lambda;N):=\bigg[-\frac{3N}{2}\xi(\lambda)\bigg]^{\frac{2}{3}},
\end{equation}
that has $\zeta(1)=0$ and we consider the model function
\begin{equation}\label{e64}
	V(\lambda;z,s,N):=F(\lambda)A(\zeta)\e^{\frac{2}{3}\zeta^{\frac{3}{2}}\sigma_3}\phi(\lambda;z_N,s)^{-\frac{1}{2}\sigma_3}\Big|_{\zeta=\zeta(\lambda;N)},\ \ \ \lambda\in\mathbb{D}_{\epsilon}(1)\setminus\Sigma_T,
\end{equation}
which is defined in terms of the solution $A(\zeta)$ of RHP \ref{Airymodel}, where
\begin{equation*}
	\phi(\lambda)=\phi(\lambda;z_N,s)=(\lambda+z_N)^s\lambda^s,\ \ \ \lambda\in\mathbb{C}\setminus(-\infty,0],
\end{equation*}
appeared in \eqref{e58} and $\mathbb{D}_{\epsilon}(1)\ni\lambda\mapsto F(\lambda)$ denotes the analytic multiplier
\begin{equation*}
	F(\lambda)=F(\lambda;z_N,s):=P(\lambda;z_N,s)\phi(\lambda;z_N,s)^{\frac{1}{2}\sigma_3}F_0\big(\zeta(\lambda;N)\big)^{-1},\ \ \ \lambda\in\mathbb{D}_{\epsilon}(1),
\end{equation*}
with
\begin{equation*}
	F_0(\zeta):=\zeta^{-\frac{1}{4}\sigma_3}\frac{1}{\sqrt{2}}\begin{bmatrix}1&1\\ -1 & 1\end{bmatrix}\e^{-\im\frac{\pi}{4}\sigma_3},\ \ \ \zeta\in\mathbb{C}\setminus(-\infty,0],\ \ \ \textnormal{arg}\,\zeta\in(-\pi,\pi).
\end{equation*}
\begin{remark} Analyticity of $\lambda\mapsto F(\lambda)$ in $\mathbb{D}_{\epsilon}(1)$ follows from the observation that $\lambda\mapsto\zeta(\lambda;N)$ is locally conformal and so
\begin{equation*}
	F_{0+}(\zeta)=F_{0-}(\zeta)\begin{bmatrix}0&1\\ -1 & 0\end{bmatrix},\ \ \ \zeta\in(-\epsilon,0);\ \ \ \ \ \ \ \ F_{0+}(\zeta)=F_{0-}(\zeta),\ \ \ \zeta\in(0,\epsilon),
\end{equation*}
tell us that $F(\lambda)$ is analytic in $\mathbb{D}_{\epsilon}(1)\setminus\{1\}$. On the other hand, our explicit formul\ae\,yield the rough estimate
\begin{equation*}
	F(\lambda)=\mathcal{O}\big(|\lambda-1|^{-\frac{1}{2}}\big),\ \ \ \ \lambda\rightarrow 1,
\end{equation*}
and thus $\lambda\mapsto F(\lambda)$ must be analytic in $\mathbb{D}_{\epsilon}(1)$.
\end{remark}
What results from \eqref{e63} are the following analytic and asymptotic properties of $V(\lambda)$.
\begin{problem}[Parametrix at $\lambda=1$]\label{1para} Let $s>-\frac{1}{2}$ and $z>0$. The model function $V(\lambda)=V(\lambda;z,N,s)\in\mathbb{C}^{2\times 2}$ defined in \eqref{e64} has the following properties:
\begin{enumerate}
	\item[(1)] $\lambda\mapsto V(\lambda)$ is analytic for $\lambda\in\mathbb{D}_{\epsilon}(1)\setminus\Sigma_T$ where we assume that $\gamma_{\pm}$ in Figure \ref{fig6} fall into $\lambda=1$ at angle $\frac{2\pi}{3}$ and $-\frac{2\pi}{3}$. Moreover, $V(\lambda)$ attains continuous limiting values $V_{\pm}(\lambda),\lambda\in(\Sigma_T\cap\mathbb{D}_{\epsilon}(1))\setminus\{1\}$.
	\item[(2)] The limiting values $V_{\pm}(\lambda),\lambda\in(\Sigma_T\cap\mathbb{D}_{\epsilon}(1))\setminus\{1\}$ satisfy the jump conditions
	\begin{align*}
		V_+(\lambda)=&\,\,V_-(\lambda)\begin{bmatrix}0&\phi(\lambda)\\ -\phi(\lambda)^{-1}&0\end{bmatrix},\ \ \ \lambda\in(1-\epsilon,1);\\
		V_+(\lambda)=&\,\,V_-(\lambda)\begin{bmatrix}1&\phi(\lambda)\e^{2N\xi(\lambda)}\\ 0 &1\end{bmatrix},\ \ \ \lambda\in(1,1+\epsilon);\\
		V_+(\lambda)=&\,\,V_-(\lambda)\begin{bmatrix}1&0\\ \e^{-2N\xi(\lambda)}\phi(\lambda)^{-1}&1\end{bmatrix},\ \ \ \lambda\in\big(\gamma_{\pm}\cap\mathbb{D}_{\epsilon}(1)\big)\setminus\{1\}.
	\end{align*}
	\item[(3)] Near $\lambda=1$, $V(\lambda)$ is bounded.
	\item[(4)] As $N\rightarrow\infty$, $V(\lambda)$ in \eqref{e63} compares to $P(\lambda)$ in \eqref{e59} as follows:
	\begin{equation}\label{e65}
		V(\lambda)=\bigg\{I+V_N(\lambda;z_N,s)+\mathcal{O}\big(N^{-2}\big)\bigg\}P(\lambda),
	\end{equation}
	uniformly for $0<\delta_1\leq|\lambda-1|\leq\delta_2<\frac{1}{2}$, for any fixed $z>0$ and $s>-\frac{1}{2}$. Here,
	\begin{equation*}
		V_N(\lambda;z_N,s)=-\frac{7N^{-\frac{1}{3}}}{48\zeta(\lambda;N)}\widehat{F}_N(\lambda)\begin{bmatrix}0&0\\ 1&0\end{bmatrix}\widehat{F}_N(\lambda)^{-1}+\frac{5N^{\frac{1}{3}}}{48\zeta(\lambda;N)^2}\widehat{F}_N(\lambda)\begin{bmatrix}0&1\\ 0&0\end{bmatrix}\widehat{F}_N(\lambda)^{-1}=\mathcal{O}\big(N^{-1}\big),
	\end{equation*}
	that we express in terms of the locally analytic
	\begin{equation*}
		\widehat{F}_N(\lambda)=\widehat{F}_N(\lambda;z_N,s):=F(\lambda;z_N,s)N^{-\frac{1}{6}\sigma_3},\ \ \ \lambda\in\mathbb{D}_{\epsilon}(1),
	\end{equation*}
	and where we have used that $\widehat{F}_N(\lambda)=\mathcal{O}(1)$ as $N\rightarrow\infty$ for $0<\delta_1\leq|\lambda-1|\leq\delta_2<\frac{1}{2}$ and similarly $\zeta(\lambda;N)=\mathcal{O}(N^{\frac{2}{3}})$ in the same limit by \eqref{e63}.
\end{enumerate}
\end{problem}
The above concludes our explicit construction near $\lambda=1$.
\begin{remark} By conditions $(2)$ and $(3)$ in RHP \ref{1para} we conclude that
\begin{equation*}
	\lambda\mapsto T(\lambda;z,s,N)V(\lambda;z,s,N)^{-1}
\end{equation*}
is analytic at $\lambda=1$. Furthermore, the restriction $2s>-1$ and $z>0$ is easily lifted to $\textnormal{Re}(2s)>-1$ and $\textnormal{Re}(z)>0$ in \eqref{e64}, which is not the case when working with \eqref{e61}.
\end{remark}

\subsection{The final ratio transformation} We fix $0<\epsilon<\frac{1}{4}$ and compare the explicit $P(\lambda),U(\lambda),V(\lambda)$ in \eqref{e59},\eqref{e61},\eqref{e64} to $T(\lambda)$ from RHP \ref{RHP5}. This amounts to defining
\begin{equation}\label{e66}
	R(\lambda;z,s,N):=T(\lambda;z,s,N)\begin{cases}U(\lambda;z,s,N)^{-1},&\lambda\in\mathbb{D}_{\epsilon}(0)\smallskip\\
	V(\lambda;z,s,N)^{-1},&\lambda\in\mathbb{D}_{\epsilon}(1)\smallskip\\
	P(\lambda;z_N,s)^{-1},&\lambda\in\mathbb{C}\setminus(\Sigma_T\cup\overline{\mathbb{D}_{\epsilon}(0)}\cup\overline{\mathbb{D}_{\epsilon}(1)})
	\end{cases},
\end{equation}
and collecting the properties of the resulting ratio function below.
\begin{problem}[Small norm problem]\label{ratio} Let $s>-\frac{1}{2},z>0$ and $N\geq N_0$ so that $z_N\in\mathbb{D}_{\epsilon}(0)$. The function $R(\lambda)=R(\lambda;z,s,N)\in\mathbb{C}^{2\times 2}$ defined in \eqref{e66} has the following properties:
\begin{enumerate}
	\item[(1)] $\lambda\mapsto R(\lambda)$ is analytic for $\lambda\in\mathbb{C}\setminus\Sigma_R$ with $\Sigma_R$ shown in Figure \ref{fig7}. On $\Sigma_R$, $R(\lambda)$ admits continuous boundary values $R_{\pm}(\lambda)$ as one approaches $\Sigma_R$ from either side of $\mathbb{C}\setminus\Sigma_R$ in a non-tangential fashion.
	\begin{figure}[tbh]
	\begin{tikzpicture}[xscale=0.9,yscale=0.9]
	        \draw [->] (-4,0) -- (5,0) node [right] {$\footnotesize{\textnormal{Re}(\lambda)}$};
	\draw [thick, color=red, decoration={markings, mark=at position 0.25 with {\arrow{>}}},decoration={markings, mark=at position 0.5 with {\arrow{>}}},  decoration={markings, mark=at position 0.85 with {\arrow{>}}}, postaction={decorate}] (-2.5,0) -- (5,0);
	  \draw [->] (-2.5,-1.8) -- (-2.5,1.8) node [above] {$\footnotesize{\textnormal{Im}(\lambda)}$};
		\draw[thick,color=red, decoration={markings, mark=at position 0.5 with {\arrow{<}}}, postaction={decorate} ] (2.5,0) arc (38.65980825:141.3401918:3.201562119);
		\draw[thick,color=red, decoration={markings, mark=at position 0.5 with {\arrow{<}}}, postaction={decorate} ] (2.5,0) arc (-38.65980825:-141.3401918:3.201562119);
	\node [right] at (-1.5,1.4) {{\footnotesize $\partial\Omega_+$}};
	\node [right] at (-1.5,-1.4) {{\footnotesize $\partial\Omega_-$}};
	\draw [thick, color=red, fill=white, decoration={markings, mark=at position 0.07 with {\arrow{<}}}, decoration={markings, mark=at position 0.4 with {\arrow{<}}}, decoration={markings, mark=at position 0.7 with {\arrow{<}}},postaction={decorate}] (-2.5,0) circle [radius=0.6];
	\draw [thick, color=red, fill=white, decoration={markings, mark=at position 0.15 with {\arrow{<}}}, decoration={markings, mark=at position 0.45 with {\arrow{<}}}, decoration={markings, mark=at position 0.75 with {\arrow{<}}},postaction={decorate}] (2.475,0) circle [radius=0.6];
	\draw[fill, color=blue!60!black] (2.475,0) circle [radius=0.04];
	\draw[fill, color=blue!60!black](-2.5,0) circle [radius=0.04];
\end{tikzpicture}
\caption{The oriented jump contour $\Sigma_R$, in the complex $\lambda$-plane drawn in \textcolor{red}{red}, for $R(\lambda)$. The points $\lambda=0,1$ are colored in \textcolor{blue!60!black}{blue}.}
\label{fig7}
\end{figure}
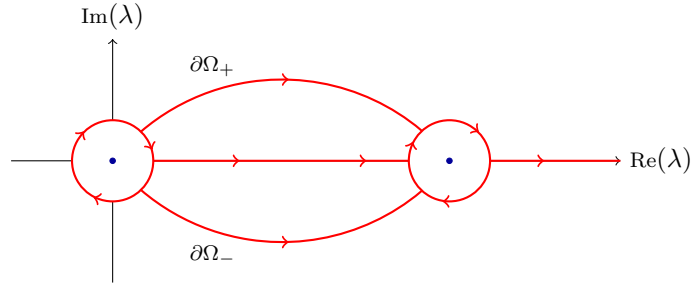

	\item[(2)] On $\Sigma_R$ we record the jump condition $R_+(\lambda)=R_-(\lambda)G_R(\lambda)$ with $G_R(\lambda)=G_R(\lambda;z,s,N)$ given as
	\begin{align*}
		G_R(\lambda)=&\,\,I+\e^{-2N\xi(\lambda)}\phi(\lambda)^{-1}P(\lambda)\begin{bmatrix}0&0\\ 1&0\end{bmatrix}P(\lambda)^{-1},\ \ \lambda\in\partial\Omega_{\pm};\\
		G_R(\lambda)=&\,\,I+\e^{2N\xi(\lambda)}\phi(\lambda)P(\lambda)\begin{bmatrix}0&1\\ 0&0\end{bmatrix}P(\lambda)^{-1},\ \ \lambda\in(1+\epsilon,\infty);
	\end{align*}
	as well as
	\begin{equation*}
		G_R(\lambda)=U(\lambda)P(\lambda)^{-1},\ \ \ \lambda\in\partial\mathbb{D}_{\epsilon}(0);\ \ \ \ \ \ \ \ G_R(\lambda)=V(\lambda)P(\lambda)^{-1},\ \ \ \lambda\in\partial\mathbb{D}_{\epsilon}(1).
	\end{equation*}
	By construction, $R(\lambda)$ has no jumps inside $\mathbb{D}_{\epsilon}(0)\cup\mathbb{D}_{\epsilon}(1)$ and on $(\epsilon,1-\epsilon)\subset\mathbb{R}$.
	\item[(3)] As $\lambda\rightarrow\infty$, $R(\lambda)\rightarrow I$.
\end{enumerate}
\end{problem}
The crucial observation that underwrites RHP \ref{ratio} is contained in the below small norm estimate. That estimate is a direct consequence of Proposition \ref{norm1}, of the fact that $P(\lambda)=\mathcal{O}(1)$ as $N\rightarrow\infty$ for $\lambda$ away from $\{0,1\}$, and of the matching conditions \eqref{e62} and \eqref{e65}.
\begin{prop} Fix $s>-\frac{1}{2}$ and $z>0$. There exist $c=c(s,z)>0,N_0=N_0(s,z)\in\mathbb{N}$ so that
\begin{equation*}
	\|G_R(\cdot;z,s,N)-I\|_{L^2\cap L^{\infty}(\Sigma_R)}\leq\frac{c}{N}\ \ \ \forall\,\epsilon\in\Big(0,\frac{1}{4}\Big),\ \ N\geq N_0.
\end{equation*}
\end{prop}
So, RHP \ref{ratio} is a small norm problem and thus the general theory of \cite{DZ} applies to its solvability, for $N\geq N_0$.
\begin{theo}\label{newby1} For every $s>-\frac{1}{2}$ and $z>0$, there exist $c=c(s,z)>0$ and $N_0=N_0(s,z)\in\mathbb{N}$ so that RHP \ref{ratio} is uniquely solvable for all $N\geq N_0$ and any $\epsilon\in(0,\frac{1}{4})$. Its solution admits the integral representation
\begin{equation}\label{e67}
	R(\lambda)=I+\frac{1}{2\pi\im}\int_{\Sigma_R}R_-(\zeta)\big(G_R(\zeta)-I\big)\frac{\d\zeta}{\zeta-\lambda},\ \ \ \lambda\in\mathbb{C}\setminus\Sigma_R,
\end{equation}
where
\begin{equation}\label{e68}
	\|R_-(\cdot;z,s,N)-I\|_{L^2(\Sigma_R)}\leq\frac{c}{N}\ \ \ \ \forall\,N\geq N_0,\ \ \ \epsilon\in\Big(0,\frac{1}{4}\Big).
\end{equation}
\end{theo}
We have now completed the asymptotic legwork for RHP \ref{RHP3} and are left to harvest our results.

\subsection{Distilling asymptotics} From \eqref{e49},\eqref{e54},\eqref{e57},\eqref{e66} and \eqref{e67}  results the exact chain of identities
\begin{align*}
	u_N(z/N;s)=&\,\,N(N+2s)+o_1^{11}(z/N,s,N)=N(N+2s)+N\big(4r_1(z,s,N)-N\big)\\
	=&\,\,2sN+4N\lim_{\lambda\rightarrow\infty}\lambda\big(R(\lambda)P(\lambda)-I\big)^{11}=2sN+4N\big(P_1^{11}+R_1^{11}\big),
\end{align*}
with $P_1=P_1(z_N,s)$ written in condition $(4)$ of RHP \ref{outerRHP} and $R_1=R_1(z,s,N)$ equal to the integral
\begin{equation}\label{e68x}
	R_1=\frac{\im}{2\pi}\int_{\Sigma_R}R_-(\zeta)\big(G_R(\zeta)-I\big)\d\zeta.
\end{equation}
Evidently, for fixed $z>0$,
\begin{equation*}
	P_1^{11}=-\frac{s}{2}\frac{\sqrt{1+z_N}}{\sqrt{1+z_N}+\sqrt{z_N}}=-\frac{s}{2}\bigg(1-\frac{\sqrt{z}}{2N}+\frac{z}{4N^2}+\mathcal{O}\big(N^{-3}\big)\bigg),\ \ N\rightarrow\infty,
\end{equation*}
and we have thus obtained
\begin{lem} As $N\rightarrow\infty$,
\begin{equation}\label{e69}
	u_N(z/N;s)=s\sqrt{z}+4NR_1^{11}+\mathcal{O}\big(N^{-1}\big),
\end{equation}
for any $s>-\frac{1}{2},z>0$ on compact subsets. See \eqref{e68x} for the matrix-valued coefficient $R_1=R_1(z,s,N)$.
\end{lem}
Noting that by \eqref{e68}, Proposition \ref{norm1} and \eqref{e62},\eqref{e65}, through Cauchy-Schwarz inequality,
\begin{equation*}
	R_1=\frac{\im}{2\pi}\oint_{\partial\mathbb{D}_{\epsilon}(0)}\big(G_R(\zeta)-I\big)\d\zeta+\frac{\im}{2\pi}\oint_{\partial\mathbb{D}_{\epsilon}(1)}\big(G_R(\zeta)-I\big)\d\zeta+\mathcal{O}\big(N^{-2}\big),
\end{equation*}
we now utilise \eqref{e62} and \eqref{e65} to evaluate the remaining contour integrals over $\partial\mathbb{D}_{\epsilon}(0)\cup\partial\mathbb{D}_{\epsilon}(1)$, as $N\rightarrow\infty$.
\begin{prop} As $N\rightarrow\infty$, uniformly in $z>0$ and $s>-\frac{1}{2}$ chosen on compact subsets,
\begin{align}\label{e70}
	\frac{\im}{2\pi}\oint_{\partial\mathbb{D}_{\epsilon}(0)}&\,\big(G_R(\zeta)-I\big)\d\zeta=\frac{\im}{2\pi}\oint_{\partial\mathbb{D}_{\epsilon}(0)}U_N(\zeta;z_N,s)\d\zeta+\mathcal{O}\big(N^{-2}\big)\nonumber\\
	=&\,\,\frac{\im}{4N}\sqrt{\frac{z}{\beta}}\Big(Q_1^{21}\Big(2\sqrt{\frac{z}{\beta}},s\Big)-\im\sqrt{\beta}s\Big)\mathcal{D}(\infty)^{\sigma_3}\e^{\im\frac{\pi}{4}\sigma_3}\begin{bmatrix}-1&-1\\ 1 & 1\end{bmatrix}\e^{-\im\frac{\pi}{4}\sigma_3}\mathcal{D}(\infty)^{-\sigma_3}+\mathcal{O}\big(N^{-2}\big).
\end{align}
\end{prop}
\begin{proof} Suppose $\lambda<0$. Then as $\omega:=-\lambda\downarrow 0$, we have the Puiseux expansion
\begin{equation*}
	\mathcal{D}(\lambda)\stackrel{\eqref{Sze}}{=}(\omega z_N)^{\frac{s}{2}}\bigg[1+s\big(\sqrt{1+z_N}-2\sqrt{z_N}\big)\sqrt{\frac{\omega}{z_N}}-\frac{s}{2}\big(4s\sqrt{z_N(1+z_N)}-s(1+5z_N)+1\big)\frac{\omega}{z_N}+\mathcal{O}\big(\omega^{\frac{3}{2}}\big)\bigg],
\end{equation*}
from which it follows that, in the same limit $\omega=-\lambda\downarrow 0$,
\begin{align*}
	P(\lambda;z_N,s)\omega^{\frac{s}{2}\sigma_3}=\mathcal{D}(\infty)^{\sigma_3}\frac{1}{2}&\,\e^{\im\frac{\pi}{4}\sigma_3}\bigg\{\frac{\iota(\omega)}{\sqrt[4]{\omega}}\begin{bmatrix}1&-1\\ -1 & 1\end{bmatrix}+\frac{\sqrt[4]{\omega}}{\iota(\omega)}\begin{bmatrix}1 & 1\\ 1 & 1\end{bmatrix}\bigg\}z_N^{-\frac{s}{2}\sigma_3}\bigg\{I-s\sigma_3\big(\sqrt{1+z_N}-2\sqrt{z_N}\big)\\
	&\,\times\sqrt{\frac{\omega}{z_N}}+\frac{s}{2}\Big(\big((1+5z_N)s-4s\sqrt{z_N(1+z_N)}\big)I+\sigma_3\Big)\frac{\omega}{z_N}+\mathcal{O}\big(\omega^{\frac{3}{2}}\big)\bigg\}\e^{-\im\frac{\pi}{4}\sigma_3},
\end{align*}
where $\omega\mapsto\iota(\omega):=(1+\omega)^{\frac{1}{4}}$ is analytic and non-zero at $\omega=0$. Next, for $\zeta\in(-\beta,0)$,
\begin{equation*}
	\lim_{\epsilon\downarrow 0}\int_{-\beta}^0\frac{\im\zeta^{\frac{1}{2}}\d x}{\sqrt{-x}(x-(\zeta+\im\epsilon))}=-\im\pi-\ln\bigg(\frac{\sqrt{\beta}+\sqrt{-\zeta}}{\sqrt{\beta}-\sqrt{-\zeta}}\bigg)=-\im\pi- 2\sqrt{\frac{\varsigma}{\beta}}+\mathcal{O}\big(\varsigma^{\frac{3}{2}}\big),\ \ \ \varsigma=-\zeta\downarrow 0,
\end{equation*}
where $\mathbb{C}\setminus[0,\infty)\ni\zeta\mapsto\zeta^{\frac{1}{2}}$ is defined with $\textnormal{arg}\,\zeta\in(0,2\pi)$. Consequently,
\begin{equation*}
	E_{0+}(\zeta;s)=\big(\sqrt[4]{\varsigma}\big)^{\sigma_3}\e^{\im\frac{\pi}{4}\sigma_3}\frac{1}{\sqrt{2}}\begin{bmatrix}1&1\\ -1 & 1\end{bmatrix}\e^{-\im\frac{\pi}{4}\sigma_3}\e^{-\im\frac{\pi}{2}s\sigma_3}\bigg\{I-s\sigma_3\sqrt{\frac{\varsigma}{\beta}}+\frac{Is^2 }{2\beta}\varsigma +\mathcal{O}\big(\varsigma^{\frac{3}{2}}\big)\bigg\},\ \ \ \varsigma=-\zeta\downarrow 0,
\end{equation*}
and thus all together, utilising analyticity of $\mathbb{D}_{\epsilon}(0)\ni\lambda\mapsto E(\lambda)$,
\begin{equation*}
	\widehat{E}_N(\lambda)=\mathcal{D}(\infty)^{\sigma_3}\e^{\im\frac{\pi}{4}\sigma_3}\bigg\{\frac{1}{\sqrt{2}}\begin{bmatrix}\kappa_1+\kappa_2 & -\kappa_1^{-1}\smallskip\\ \kappa_1-\kappa_2 & \kappa_1^{-1}\end{bmatrix}+\mathcal{O}(\lambda)\Bigg\}\e^{-\im\frac{\pi}{4}\sigma_3},\ \ \ \lambda\rightarrow 0,
\end{equation*}
that depends on
\begin{align*}
	\kappa_1:=&\,\sqrt{N}\sqrt[4]{\frac{\tau(-z_N)}{16\beta}}\stackrel{N\rightarrow\infty}{=}\sqrt[4]{\frac{z}{4\beta}}\Big(1+\mathcal{O}\big(N^{-2}\big)\Big),\\
	\kappa_2:=&\,-s\Big(\big(\sqrt{1+z_N}-2\sqrt{z_N}\big)\frac{\kappa_1}{\sqrt{z_N}}-\frac{N}{\sqrt{\beta}}\kappa_1^{-1}\Big)\stackrel{N\rightarrow\infty}{=}2s\sqrt[4]{\frac{z}{4\beta}}\Big(1+\mathcal{O}\big(N^{-1}\big)\Big).
\end{align*}
At this point we are ready to evaluate the left hand side in \eqref{e70} by residue theorem,
\begin{align*}
	\frac{\im}{2\pi}\oint_{\partial\mathbb{D}_{\epsilon}(0)}&\,U_N(\zeta;z_N,s)\d\zeta\stackrel{\eqref{e62}}{=}\frac{N\tau(-z_N)}{16\alpha\beta}\Big(Q_1^{21}\big(\gamma N\sqrt{\tau(-z_N)},s\big)-\im\sqrt{\beta}s\Big)\widehat{E}_N(0)\begin{bmatrix}0&0\\ 1 & 0\end{bmatrix}\widehat{E}_N(0)^{-1}\\
	=&\,\frac{\im\kappa_1^2}{2N}\Big(Q_1^{21}\big(\gamma N\sqrt{\tau(-z_N)},s\big)-\im\sqrt{\beta}s\Big)\mathcal{D}(\infty)^{\sigma_3}\e^{\im\frac{\pi}{4}\sigma_3}\begin{bmatrix}-1&-1\\ 1 & 1\end{bmatrix}\e^{-\im\frac{\pi}{4}\sigma_3}\mathcal{D}(\infty)^{-\sigma_3},
\end{align*}
and using the large $N$-expansion,
\begin{equation*}
	 N\sqrt{\tau(-z_N)/\beta}=2\sqrt{\frac{z}{\beta}}\Big(1+\mathcal{O}\big(N^{-2}\big)\Big),
\end{equation*}
the derivation of \eqref{e70} is complete.
\end{proof}

\begin{prop} As $N\rightarrow\infty$, uniformly in $z>0$ and $s>-\frac{1}{2}$ chosen on compact subsets,
\begin{align}\label{e71}
	\frac{\im}{2\pi}\oint_{\partial\mathbb{D}_{\epsilon}(1)}&\,\big(G_R(\zeta)-I\big)\d\zeta=\frac{\im}{2\pi}\oint_{\partial\mathbb{D}_{\epsilon}(1)}V_N(\zeta;z_N,s)\d\zeta+\mathcal{O}\big(N^{-2}\big)\nonumber\\
	=&\,\,\mathcal{D}(\infty)^{\sigma_3}\e^{\im\frac{\pi}{4}\sigma_3}\frac{1}{192N}\begin{bmatrix}3(1-16s^2) & 11+48s+48s^2\\ -11+48s-48s^2 & 3(16s^2-1)\end{bmatrix}\e^{-\im\frac{\pi}{4}\sigma_3}\mathcal{D}(\infty)^{-\sigma_3}+\mathcal{O}\big(N^{-2}\big).
\end{align}
\end{prop}
\begin{proof} Suppose $\lambda>1$. Then as $\omega:=\lambda-1\downarrow 0$, we have the Puiseux expansion
\begin{align*}
	&\mathcal{D}(\lambda)=(1+z_N)^{\frac{s}{2}}\bigg[1+s\big(\sqrt{z_N}-2\sqrt{1+z_N}\big)\sqrt{\frac{\omega}{1+z_N}}-\frac{s}{2}\big(4s\sqrt{z_N(1+z_N)}-(5z_N+4)s-2-z_N\big)\frac{\omega}{1+z_N}\\
	&\,+\frac{s}{6}\Big(\big(s^2(12+13z_N)+3s(2+z_N)-3-z_N\big)\sqrt{\frac{z_N}{1+z_N}}-2s^2(4+7z_N)-6s(2+z_N)+2(1+z_N)\Big)\frac{\omega^{\frac{3}{2}}}{1+z_N}\\
	&\,\ \ \ \ \ \ \ +\mathcal{O}\big(\omega^2\big)\bigg],
\end{align*}
which yields, as $\omega=\lambda-1\downarrow 0$,
\begin{align*}
	P(\lambda;z_N,s)=&\,\,\mathcal{D}(\infty)^{\sigma_3}\frac{1}{2}\e^{\im\frac{\pi}{4}\sigma_3}\bigg\{\frac{\iota(\omega)}{\sqrt[4]{\omega}}\begin{bmatrix}1&1\\ 1 & 1\end{bmatrix}+\frac{\sqrt[4]{\omega}}{\iota(\omega)}\begin{bmatrix}1&-1\\ -1 & 1\end{bmatrix}\bigg\}(1+z_N)^{-\frac{s}{2}\sigma_3}\bigg\{I-s\sigma_3\big(\sqrt{z_N}-2\sqrt{1+z_N}\big)\\
	&\,\times\sqrt{\frac{\omega}{1+z_N}}+\frac{s}{2}\Big(\big(-4s\sqrt{z_N(1+z_N)}+(4+5z_N)s\big)I-(2+z_N)\sigma_3\Big)\frac{\omega}{1+z_N}\\
	&\,\,\,-\frac{s}{6}\bigg(-3s\bigg[-2(2+z_N)+(2+z_N)\sqrt{\frac{z_N}{1+z_N}}\,\bigg] I+\bigg[\big(s^2(12+13z_N)-(3+z_N)\big)\sqrt{\frac{z_N}{1+z_N}}\\
	&\,\,\,\,\,-2s^2(4+7z_N)+2(1+z_N)\bigg]\sigma_3\bigg)\frac{\omega^{\frac{3}{2}}}{1+z_N}+\mathcal{O}\big(\omega^2\big)\bigg\}\e^{-\im\frac{\pi}{4}\sigma_3},
\end{align*}
with the previous $\omega\mapsto\iota(\omega)=(1+\omega)^{\frac{1}{4}}$ that is analytic and non-vanishing at $\omega=0$. What results from the above is the leading order Taylor expansion,
\begin{align*}
	\widehat{F}_N(\lambda)=\mathcal{D}(\infty)^{\sigma_3}\e^{\im\frac{\pi}{4}\sigma_3}\Bigg\{\frac{1}{\sqrt{2}}&\,\,\begin{bmatrix}\iota(\omega)\tau(\omega) & -(\iota(\omega)^{-1}+\iota(\omega)a_N)\tau(\omega)^{-1}\smallskip\\ \iota(\omega)\tau(\omega) & (\iota(\omega)^{-1}-\iota(\omega)a_N)\tau(\omega)^{-1}\end{bmatrix}\\
	&\hspace{1cm}+\frac{\omega}{\sqrt{2}}\begin{bmatrix}\sqrt[6]{2}(a_N+b_N) & -\frac{1}{\sqrt[6]{2}}(b_N+c_N)\smallskip\\ -\sqrt[6]{2}(a_N-b_N) & \frac{1}{\sqrt[6]{2}}(b_N-c_N)\end{bmatrix}+\mathcal{O}\big(\omega^2\big)\Bigg\},
\end{align*}
where we utilise the shorthands
\begin{equation*}
	\tau(\omega):=\sqrt[4]{\frac{\zeta(\lambda;N)}{\omega N^{\frac{2}{3}}}}\stackrel{\omega\rightarrow 0}{=}\sqrt[6]{2}\Big(1-\frac{\omega}{20}+\mathcal{O}\big(\omega^2\big)\Big),\ \ \ \ a_N:=s\bigg(2-\sqrt{\frac{z_N}{1+z_N}}\,\bigg)\stackrel{N\rightarrow\infty}{=}2s+\mathcal{O}\big(N^{-1}\big),
\end{equation*}
\begin{equation*}
	b_N:=\frac{s^2}{2}\big(-4\sqrt{z_N(1+z_N)}+4+5z_N\big)\frac{1}{1+z_N}\stackrel{N\rightarrow\infty}{=}2s^2+\mathcal{O}\big(N^{-1}\big),
\end{equation*}
\begin{equation*}
	c_N:=-\frac{s}{6}\bigg(\big(s^2(12+13z_N)-(3+z_N)\big)\sqrt{\frac{z_N}{1+z_N}}-2s^2(4+7z_N)+2(1+z_N)\bigg)\stackrel{N\rightarrow\infty}{\sim}\frac{s}{3}\big(4s^2-1\big).
\end{equation*}
Written in a slightly more explicit fashion, as $\omega=\lambda-1\rightarrow 0$,
\begin{align*}
	\widehat{F}_N(\lambda)=\mathcal{D}(\infty)^{\sigma_3}\e^{\im\frac{\pi}{4}\sigma_3}\Bigg\{\frac{1}{\sqrt{2}}&\,\,\begin{bmatrix}\sqrt[6]{2} & -\frac{1}{\sqrt[6]{2}}(1+a_N)\smallskip\\ \sqrt[6]{2} & \frac{1}{\sqrt[6]{2}}(1-a_N)\end{bmatrix}\\
	&\hspace{1cm}+\frac{\omega}{\sqrt{2}}\begin{bmatrix}\sqrt[6]{2}(a_N+b_N+\frac{1}{5}) & -\frac{1}{\sqrt[6]{2}}(b_N+c_N-\frac{1}{5}+\frac{3}{10}a_N)\smallskip\\ -\sqrt[6]{2}(a_N-b_N-\frac{1}{5}) & \frac{1}{\sqrt[6]{2}}(b_N-c_N-\frac{1}{5}-\frac{3}{10}a_N)\end{bmatrix}+\mathcal{O}\big(\omega^2\big)\Bigg\},
\end{align*}
which allows us to calculate the left-hand side in \eqref{e71} via residue theorem. Namely,
\begin{align*}
	&\frac{\im}{2\pi}\oint_{\partial\mathbb{D}_{\epsilon}(1)}V_N(\zeta;z_N,s)\d\zeta=-\frac{7N^{-\frac{1}{3}}}{48(2N)^{\frac{2}{3}}}\widehat{F}_N(0)\begin{bmatrix}0&0\\ 1&0\end{bmatrix}\widehat{F}_N(0)^{-1}+\frac{N^{\frac{1}{3}}}{24(2N)^{\frac{4}{3}}}\widehat{F}_N(0)\begin{bmatrix}0&1\\ 0 & 0\end{bmatrix}\widehat{F}_N(0)^{-1}\\
	&\,\,+\frac{\im}{2\pi}\oint_{\partial\mathbb{D}_{\epsilon}(1)}\frac{5N^{\frac{1}{3}}}{48(2N)^{\frac{4}{3}}}\widehat{F}_N(\lambda)\begin{bmatrix}0&1\\ 0 & 0\end{bmatrix}\widehat{F}_N(\lambda)^{-1}\frac{\d\lambda}{(\lambda-1)^2}=-\frac{7N^{-\frac{1}{3}}}{48(2N)^{\frac{2}{3}}}\widehat{F}_N(0)\begin{bmatrix}0&0\\ 1&0\end{bmatrix}\widehat{F}_N(0)^{-1}+\frac{N^{\frac{1}{3}}}{24(2N)^{\frac{4}{3}}}\\
	&\,\,\times\widehat{F}_N(0)\begin{bmatrix}0&1\\ 0 & 0\end{bmatrix}\widehat{F}_N(0)^{-1}+\frac{5N^{\frac{1}{3}}}{48(2N)^{\frac{4}{3}}}\Bigg\{\widehat{F}_N'(0)\begin{bmatrix}0&1\\ 0 & 0\end{bmatrix}\widehat{F}_N(0)^{-1}-\widehat{F}_N(0)\begin{bmatrix}0&1\\ 0 & 0\end{bmatrix}\widehat{F}_N(0)^{-1}\widehat{F}_N'(0)\widehat{F}_N(0)^{-1}\Bigg\},
\end{align*}
and thus after simplification,
\begin{align*}
	\frac{\im}{2\pi}\oint_{\partial\mathbb{D}_{\epsilon}(1)}V_N(\zeta;z_N,s)&\d\zeta=\mathcal{D}(\infty)^{\sigma_3}\e^{\im\frac{\pi}{4}\sigma_3}\Bigg\{-\frac{7}{192N}\begin{bmatrix}a_N^2-1 & -(1+a_N)^2\smallskip\\ (1-a_N)^2 & 1-a_N^2\end{bmatrix}\\
	&\,\,+\frac{1}{96N}\begin{bmatrix}-1 & 1\\ -1 & 1\end{bmatrix}+\frac{5}{192N}\begin{bmatrix}-a_N^2-\frac{2}{5} & a_N^2+2a_N+\frac{2}{5}\smallskip\\ -a_N^2+2a_N-\frac{2}{5} & a_N^2+\frac{2}{5}\end{bmatrix}\Bigg\}\e^{-\im\frac{\pi}{4}\sigma_3}\mathcal{D}(\infty)^{-\sigma_3}.
\end{align*}
The right hand side in \eqref{e71} follows after insertion of the large $N$-expansion of $a_N$.
\end{proof}

We are now ready to prove Proposition \ref{asympformulaforlargeN}.
\begin{proof}[Proof of Proposition \ref{asympformulaforlargeN}] Insert \eqref{e70} and \eqref{e71} into \eqref{e69}, this yields the leading order in \eqref{e72} with error estimate as written. To prove differentiability of the error term, we first return to \eqref{e69} where $\mathcal{O}(N^{-1})$ is clearly twice $z$-differentiable. On the other hand, the identities
\begin{align*}
	\big(\partial_zR(\lambda)\big)R(\lambda)^{-1}=&\,\,\frac{1}{2\pi\im}\int_{\Sigma_R}R_-(\zeta)\big(\partial_zG_R(\zeta)\big)G_R(\zeta)^{-1}R_-(\zeta)^{-1}\frac{\d\zeta}{\zeta-\lambda};\\
	\big(\partial_{zz}R(\lambda)\big)R(\lambda)^{-1}=&\,\,\frac{1}{2\pi\im}\int_{\Sigma_R}\Big(2\big(\partial_zR_-(\zeta)\big)\big(\partial_zG_R(\zeta)\big)+R_-(\zeta)\big(\partial_z^2G_R(\zeta)\big)\Big)G_R(\zeta)^{-1}R_-(\zeta)^{-1}\frac{\d\zeta}{\zeta-\lambda};
\end{align*}
valid for $\lambda\in\mathbb{C}\setminus\Sigma_R$, tell us that the error term in the calculation of $R_1$ in \eqref{e68x} is also twice $z$-differentiable since
\begin{equation*}
	\partial_zG_R(\lambda)=\mathcal{O}\big(N^{-1}\big),\ \ \ \partial_z^2G_R(\lambda)=\mathcal{O}\big(N^{-1}\big),\ \ \ \lambda\in\Sigma_R,
\end{equation*}
by our explicit constructions \eqref{e61} and \eqref{e64}. Together, $\mathcal{O}(N^{-1})$ in \eqref{e72} is twice $z$-differentiable.
\end{proof}

We now focus on the large $N$-limit of
\begin{equation}\label{defofvN}
v_N(z;s):=u_N(z/N;s)-z/2,\ \ \ z>0,\ \ s>-\frac{1}{2}.
\end{equation}

\begin{cor}\label{cheat} Abbreviate
\begin{equation}\label{e73}
	v(z;s):=-\im\sqrt{\frac{z}{\beta}}\,Q_1^{21}\Big(2\sqrt{\frac{z}{\beta}},s\Big)+\frac{1}{16}-s^2-\frac{z}{2},\ \ \ \ z>0,\ s>-\frac{1}{2},
\end{equation}
making use of the leading terms in \eqref{e72}. Then $v=v(z;s)$ satisfies \eqref{e42}.
\end{cor}
\begin{proof} Let $v_N(z;s)$ be as in (\ref{defofvN}).  Then \eqref{e72} says
\begin{equation}\label{e73x}
	v_N(z;s)=v(z;s)+\mathcal{O}\big(N^{-1}\big),
\end{equation}
as $N\rightarrow\infty$ uniformly on compact subsets of $(0,\infty)\ni z$ with fixed $s>-\frac{1}{2}$ and twice $z$-differentiable error term. On the other hand, \eqref{e72} rigorously establishes 
\begin{equation*}
	(zv_N'')^2=\big(4(v_N')^2-1\big)\big(v_N+s^2-zv_N'\big)+s^2+\mathcal{O}\big(N^{-1}\big),\ \ \ \ (')=\frac{\d}{\d z},\ \ \ \ ('')=\frac{\d^2}{\d z^2}.
\end{equation*}
Inserting here \eqref{e73x}, and taking afterwards $N\rightarrow\infty$, results in \eqref{e42} for $v=v(z;s)$.
\end{proof}


\begin{rem}\label{sospecial} By Remark \ref{impcons} we have the exact expression
\begin{equation*}
	v(z;0)=-\frac{z}{2},\ \ z>0,
\end{equation*}
and so necessarily $u_N(z/N;0)=\mathcal{O}(N^{-1})$ as $N\rightarrow\infty$ uniformly for $z>0$ on compact sets. This is consistent with the fact that $u_N(z;0)\equiv 0$ for $z>0$ which can be deduced directly from \eqref{1102defu} and \eqref{NtimesNHankel}.
\end{rem}

\section{Large $x$-asymptotics of RHP \ref{PIIImodel}}\label{sectiononlargex}
The workings below are inspired by those in \cite[Section $5$]{XZ} and we will now prove \eqref{cor5555}.  Suppose $Q(\zeta)=Q(\zeta;x,s)\in\mathbb{C}^{2\times 2}$ denotes the unique solution of RHP \ref{PIIImodel}, with $\beta>0$ fixed (and suppressed from our notations). We begin with the transformation
\begin{equation}\label{f1}
	H(\zeta;x,s):=Q(\zeta;x,s)\e^{\im\varpi(\zeta;x)\sigma_3},\ \ \ \zeta\in\mathbb{C}\setminus\Sigma_Q,
\end{equation}
that leads to the below localised, as $x\rightarrow+\infty$, RHP.
\begin{problem}\label{traf1} Let $x\in(0,\infty)$ and $s>-\frac{1}{2}$. The function $H(\zeta)=H(\zeta;x,s)\in\mathbb{C}^{2\times 2}$ defined in \eqref{f1} has the following properties:
\begin{enumerate}
	\item[(1)] $\zeta\mapsto H(\zeta)$ is analytic for $\zeta\in\mathbb{C}\setminus\Sigma_Q$ with $\Sigma_Q$ as shown in Figure \ref{figa}. On $\Sigma_Q\setminus\{-\beta,0\}$, $H(\zeta)$ admits continuous limiting values $H_{\pm}(\zeta)$ as we approach $\Sigma_Q\setminus\{-\beta,0\}$ from either side of $\mathbb{C}\setminus\Sigma_Q$.
	\item[(2)] The limiting values $H_{\pm}(\zeta)$ on $\Sigma_Q\setminus\{-\beta,0\}$ satisfy $H_+(\zeta)=H_-(\zeta)G_H(\zeta)$ with $G_H(\zeta)=G_H(\zeta;x,s)$ given by
	\begin{align*}
		G_H(\zeta)=\e^{-\im\pi s\sigma_3},\ \ \zeta\in\Sigma_1;&\hspace{1cm} G_H(\zeta)=\begin{bmatrix}0&1\\ -1 & 0\end{bmatrix},\ \ \zeta\in\Sigma_3;\\
		G_H(\zeta)=\begin{bmatrix}1&0\\ \e^{2\pi\im s+2\im\varpi(\zeta;x)} & 1\end{bmatrix},\ \ \zeta\in\Sigma_2;&\hspace{1cm} G_H(\zeta)=\begin{bmatrix}1&0\\ \e^{-2\pi\im s+2\im\varpi(\zeta;x)} & 1\end{bmatrix},\ \ \zeta\in\Sigma_4.
	\end{align*}
	\item[(3)] Near $\zeta=-\beta$, $\zeta\mapsto H(\zeta)$ is weakly singular in that, with some $\zeta\mapsto\widehat{H}(\zeta)=\widehat{H}(\zeta;x,s)$ analytic at $\zeta=-\beta$ and non-vanishing, $H(\zeta)$ is of the form
	\begin{equation*}
		H(\zeta)=\widehat{H}(\zeta)(\zeta+\beta)^{\frac{s}{2}\sigma_3},\ \ \ \ \ \textnormal{arg}(\zeta+\beta)\in(0,2\pi).
	\end{equation*}
	\item[(4)] Near $\zeta=0$, $\zeta\mapsto H(\zeta)$ is weakly singular in that, with some $\zeta\mapsto\widecheck{H}(\zeta)=\widecheck{H}(\zeta;x,s)$ analytic at $\zeta=0$ and non-vanishing, $H(\zeta)$ is of the form
	\begin{equation*}
		H(\zeta)=\widecheck{H}(\zeta)\zeta^{\frac{s}{2}\sigma_3}\mathcal{M}_j(\zeta)\e^{\im\varpi(\zeta;x)\sigma_3},\ \ \ \zeta\in\Omega_j\cap\mathbb{D}_{\epsilon}(0),\ \ \ \ \ \ \ \ \textnormal{arg}\,\zeta\in(0,2\pi)
	\end{equation*}
	with $\mathcal{M}_j(\zeta)$ as in condition $(4)$ of RHP \ref{PIIImodel}.
	\item[(5)] As $\zeta\rightarrow\infty$ and $\zeta\notin\Sigma_Q$, $H(\zeta)$ is normalised as follows,
	\begin{equation*}
		H(\zeta)=\Bigg\{I+\sum_{k=1}^2Q_k(x,s)\zeta^{-k}+\mathcal{O}\big(\zeta^{-3}\big)\Bigg\}\zeta^{\frac{1}{4}\sigma_3}\frac{1}{\sqrt{2}}\begin{bmatrix}1&1\\ -1 & 1\end{bmatrix}\e^{-\im\frac{\pi}{4}\sigma_3},
	\end{equation*}
	with $Q_k(x,s)$ as in \eqref{A4} and fractional powers with cut on $[0,\infty)\subset\mathbb{R}$ and branch fixed by $\textnormal{arg}\,\zeta\in(0,2\pi)$.
\end{enumerate}
\end{problem}
Seeing that $\textnormal{Re}(\pm 2\pi\im s+2\im\varpi(\zeta;x))\rightarrow-\infty$ as $x\rightarrow+\infty$, uniformly for $\zeta\in\Sigma_{2,4}\setminus\mathbb{D}_{\epsilon}(0)$ with $\epsilon>0$ fixed, we are at once led to the below local constructions.
\subsection{The outer parametrix} We work with, for $\zeta\in\mathbb{C}\setminus[-\beta,\infty)$ and $\textnormal{arg}\,\zeta\in(0,2\pi)$,
\begin{equation}\label{f2}
	L(\zeta;s):=\begin{bmatrix}1 & -\im s\sqrt{\beta}\\ 0 & 1\end{bmatrix}\zeta^{\frac{1}{4}\sigma_3}\frac{1}{\sqrt{2}}\begin{bmatrix}1&1\\ -1 & 1\end{bmatrix}\e^{-\im\frac{\pi}{4}\sigma_3}\exp\left[\im\zeta^{\frac{1}{2}}\int_{-\beta}^0\frac{s\sigma_3}{2\sqrt{-x}}\frac{\d x}{x-\zeta}\right],
\end{equation}
that previously appeared in \eqref{e61}.
\begin{problem}[Outer parametrix] Let $s>-\frac{1}{2}$. The model function $L(\zeta)=L(\zeta;s)\in\mathbb{C}^{2\times 2}$ defined in \eqref{f2} has the following properties:
\begin{enumerate}
	\item[(1)] $L(\zeta)$ is analytic for $\zeta\in\mathbb{C}\setminus[-\beta,\infty)$ and extends continuously to $\mathbb{C}\setminus\{-\beta,0\}$.
	\item[(2)] On $(-\beta,0)\cup(0,\infty)\ni\zeta$ the boundary values $L_{\pm}(\zeta)=\lim_{\epsilon\downarrow 0}L(\zeta\pm\im\epsilon)$ obey
	\begin{equation*}
		L_+(\zeta)=L_-(\zeta)\e^{-\im\pi s\sigma_3},\ \ \ \zeta\in(-\beta,0);\hspace{1cm} L_+(\zeta)=L_-(\zeta)\begin{bmatrix}0 & 1\\ -1 & 0\end{bmatrix},\ \ \ \zeta\in(0,\infty).
	\end{equation*}
	\item[(3)] The mapping $(-\beta,0)\cup(0,1)\ni\zeta\mapsto L_{\pm}(\zeta)$ is entry-wise in $L^2(-\beta,1)$.
	\item[(4)] As $\zeta\rightarrow\infty$,
	\begin{equation*}
		L(\zeta)=\Bigg\{I+\frac{1}{\zeta}\begin{bmatrix}\frac{1}{2}s^2\beta & -\frac{\im}{3}\beta^{\frac{3}{2}}(s-s^3)\smallskip\\ \im s\sqrt{\beta} & -\frac{1}{2}s^2\beta\end{bmatrix}+\mathcal{O}\big(\zeta^{-2}\big)\Bigg\}\zeta^{\frac{1}{4}\sigma_3}\frac{1}{\sqrt{2}}\begin{bmatrix}1 & 1\\ -1 & 1\end{bmatrix}\e^{-\im\frac{\pi}{4}\sigma_3},\ \ \ \textnormal{arg}\,\zeta\in(0,2\pi).
	\end{equation*}
\end{enumerate}

\end{problem}
\subsection{The parametrix at $\zeta=0$} Here RHP \ref{BessRHP} is utilised. Specifically, we define (with $0<\epsilon<\beta$)
\begin{equation}\label{f3}
	W(\zeta;x,s):=E_W(\zeta)J\big(x^2\zeta;s)\e^{\im\varpi(\zeta;x)\sigma_3}\begin{cases}\e^{-\im\frac{\pi}{2}s\sigma_3},&\textnormal{arg}\,\zeta\in(0,\pi)\smallskip\\ \e^{\im\frac{\pi}{2}s\sigma_3},&\textnormal{arg}\,\zeta\in(\pi,2\pi)\end{cases},\ \ \ \zeta\in\mathbb{D}_{\epsilon}(0)\setminus\Sigma_Q,
\end{equation}
where $J(\zeta;\nu)$ is as in \eqref{A11}, $\varpi(\zeta;x)=x\zeta^{\frac{1}{2}}$ is defined for $\zeta\in\mathbb{C}\setminus[0,\infty)$ with $\textnormal{arg}\,\zeta\in(0,2\pi)$ and $\mathbb{D}_{\epsilon}(0)\ni\zeta\mapsto E_W(\zeta)$ denotes the analytic multiplier
\begin{equation*}
	E_W(\zeta)=E_W(\zeta;x,s):=L(\zeta;s)\begin{cases}\e^{\im\frac{\pi}{2}s\sigma_3},&\textnormal{arg}\,\zeta\in(0,\pi)\\ \e^{-\im\frac{\pi}{2}s\sigma_3},&\textnormal{arg}\,\zeta\in(\pi,2\pi)\end{cases}\times\e^{\im\frac{\pi}{4}\sigma_3}\frac{1}{\sqrt{2}}\begin{bmatrix}1&-1\\ 1 & 1\end{bmatrix}\big(\e^{-\im\pi}x^2\zeta\big)^{-\frac{1}{4}\sigma_3},
\end{equation*}
with $L(\zeta;s)$ in \eqref{f2}. Here $(\e^{-\im\pi}\zeta)^{\alpha}$ is defined with cut on $[0,\infty)\subset\mathbb{R}$ and $\textnormal{arg}\,\zeta\in(0,2\pi)$. What results for $\zeta\mapsto W(\zeta)$ is summarised below.
\begin{problem}[Origin parametrix] Let $x>0$ and $s>-\frac{1}{2}$. The model function $W(\zeta)=W(\zeta;x,s)\in\mathbb{C}^{2\times 2}$ defined in \eqref{f3} has the following properties:
\begin{enumerate}
	\item[(1)] $\zeta\mapsto W(\zeta)$ is analytic for $\zeta\in\mathbb{C}\setminus(\Sigma_Q\cap\mathbb{D}_{\epsilon}(0))$ with $\Sigma_Q$ shown in Figure \ref{figa}. On $(\Sigma_Q\cap\mathbb{D}_{\epsilon}(0))\setminus\{0\}$, $W(\zeta)$ admits continuous limiting values $W_{\pm}(\zeta)$ as we approach $(\Sigma_Q\cap\mathbb{D}_{\epsilon}(0))\setminus\{0\}$ from either side of $\mathbb{C}\setminus\Sigma_Q$.
	\item[(2)] The limiting values $W_{\pm}(\zeta)$ on $(\Sigma_Q\cap\mathbb{D}_{\epsilon}(0))\setminus\{0\}$ obey $W_+(\zeta)=W_-(\zeta)G_W(\zeta)$ with $G_W(\zeta)=G_W(\zeta;x,s)$ given by
	\begin{align*}
		G_W(\zeta)=\e^{-\im\pi s\sigma_3},\ \zeta\in\Sigma_1\cap\mathbb{D}_{\epsilon}(0);&\hspace{0.5cm}G_W(\zeta)=\begin{bmatrix}0&1\\ -1 & 0\end{bmatrix},\ \zeta\in\Sigma_3\cap\mathbb{D}_{\epsilon}(0);\\
		G_W(\zeta)=\begin{bmatrix}1 & 0\\ \e^{-2\pi\im s+2\im\varpi(\zeta;x)} & 1\end{bmatrix},\ \zeta\in\Sigma_4\cap\mathbb{D}_{\epsilon}(0);&\hspace{0.5cm}G_W(\zeta)=\begin{bmatrix}1 & 0\\ \e^{2\pi\im s+2\im\varpi(\zeta;s)} & 1\end{bmatrix},\ \zeta\in\Sigma_2\cap\mathbb{D}_{\epsilon}(0).
	\end{align*}
	\item[(3)] Near $\zeta=0$, $\zeta\mapsto W(\zeta)$ is weakly singular in that, with some $\zeta\mapsto\widecheck{W}(\zeta)=\widecheck{W}(\zeta;s,s)$ analytic at $\zeta=0$ and non-vanishing, $W(\zeta)$ is of the form
	\begin{equation*}
		W(\zeta)=\widecheck{W}(\zeta)\zeta^{\frac{s}{2}\sigma_3}\mathcal{M}_j(\zeta),\ \ \ \zeta\in\Omega_j\cap\mathbb{D}_{\epsilon}(0),\ \ \ \ \ \ \ \ \textnormal{arg}\,\zeta\in(0,2\pi),
	\end{equation*}
	with $\mathcal{M}_j(\zeta)$ as in condition $(4)$ of RHP \ref{PIIImodel}.
	\item[(4)] As $x\rightarrow\infty$, $W(\zeta)$ in \eqref{f3} compares to $L(\zeta)$ in \eqref{f2} as follows:
	\begin{equation}\label{f4}
		W(\zeta)=\bigg\{I+W_x(\zeta;s)+\mathcal{O}\big(x^{-2}\big)\bigg\}L(\zeta),
	\end{equation}
	uniformly for $0<\delta_1\leq|\zeta|\leq\delta_2<\frac{\beta}{4}$, for any fixed $s>-\frac{1}{2}$. Here,
	\begin{equation*}
		W_x(\zeta;s)=\frac{1-4s^2}{8x\zeta}\widehat{E}_W(\zeta)\begin{bmatrix}0&0\\ 1 & 0\end{bmatrix}\widehat{E}_W(\zeta)^{-1}=\mathcal{O}\big(x^{-1}\big),
	\end{equation*}
	that we express in terms of the locally analytic
	\begin{equation*}
		\widehat{E}_W(\zeta)=\widehat{E}_W(\zeta;x,s):=E_W(\zeta;s,x)x^{\frac{1}{2}\sigma_3},\ \ \ \zeta\in\mathbb{D}_{\epsilon}(0),
	\end{equation*}
	and where we have used that $\widehat{E}_W(\zeta)=\mathcal{O}(1)$ as $x\rightarrow\infty$ for $0<\delta_1\leq|\zeta|\leq\delta_2<\frac{\beta}{4}$.
\end{enumerate}

\end{problem}
\subsection{The final ratio transformation} We fix $0<\epsilon<\frac{\beta}{4}$ and compare the explicit $L(\zeta),W(\zeta)$ in \eqref{f2},\eqref{f3} to $H(\zeta)$ from RHP \ref{traf1}. This amounts to defining
\begin{equation}\label{f5}
	Z(\zeta;x,s):=H(\zeta;x,s)\begin{cases}W(\zeta;x,s)^{-1},&\zeta\in\mathbb{D}_{\epsilon}(0)\smallskip\\ L(\zeta;s)^{-1},&\zeta\in\mathbb{C}\setminus(\Sigma_Q\cup\overline{\mathbb{D}_{\epsilon}(0)})\end{cases},
\end{equation}
and collecting the properties of the resulting ratio function below.
\begin{problem}[Small norm problem]\label{SN2} Let $s>-\frac{1}{2}$ and $x>0$. The function $Z(\zeta)=Z(\zeta;x,s)\in\mathbb{C}^{2\times 2}$ defined in \eqref{f5} has the following properties:
\begin{enumerate}
	\item[(1)] $\zeta\mapsto Z(\zeta)$ is analytic for $\zeta\in\mathbb{C}\setminus\Sigma_Z$ with $\Sigma_Z$ shown in Figure \ref{fig8}. On $\Sigma_Z$, $Z(\zeta)$ admits continuous boundary values $Z_{\pm}(\zeta)$ as one approaches $\Sigma_Z$ from either side $\mathbb{C}\setminus\Sigma_Z$ in a non-tangential fashion.
	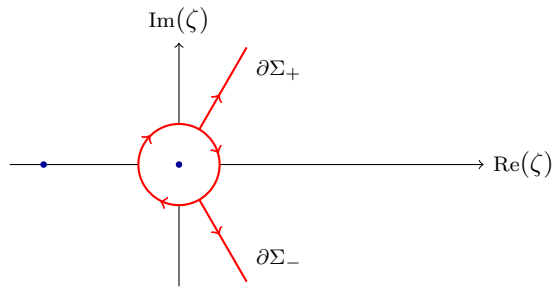
\begin{figure}[tbh]
	\begin{tikzpicture}[xscale=0.9,yscale=0.9]
	        \draw [->] (-5,0) -- (2,0) node [right] {$\footnotesize{\textnormal{Re}(\zeta)}$};
	  \draw [->] (-2.5,-1.8) -- (-2.5,1.8) node [above] {$\footnotesize{\textnormal{Im}(\zeta)}$};
	  \draw [thick,red,decoration={markings,mark= at position 0.6 with {\arrow{>}}},postaction={decorate}] (-2.5,0) -- (-1.5,1.732050);
	\draw [thick,red,decoration={markings,mark= at position 0.6 with {\arrow{>}}},postaction={decorate}] (-2.5,0) -- (-1.5,-1.732050);
	\node [right] at (-1.5,1.4) {{\footnotesize $\partial\Sigma_+$}};
	\node [right] at (-1.5,-1.4) {{\footnotesize $\partial\Sigma_-$}};
	\draw [thick, color=red, fill=white, decoration={markings, mark=at position 0.07 with {\arrow{<}}}, decoration={markings, mark=at position 0.4 with {\arrow{<}}}, decoration={markings, mark=at position 0.7 with {\arrow{<}}},postaction={decorate}] (-2.5,0) circle [radius=0.6];
	\draw[fill, color=blue!60!black](-2.5,0) circle [radius=0.04];
	\draw[fill, color=blue!60!black](-4.5,0) circle [radius=0.04];
\end{tikzpicture}
\caption{The oriented jump contour $\Sigma_Z$, in the complex $\zeta$-plane drawn in \textcolor{red}{red}, for $Z(\zeta)$. The points $\zeta=-\beta,0$ are colored in \textcolor{blue!60!black}{blue}.}
\label{fig8}
\end{figure}

	\item[(2)] On $\Sigma_Z$ we record the jump condition $Z_+(\zeta)=Z_-(\zeta)G_Z(\zeta)$ with $G_Z(\zeta)=G_Z(\zeta;x,s)$ given as
	\begin{equation*}
		G_Z(\zeta)=I+\e^{\mp 2\pi\im s+2\im\varpi(\zeta;x)}L(\zeta)\begin{bmatrix}0 &0\\ 1 & 0\end{bmatrix}L(\zeta)^{-1},\ \ \ \zeta\in\partial\Sigma_{\pm};
	\end{equation*}
	as well as
	\begin{equation*}
		G_Z(\zeta)=W(\zeta)L(\zeta)^{-1},\ \ \ \zeta\in\partial\mathbb{D}_{\epsilon}(0).
	\end{equation*}
	By construction, $Z(\zeta)$ has no jumps inside $\mathbb{D}_{\epsilon}(0)$ and on $(-\beta,\epsilon)\cup(\epsilon,\infty)\subset\mathbb{R}$.
	\item[(3)] As $\zeta\rightarrow\infty$, $Z(\zeta)\rightarrow I$.
\end{enumerate}
\end{problem}
In order to conclude the asymptotic analysis of RHP \ref{SN2} we note that $L(\zeta)=\mathcal{O}(1)$ as $x\rightarrow\infty$ for $\zeta$ away from $\{-\beta,0\}$, and so \eqref{f4} and the previous $\textnormal{Re}(2\im\varpi(\zeta;s))\rightarrow-\infty$ as $x\rightarrow+\infty$ for $\zeta\in\Sigma_{\pm}$ lead to the following small norm estimate.
\begin{prop} Fix $s>-\frac{1}{2}$. There exist $c=c(s)>0$ and $x_0=x_0(s)>0$ so that
\begin{equation*}
	\|G_Z(\cdot;x,s)-I\|_{L^2\cap L^{\infty}(\Sigma_Z)}\leq\frac{c}{x}\ \ \ \ \forall\,x\geq x_0.
\end{equation*}
\end{prop}
Consequently, RHP \ref{SN2} is solvable for $x\geq x_0$ and we summarise our findings below.
\begin{theo}\label{newby2} For every $s>-\frac{1}{2}$, there exist $c=c(s)>0$ and $x_0=x_0(s)>0$ so that RHP \ref{SN2} is uniquely solvable for all $x\geq x_0$ and any $\epsilon\in(0,\frac{\beta}{4})$. Its solution admits the integral representation
\begin{equation*}
	Z(\zeta)=I+\frac{1}{2\pi\im}\int_{\Sigma_Z}Z_-(\mu)\big(G_Z(\mu)-I\big)\frac{\d\mu}{\mu-\zeta},\ \ \ \ \zeta\in\mathbb{C}\setminus\Sigma_Z,
\end{equation*}
where
\begin{equation*}
	\|Z_-(\cdot;x,s)-I\|_{L^2(\Sigma_Z)}\leq\frac{c}{x}\ \ \ \ \ \forall\,x\geq x_0.
\end{equation*}
\end{theo}
We are now left to extract the needed large $x$-asymptotics.
\subsection{Distilling asymptotics} Simply trace back \eqref{f1} and \eqref{f5},
\begin{align*}
	Q_1(x,s)=&\,\,\lim_{\zeta\rightarrow\infty}\zeta\Bigg\{Q(\zeta;x,s)\e^{\im\varpi(\xi;x)\sigma_3}\e^{\im\frac{\pi}{4}\sigma_3}\frac{1}{\sqrt{2}}\begin{bmatrix}1&-1\\ 1 & 1\end{bmatrix}\zeta^{-\frac{1}{4}\sigma_3}-I\Bigg\}\\
	&\hspace{2cm}=\begin{bmatrix}\frac{1}{2}s^2\beta & -\frac{\im}{3}\beta^{\frac{3}{2}}(s-s^3)\smallskip\\ \im s\sqrt{\beta} & -\frac{1}{2}s^2\beta\end{bmatrix}+\frac{\im}{2\pi}\int_{\Sigma_Z}Z_-(\mu)\big(G_Z(\mu)-I\big)\d\mu,
\end{align*}
and note that, through an application of the Cauchy-Schwarz inequality,
\begin{equation*}
	\frac{\im}{2\pi}\int_{\Sigma_Z}Z_-(\mu)\big(G_Z(\mu)-I\big)\d\mu=\frac{\im}{2\pi}\oint_{\partial\mathbb{D}_{\epsilon}(0)}\big(G_Z(\mu)-I\big)\d\mu+\mathcal{O}\big(x^{-2}\big),\ \ \ x\rightarrow+\infty,
\end{equation*}
and so we are left to evaluate the below contour integral, as $x\rightarrow+\infty$.
\begin{prop} As $x\rightarrow+\infty$, uniformly in $s>-\frac{1}{2}$ chosen on compact subsets,
\begin{align*}
	\frac{\im}{2\pi}\oint_{\partial\mathbb{D}_{\epsilon}(0)}\big(G_Z(\mu)-I\big)\d\mu=&\,\,\frac{\im}{2\pi}\oint_{\partial\mathbb{D}_{\epsilon}(0)}W_x(\mu;s)\d\mu+\mathcal{O}\big(x^{-2}\big)\\
	=&\,\,\frac{1-4s^2}{8x}\e^{\im\frac{\pi}{4}\sigma_3}\begin{bmatrix}-s\sqrt{\beta} & -\beta s^2\smallskip\\ 1 & s\sqrt{\beta}\end{bmatrix}\e^{-\im\frac{\pi}{4}\sigma_3}+\mathcal{O}\big(x^{-2}\big).
\end{align*}
\end{prop}
\begin{proof} Since
\begin{equation*}
	\widehat{E}_W(\zeta)=\e^{\im\frac{\pi}{4}\sigma_3}\Bigg\{\begin{bmatrix}1-s^2&-s\sqrt{\beta}\\ \frac{s}{\sqrt{\beta}} & 1\end{bmatrix}+\mathcal{O}(\zeta)\Bigg\},\ \ \ \zeta\rightarrow 0,
\end{equation*}
we obtain by residue theorem
\begin{equation*}
	\frac{\im}{2\pi}\oint_{\partial\mathbb{D}_{\epsilon}(0)}W_x(\mu;s)\d\mu=\frac{1-4s^2}{8x}\widehat{E}_W(0)\begin{bmatrix}0 & 0\\ 1 & 0\end{bmatrix}\widehat{E}_W(0)^{-1}=\frac{1-4s^2}{8x}\e^{\im\frac{\pi}{4}\sigma_3}\begin{bmatrix}-s\sqrt{\beta} & -\beta s^2\smallskip\\ 1 & s\sqrt{\beta}\end{bmatrix}\e^{-\im\frac{\pi}{4}\sigma_3},
\end{equation*}
which concludes our proof.
\end{proof}
\begin{cor} As $x\rightarrow+\infty$, uniformly in $s>-\frac{1}{2}$ on compact subsets,
\begin{equation}\label{f555}
	Q_1(x,s)=\begin{bmatrix}\frac{1}{2}s^2\beta & -\frac{\im}{3}\beta^{\frac{3}{2}}(s-s^3)\smallskip\\ \im s\sqrt{\beta} & -\frac{1}{2}s^2\beta\end{bmatrix}+\frac{1-4s^2}{8x}\e^{\im\frac{\pi}{4}\sigma_3}\begin{bmatrix}-s\sqrt{\beta} & -\beta s^2\smallskip\\ 1 & s\sqrt{\beta}\end{bmatrix}\e^{-\im\frac{\pi}{4}\sigma_3}+\mathcal{O}\big(x^{-2}\big).
\end{equation}
\end{cor}
The above concludes our workings on the large $x$-asymptotics of $Q(\zeta;x,s)$.
\begin{proof}[Proof of Corollary \ref{cor5555}]
This follows directly from (\ref{f555}).
\end{proof}

\begin{rem}\label{remarkaftercorollary49}
Using that $v(z;s)$ solves \eqref{e42}, compare Corollary \ref{cheat}, \eqref{f1000} is sufficient to compute subleading terms, in particular
\begin{equation*}
	v(z;s)=-\frac{z}{2}+s\sqrt{z}-\frac{3}{4}s^2+\frac{s(4s^2-1)}{32\sqrt{z}}+\frac{s^2(4s^2-1)}{64z}+\frac{3s(4s^2-1)(4s^2+3)}{2048z\sqrt{z}}	
	+\mathcal{O}\big(z^{-2}\big),\ \ \ z\rightarrow+\infty,
\end{equation*}
uniformly for $s>-\frac{1}{2}$ on compact sets.
\end{rem}

\section{Small $x$-asymptotics of RHP \ref{PIIImodel}}\label{sectiononsmallx}
We now establish the estimates in Proposition \ref{1102asymoforsmallz}. We draw inspiration from \cite[Section $6$]{XZ}. Suppose $Q(\zeta)=Q(\zeta;x,s)\in\mathbb{C}^{2\times 2}$ denotes the unique solution of RHP \ref{PIIImodel}, with $\beta>0$ fixed (and suppressed from our notations). We begin with the transformation
\begin{equation}\label{f6}
	H(\zeta;x,s):=x^{\frac{1}{2}\sigma_3}Q\big(\zeta x^{-2};x,s\big),\ \ \ \zeta\in\mathbb{C}\setminus\Sigma_Q,
\end{equation}
that leads to the below RHP.
\begin{problem}\label{traf2} Let $x>0$ and $s>-\frac{1}{2}$. The function $H(\zeta)=H(\zeta;x,s)\in\mathbb{C}^{2\times 2}$ defined in \eqref{f6} has the following properties:
\begin{enumerate}
	\item[(1)] $\zeta\mapsto H(\zeta)$ is analytic for $\zeta\in\mathbb{C}\setminus\Sigma_H$ where $\Sigma_H=\bigcup_{j=1}^4\Gamma_j$ consists of the four oriented rays
	\begin{equation*}
		\Gamma_1:=(-\beta x^2,0)\subset\mathbb{R},\ \ \ \ \ \ \Gamma_2:=\e^{\im\frac{5\pi}{3}}(0,\infty),\ \ \ \ \ \ \Gamma_3:=(0,\infty)\subset\mathbb{R},\ \ \ \ \ \ \Gamma_4:=\e^{\im\frac{\pi}{3}}(0,\infty),
	\end{equation*}
	and is shown in Figure \ref{fig9}. On $\Sigma_H\setminus\{-\beta x^2,0\}$, $H(\zeta)$ admits continuous limiting values $H_{\pm}(\zeta)$ as we approach $\Sigma_H\setminus\{-\beta x^2,0\}$ from either side of $\mathbb{C}\setminus\Sigma_H$.
	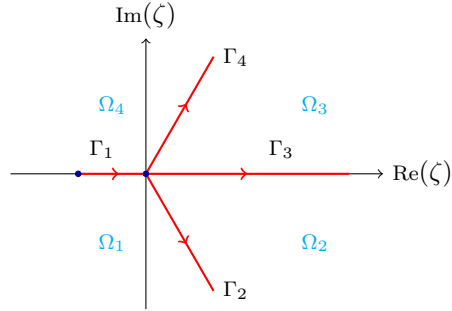
\begin{figure}[tbh]
	\begin{tikzpicture}[xscale=0.9,yscale=0.9]
	\draw [thick,red,decoration={markings,mark= at position 0.6 with {\arrow{>}}},postaction={decorate}] (0,0) -- (1,1.732050);
	\draw [thick,red,decoration={markings,mark= at position 0.6 with {\arrow{>}}},postaction={decorate}] (0,0) -- (1,-1.732050);
	 \draw [->] (-2,0) -- (3.5,0) node [right] {$\footnotesize{\textnormal{Re}(\zeta)}$};
        \draw [thick,red,decoration={markings,mark= at position 0.5 with {\arrow{>}}},postaction={decorate}] (0,0) -- (3,0);
  \draw [->] (0,-2) -- (0,2) node [above] {$\footnotesize{\textnormal{Im}(\zeta)}$};
  	\draw [thick,red,decoration={markings,mark= at position 0.6 with {\arrow{>}}},postaction={decorate}] (-1,0) -- (0,0);
  \draw [fill, color=blue!60!black] (-1,0) circle [radius=0.04];
  \draw [fill, color=blue!60!black] (0,0) circle [radius=0.04];
		\node [right] at (1,1.72) {{\footnotesize $\Gamma_4$}};
		\node [right] at (1,-1.72) {{\footnotesize $\Gamma_2$}};
		\node [above] at (2,0.1) {{\footnotesize $\Gamma_3$}};
		\node [above] at (-0.65,0.1) {{\footnotesize $\Gamma_1$}};
		\node [below] at (-0.5,-0.75) {{\footnotesize\textcolor{cyan}{ $\Omega_1$}}};
		\node [above] at (-0.5,0.75) {{\footnotesize\textcolor{cyan}{ $\Omega_4$}}};
		\node [above] at (2.5,0.75) {{\footnotesize\textcolor{cyan}{ $\Omega_3$}}};
		\node [below] at (2.5,-0.75) {{\footnotesize\textcolor{cyan}{ $\Omega_2$}}};
\end{tikzpicture}
\caption{The oriented jump contour $\Sigma_H$ in the complex $\zeta$-plane. The (end) points $\zeta=0,-\beta x^2$ are colored in \textcolor{blue!60!black}{blue}. The sectors in between in \textcolor{cyan}{cyan}.}
\label{fig9}
\end{figure}

	\item[(2)] The limiting values $H_{\pm}(\zeta)$ on $\Sigma_H\setminus\{-\beta x^2,0\}$ satisfy $H_+(\zeta)=H_-(\zeta)G_H(\zeta)$ with $G_H(\zeta)=G_H(\zeta;s)$ given by
	\begin{align*}
		G_H(\zeta)=\e^{-\im\pi s\sigma_3},\ \ \zeta\in\Gamma_1;&\hspace{1cm} G_H(\zeta)=\begin{bmatrix}0&1\\ -1&0\end{bmatrix},\ \ \zeta\in\Gamma_3;\\
		G_H(\zeta)=\begin{bmatrix}1 & 0 \\ \e^{2\pi\im s}& 1\end{bmatrix},\ \ \zeta\in\Gamma_2;&\hspace{1cm} G_H(\zeta)=\begin{bmatrix}1&0\\ \e^{-2\pi\im s} & 1\end{bmatrix},\ \ \zeta\in\Gamma_4.
	\end{align*}
	\item[(3)] Near $\zeta=-\beta x^2$, $\zeta\mapsto H(\zeta)$ is weakly singular in that, with some $\zeta\mapsto\widehat{H}(\zeta)=\widehat{H}(\zeta;x,s)$ analytic at $\zeta=-\beta x^2$ and non-vanishing, $H(\zeta)$ is of the form
	\begin{equation*}
		H(\zeta)=\widehat{H}(\zeta)(\zeta+\beta x^2)^{\frac{s}{2}\sigma_3},\ \ \ \ \textnormal{arg}(\zeta+\beta x^2)\in(0,2\pi).
	\end{equation*}
	\item[(4)] Near $\zeta=0$, $\zeta\mapsto H(\zeta)$ is weakly singular in that, with some $\zeta\mapsto\widecheck{H}(\zeta)=\widecheck{H}(\zeta;x,s)$ analytic at $\zeta=0$ and non-vanishing, $H(\zeta)$ is of the form
	\begin{equation*}
		H(\zeta)=\widecheck{H}(\zeta)\zeta^{\frac{s}{2}\sigma_3}\mathcal{M}_j(\zeta),\ \ \zeta\in\Omega_j\cap\mathbb{D}_{\epsilon}(0),\ \ \ \ \ \ \textnormal{arg}\,\zeta\in(0,2\pi),
	\end{equation*}
	with $\mathcal{M}_j(\zeta)$ as in condition $(4)$ of RHP \ref{PIIImodel}.
	\item[(5)] As $\zeta\rightarrow\infty$ and $\zeta\notin\Sigma_H$, $H(\zeta)$ is normalized as follows,
	\begin{equation*}
		H(\zeta)=\bigg\{I+\sum_{k=1}^2H_k(x,s)\zeta^{-k}+\mathcal{O}\big(\zeta^{-3}\big)\bigg\}\zeta^{\frac{1}{4}\sigma_3}\frac{1}{\sqrt{2}}\begin{bmatrix}1&1\\ -1&1\end{bmatrix}\e^{-\im\frac{\pi}{4}\sigma_3}\e^{-\im\zeta^{\frac{1}{2}}\sigma_3},
	\end{equation*}
	with $H_k(x,s)=x^{\frac{1}{2}\sigma_3}Q_k(x,s)x^{2k}x^{-\frac{1}{2}\sigma_3}$ independent of $\zeta$ and all fractional powers defined with cut on $[0,\infty)\subset\mathbb{R}$ such that $\textnormal{arg}\,\zeta\in(0,2\pi)$.
\end{enumerate}
\end{problem}
Noticing that $\Gamma_1$ shrinks to $\emptyset$ as $x\downarrow 0$, we are at once led to the below local constructions.
\subsection{The outer parametrix} We work with $\e^{\im\frac{\pi}{4}\sigma_3}J(\zeta;2s)$ as defined in \eqref{A11} for the outer parametrix which approximates $H(\zeta;x,s)$ well for $\zeta\in\Sigma_H\setminus\mathbb{D}_{\epsilon}(0)$ with $\epsilon>0$ fixed. Here we assume $x$ to be sufficiently close to zero so that $-\beta x^2\in\mathbb{D}_{\epsilon}(0)$.
\subsection{The parametrix at $\zeta=0$} We recall \eqref{A11} as well as Remark \ref{nearzero} and define, for $x$ suffciently close to zero so that $-\beta x^2\in\mathbb{D}_{\epsilon}(0)$, assuming first $2s\notin\mathbb{Z}_{\geq 0}$,
\begin{equation}\label{f7}
	W(\zeta;x,s):=E_W(\zeta)\begin{bmatrix}1 &m(\zeta x^{-2})\smallskip \\ 0 & 1\end{bmatrix}\big(\zeta+\beta x^2\big)^{\frac{s}{2}\sigma_3}\big(\e^{-\im\pi}\zeta\big)^{\frac{s}{2}\sigma_3}\begin{bmatrix}1&\frac{\im}{2}\frac{1}{\sin(2\pi s)}\\ 0 & 1\end{bmatrix}\mathcal{N}(\zeta;2s),\ \ \zeta\in\mathbb{D}_{\epsilon}(0)\setminus\Sigma_H.
\end{equation}
Here, $\textnormal{arg}(\zeta+\beta x^2)\in(0,2\pi)$ and $\textnormal{arg}(\e^{-\im\pi}\zeta)\in(-\pi,\pi)$, $m(\zeta)=m(\zeta;x,s)$ denotes the scalar function
\begin{equation*}
	m(\zeta;x,s):=-\frac{x^{4s}}{2\pi\im}\int_{-\beta}^0\frac{\e^{\im\pi s}(-\tau)^s(\tau+\beta)^s}{2\cos(\pi s)}\frac{\d\tau}{\tau-\zeta},\ \ \ \ \zeta\in\mathbb{C}\setminus[-\beta,0],
\end{equation*}
and $\mathbb{D}_{\epsilon}(0)\ni\zeta\mapsto E_W(\zeta)$ is the locally analytic multiplier
\begin{equation}\label{defEWfornonint}
E_W(\zeta)=E_W(\zeta;s):=\e^{\im\frac{\pi}{4}\sigma_3}J(\zeta;2s)\mathcal{N}(\zeta;2s)^{-1}\begin{bmatrix}1 & -\frac{\im}{2}\frac{1}{\sin(2\pi s)}\\ 0 & 1\end{bmatrix}\big(\e^{-\im\pi}\zeta\big)^{-s\sigma_3}\e^{-\im\frac{\pi}{2}s\sigma_3},
\end{equation}
with $J(\zeta;\nu)$ from \eqref{A11}. We emphasize that $\zeta\mapsto m(\zeta;s)$ is constructed in such a way that
\begin{equation*}
	m_+\big(\zeta x^{-2}\big)-m_-\big(\zeta x^{-2}\big)=-\frac{\e^{\im\pi s}(-\zeta)^s}{2\cos(\pi s)}\big(\zeta+\beta x^2\big)^s,\ \ \zeta\in\Gamma_1;\ \ \ \ \ \ m(\zeta)=\mathcal{O}\big(\zeta^{-1}\big),\ \ \zeta\rightarrow\infty.
\end{equation*}
If on the other hand $2s\in\mathbb{Z}_{\geq 0}$, then \eqref{f7} needs to be modified according to, for $\zeta\in\mathbb{D}_{\epsilon}(0)\setminus\Sigma_H$,
\begin{equation}\label{f8}
	W(\zeta;x,s):=E_W(\zeta)\begin{bmatrix}1 & m(\zeta x^{-2})\smallskip\\ 0 & 1\end{bmatrix}\big(\zeta+\beta x^2\big)^{\frac{s}{2}\sigma_3}\big(\e^{-\im\pi}\zeta\big)^{\frac{s}{2}\sigma_3}\begin{bmatrix}1 & -\frac{\e^{2\pi\im s}}{2\pi\im}\ln(\e^{-\im\pi}\zeta)\smallskip\\ 0 & 1\end{bmatrix}\mathcal{N}(\zeta;2s),
\end{equation}
where $m(\zeta)=m(\zeta;s,x)$ is the scalar function
\begin{equation*}
	m(\zeta;x,s):=-\frac{x^{4s}}{2\pi\im}\int_{-\beta}^0\frac{\e^{3\pi\im s}(-\tau)^s(\tau+\beta)^s\ln(-\tau x^2)\sin(\pi s)}{\pi}\frac{\d\tau}{\tau-\zeta},\ \ \ \ \zeta\in\mathbb{C}\setminus[-\beta,0],
\end{equation*}
and $\mathbb{D}_{\epsilon}(0)\ni\zeta\mapsto E_W(\zeta)$ the analytic multiplier
\begin{equation}\label{defEwforint}
E_W(\zeta)=E_W(\zeta;x,s):= \e^{\im\frac{\pi}{4}\sigma_3}J(\zeta;2s)\mathcal{N}(\zeta;2s)^{-1}\begin{bmatrix}1 & \frac{\e^{2\pi\im s}}{2\pi\im}\ln(\e^{-\im\pi}\zeta)\\ 0 & 1\end{bmatrix}\big(\e^{-\im\pi}\zeta\big)^{-s\sigma_3}\e^{-\im\frac{\pi}{2}s\sigma_3}.
\end{equation}
Observe that this time
\begin{equation*}
	m_+\big(\zeta x^{-2}\big)-m_-\big(\zeta x^{-2}\big)=-\frac{1}{\pi}\e^{3\pi\im s}\ln(-\zeta)\big(\zeta+\beta x^2\big)^s(-\zeta)^s\sin(\pi s),\ \zeta\in\Gamma_1;\ \ \ m(\zeta)=\mathcal{O}\big(\zeta^{-1}\big),\ \ \zeta\rightarrow\infty.
\end{equation*}
The relevant analytic and asymptotic properties of \eqref{f7},\eqref{f8} are summarised below.
\begin{problem}[Origin parametrix] Let $0<x\ll 1$ and $s>-\frac{1}{2}$ so that $-\beta x^2\in\mathbb{D}_{\epsilon}(0)$. The model function $W(\zeta)=W(\zeta;x,s)\in\mathbb{C}^{2\times 2}$ defined in \eqref{f7},\eqref{f8} has the following properties:
\begin{enumerate}
	\item[(1)] $\zeta\mapsto W(\zeta)$ is analytic for $\zeta\in\mathbb{C}\setminus(\Sigma_H\cap\mathbb{D}_{\epsilon}(0))$ with $\Sigma_H$ shown in Figure \ref{fig9}. On $(\Sigma_H\cap\mathbb{D}_{\epsilon}(0))\setminus\{-\beta x^2,0\}$, $W(\zeta)$ admits continuous limiting values $W_{\pm}(\zeta)$ as we approach $(\Sigma_H\cap\mathbb{D}_{\epsilon}(0))\setminus\{-\beta x^2,0\}$ from either side of $\mathbb{C}\setminus\Sigma_H$.
	\item[(2)] The limiting values $W_{\pm}(\zeta)$ on $(\Sigma_H\cap\mathbb{D}_{\epsilon}(0))\setminus\{-\beta x^2,0\}$ obey $W_+(\zeta)=W_-(\zeta)G_W(\zeta)$ with $G_W(\zeta)=G_W(\zeta;x,s)$ given by
	\begin{align*}
		G_W(\zeta)=\e^{-\im\pi s\sigma_3},\ \zeta\in\Gamma_1\cap\mathbb{D}_{\epsilon}(0);&\hspace{1cm}G_W(\zeta)=\begin{bmatrix}0&1\\ -1 & 0\end{bmatrix},\ \zeta\in\Gamma_3\cap\mathbb{D}_{\epsilon}(0);\\
		G_W(\zeta)=\begin{bmatrix}1&0\\ \e^{2\pi\im s}& 1\end{bmatrix},\ \zeta\in\Gamma_2\cap\mathbb{D}_{\epsilon}(0);&\hspace{1cm}G_W(\zeta)=\begin{bmatrix}1 & 0\\ \e^{-2\pi\im s} & 1\end{bmatrix},\ \zeta\in\Gamma_4\cap\mathbb{D}_{\epsilon}(0).
	\end{align*}
	\item[(3)] Near $\zeta=-\beta x^2,\zeta\mapsto W(\zeta)$ is weakly singular in that, with some $\zeta\mapsto\widehat{W}(\zeta)=\widehat{W}(\zeta;x,s)$ analytic at $\zeta=-\beta x^2$ and non-vanishing, $W(\zeta)$ is of the form
	\begin{equation*}
		W(\zeta)=\widehat{W}(\zeta)\big(\zeta+\beta x^2\big)^{\frac{s}{2}\sigma_3},\ \ \ \ \textnormal{arg}(\zeta+\beta x^2)\in(0,2\pi).
	\end{equation*}
	\item[(4)] Near $\zeta=0$, $\zeta\mapsto W(\zeta)$ is weakly singular in that, with some $\zeta\mapsto\widecheck{W}(\zeta)=\widecheck{W}(\zeta;x,s)$ analytic at $\zeta=0$ and non-vanishing, $W(\zeta)$ is of the form
	\begin{equation*}
		W(\zeta)=\widecheck{W}(\zeta)\zeta^{\frac{s}{2}\sigma_3}\mathcal{M}_j(\zeta),\ \ \ \zeta\in\Omega_j\cap\mathbb{D}_{\epsilon}(0),\ \ \ \ \ \ \ \textnormal{arg}\,\zeta\in(0,2\pi),
	\end{equation*}
	with $\mathcal{M}_j(\zeta)$ as in condition $(4)$ of RHP \ref{PIIImodel}.
	\item[(5)] As $x\downarrow 0$, $W(\zeta)$ in \eqref{f7},\eqref{f8} compares to $J(\zeta;2s)$ in \eqref{A11} as follows: if $2s\notin\mathbb{Z}_{\geq 0}$, then
	\begin{equation}\label{f9}
		W(\zeta)=\bigg\{I+W_x(\zeta;s)+\mathcal{O}\big(x^{2\kappa}\big)\bigg\}\e^{\im\frac{\pi}{4}\sigma_3}J(\zeta;2s),\ \ \ \ \kappa:=2\min\{1,1+s\}>1,
	\end{equation}
	uniformly for $0<\delta_1\leq|\zeta|\leq\delta_2<\epsilon$, for any fixed $s>-\frac{1}{2}$. Here,
	\begin{equation*}
		W_x(\zeta;s)=E_W(\zeta)\frac{x^2}{\zeta}\begin{bmatrix}\frac{\beta s}{2} & \mathsf{h}_{s}(x)\smallskip\\ 0 & -\frac{\beta s}{2}\end{bmatrix}E_W(\zeta)^{-1};\ \ \ \ \ \ \mathsf{h}_{s}(x):=\frac{x^{4s}}{2\pi\im}\int_{-\beta}^0\frac{\e^{\im\pi s}(-\tau)^s(\tau+\beta)^s}{2\cos(\pi s)}\d\tau.
	\end{equation*}
	If on the other hand $2s\in\mathbb{Z}_{\geq 0}$, then instead
	\begin{equation}\label{f10}
		W(\zeta)=\bigg\{I+W_x(\zeta;s)+\mathcal{O}\big(\max\{x^4,x^{4+4s}\ln x\}\big)\bigg\}\e^{\im\frac{\pi}{4}\sigma_3}J(\zeta;2s),
	\end{equation}
	uniformly for $0<\delta_1\leq|\zeta|\leq\delta_2<\epsilon$, for any fixed $s>-\frac{1}{2}$. Here,
	\begin{equation*}
		W_x(\zeta;s)=E_W(\zeta)\frac{x^2}{\zeta}\begin{bmatrix}\frac{\beta s}{2} & \tilde{\mathsf{h}}_{s}(x)\\ 0 & -\frac{\beta s}{2}\end{bmatrix}E_W(\zeta)^{-1};\ \ \ \ \ \ \tilde{\mathsf{h}}_{s}(x):=\frac{x^{4s}}{2\pi\im}\int_{-\beta}^0\frac{e^{3\pi\im s}(-\tau)^s(\tau+\beta)^s\ln(-\tau x^2)\sin(\pi s)}{\pi}\d\tau.
	\end{equation*}
\end{enumerate}
\end{problem}
\begin{rem} In establishing conditions $(3)$ we make use of the local expansions
\begin{align*}
	m(\zeta x^{-2})=&\,\,-\frac{\im}{2}\frac{\beta^sx^{2s}}{\sin(2\pi s)}\big(\zeta+\beta x^2\big)^s+\mathcal{O}(1),\ \ \ \ \ 2s\notin\mathbb{Z}_{\geq 0};\\
	m(\zeta x^{-2})=&\,\,\frac{\e^{2\pi\im s}}{2\pi\im}\beta^sx^{2s}\ln(\beta x^2)\big(\zeta+\beta x^2\big)^s+\mathcal{O}(1),\ \ \ \ \ \ 2s\in\mathbb{Z}_{\geq 0};
\end{align*}
valid as $\zeta\rightarrow-\beta x^2$ with $\textnormal{arg}(\zeta+\beta x^2)\in(0,2\pi)$. These follow from our integral representations of $m(\zeta)$. One can argue similarly in condition $(4)$, or simply observe that \eqref{f7},\eqref{f8} leave $\zeta=-\beta x^2,0$ to be isolated singularities of
\begin{equation*}
	\mathbb{D}_{\epsilon}(-\beta x^2)\ni\zeta\mapsto W(\zeta)\big(\zeta+\beta x^2\big)^{-\frac{s}{2}\sigma_3},\ \ \ \ \ \ \mathbb{D}_{\epsilon}(0)\ni\zeta\mapsto W(\zeta)\mathcal{M}_j(\zeta)^{-1}\zeta^{-\frac{s}{2}\sigma_3},
\end{equation*}
which have to be removable by our explicit formulae \eqref{f7},\eqref{f8}.
\end{rem}
\begin{rem}\label{structure} A closer inspection of the workings underwriting \eqref{f9} reveals that the same matching between $W(\zeta)$ and $\e^{\im\frac{\pi}{4}\sigma_3}J(\zeta;2s)$ as $x\downarrow 0$ on $0<\delta_1\leq|\zeta|\leq\delta_2<\epsilon$ extends to a full asymptotic matching of the form
\begin{equation*}
	W(\zeta)\big(\e^{\im\frac{\pi}{4}\sigma_3}J(\zeta;2s)\big)^{-1}\sim\begin{cases} I+x^{4s}\mathbb{C}[[x^2]]+\mathbb{C}[[x^2]],& 2s\notin\mathbb{Z}_{\geq 0}\smallskip\\ I+\mathbb{C}[[x^2]],&2s\in\mathbb{Z}_{\geq 1},\ 2s\ \textnormal{even}\smallskip\\ I+x^{4s}\ln(x^2)\mathbb{C}[[x^2]]+x^{4s}\mathbb{C}[[x^2]]+\mathbb{C}[[x^2]],&2s\in\mathbb{Z}_{\geq 1},\ 2s\ \textnormal{odd}
	\end{cases}
\end{equation*}
where $\mathbb{C}[[x^2]]$ are matrix-valued  power series in $x^2$ of degree at least one, with coefficients that are meromorphic in $\zeta$. Moreover there is an $x_{0}$ (depending only on $\delta_{1}$ and $\delta_{2}$) such that each series  $\mathbb{C}[[x^2]]$ absolutely convergent for $|x|<x_0$, uniformly for $\zeta$ with $|\zeta|\in[\delta_1,\delta_2]$. In particular, when $s=0$, then we have equality
\begin{equation*}
	W(\zeta)=\e^{\im\frac{\pi}{4}\sigma_3}J(\zeta;0),\ \ \ \ 0<\delta_1\leq|\zeta|\leq\delta_2<\epsilon.
\end{equation*}
\end{rem}

\begin{proof}[Proof of Remark \ref{structure}]
Let $0<\delta_1\le|\zeta|\le \delta_2<\epsilon$ and set 
$
x_0= \sqrt{\delta_1/(2\beta)}.
$
In the following, we show that for $|x|<x_0$, the asymptotic series expansion, as $x\downarrow 0$, of
\[
W(\zeta;x,s) \big(\e^{\im \frac{\pi}{4} \, \sigma_3} J(\zeta;2s)\big)^{-1},
\]
converges absolutely and uniformly with respect to $\zeta$. By \eqref{f7} and \eqref{f8}, we have 
\begin{equation}\label{1019formula1}
W(\zeta;x,s) \big(\e^{\im\frac{\pi}{4} \, \sigma_3} J(\zeta;2s)\big)^{-1}=E_{W}(\zeta)\begin{bmatrix}1 & m(\zeta x^{-2}) \\ 0 & 1\end{bmatrix}(1+\beta x^2/\zeta)^{\frac{s}{2}\sigma_{3}}E_{W}(\zeta)^{-1}.
\end{equation}
Consider the factor $(\zeta + \beta x^2)^{s/2}$. Using the binomial series,
$
(\zeta + \beta x^2)^{s/2} = \zeta^{s/2} \sum_{k=0}^\infty \binom{s/2}{k} \left(\beta x^2/\zeta\right)^k.
$
Since 
\bea\label{quantitivecon}
|\beta x^2 / \zeta| \le \beta x^2 / \delta_1 < 1/2
\eea
for $|x|<x_0$, this series converges absolutely, uniformly for $\zeta$ in the annulus $\delta_1 \le |\zeta|\le \delta_2$. For $2s\notin\mathbb{Z}_{\ge 0}$, the scalar function
\[
m(\zeta x^{-2};x,s) = -\frac{x^{4s}}{2\pi i} \int_{-\beta}^0 \frac{(-\tau)^s (\tau+\beta)^s \e^{\im\pi s}}{2 \cos(\pi s)} \frac{d\tau}{\tau - \zeta x^{-2}}, \ \ \ \ \zeta x^{-2}\in \mathbb{C}\setminus[-\beta,0]
\]
can be expanded using the geometric series
$
(\tau - \zeta x^{-2})^{-1}=-\zeta^{-1}x^2 \sum_{k=0}^\infty \left( \zeta^{-1}\tau x^2\right)^k, |\tau x^2 / \zeta| < 1/2.
$
Substituting this series into the integral gives
$
m(\zeta x^{-2};x,s) = x^{4s} \sum_{k=1}^\infty \mathsf{f}_k(s) (\zeta^{-1}x^2)^{k}, 
$
where
\bea\label{defforfk}
\mathsf{f}_k(s) = \frac{1}{2\pi \im}\int_{-\beta}^0 \frac{(-\tau)^{s+k-1} (\tau+\beta)^s \e^{\im\pi s}}{\cos(\pi s)} d\tau=
\frac{\e^{\im \pi s} \beta^{2s}}{4\pi \im \cos(\pi s )}\frac{\Gamma(s+k)\Gamma(s+1)}{\Gamma(2s+1+k)}.
\eea
Note that $\mathsf{f}_{k}(s)=\mathcal{O}(k^{-s-1})$ so that the series $\sum_{k=1}^\infty \mathsf{f}_k(s)(\zeta^{-1}x^{2})^k$ converges absolutely for $|x|<x_0$, uniformly in $\zeta\in[\delta_1,\delta_2]$. By (\ref{1019formula1}), for $2s\notin \mathbb{Z}_{\geq 0}$, with $E_{W}(\zeta)$ given as in (\ref{defEWfornonint}), 
\begin{align}\label{explicitfornonint}
&W(\zeta;x,s) \big(\e^{\im \frac{\pi}{4} \, \sigma_3} J(\zeta;2s)\big)^{-1}-I\nonumber\\
=&\,\,E_{W}(\zeta)\begin{bmatrix}
\sum_{k=1}^\infty \binom{s/2}{k} \left(\beta x^2/\zeta\right)^k & x^{4s}\sum_{k=1}^\infty \mathsf{g}_{k}(s)\left(\beta x^2/\zeta\right)^k\smallskip\\
0&\sum_{k=1}^\infty \binom{-s/2}{k} \left(\beta x^2/\zeta\right)^k
\end{bmatrix}E_{W}(\zeta)^{-1},
\end{align}
where $\mathsf{g}_{k}(s)=\mathsf{f}_{k}(s)+\sum_{n=1}^{k-1}\mathsf{f}_{k-n}(s)\binom{-s/2}{n}\beta^n$. Here $\binom{-s/2}{n}=O(n^{s/2-1})$ and we recall \eqref{quantitivecon}. So all the series involved in the matrix \eqref{explicitfornonint} converge absolutely for $|x|<x_0$, uniformly in $\zeta \in [\delta_1, \delta_2]$. Note that $E_{W}(\zeta
)$ is independent of $x$ and  analytic when $\zeta\in \mathbb{D}_{\epsilon}(0)$, so we have the conclusion in Remark \ref{structure} for $2s\notin \mathbb{Z}_{\geq 0}$. When $2s\in\mathbb{Z}_{\ge 1}$, the integral defining $m(\zeta x^{-2})$ may produce logarithmic terms, leading to contributions of the form $x^{4s}\ln(x^2) \mathbb{C}[[x^2]]$. More specifically, by a similar argument to that used in the case $2s\notin \mathbb{Z}_{\geq 0}$, we conclude that for $2s\in\mathbb{Z}_{\ge 1}$, with $E_{W}(\zeta)$ given as in (\ref{defEwforint}),
\begin{align}\label{explicitforint}
&W(\zeta;x,s) \big(\e^{\im \frac{\pi}{4} \, \sigma_3} J(\zeta;2s)\big)^{-1}-I\nonumber\\
=&\,\,E_{W}(\zeta)\begin{bmatrix}
\sum_{k=1}^\infty \binom{s/2}{k} \left(\beta x^2/\zeta\right)^k & x^{4s}\sum_{k=1}^\infty \left(\mathsf{g}_{k,1}(s)+\ln(x^2)\mathsf{g}_{k,2}(s)\right)\left(\beta x^2/\zeta\right)^k\smallskip\\
0&\sum_{k=1}^\infty \binom{-s/2}{k} \left(\beta x^2/\zeta\right)^k
\end{bmatrix}E_{W}(\zeta)^{-1},
\end{align}
where 
\begin{align*}
\mathsf{f}_{k,1}(s)& =\frac{\e^{3\im\pi s}\sin(\pi s)\beta^{2s}}{2\im \pi^2}\int_{0}^{1} \tau^{s+k-1} (1-\tau)^s \ln(\beta \tau)d\tau\\
&=\frac{\e^{3\im\pi s}\sin(\pi s)\beta^{2s}}{2\im \pi^2}\frac{\Gamma(s+k)\Gamma(s+1)}{\Gamma(2s+1+k)}\left(\ln \beta+\psi(s+k)-\psi(2s+k+1)\right),
\end{align*}
with $\psi$ the digamma function,
and
\beas
\mathsf{f}_{k,2}(s)=
\frac{\e^{3\im \pi s}\sin(\pi s) \beta^{2s}}{2\im \pi^2}\frac{\Gamma(s+k)\Gamma(s+1)}{\Gamma(2s+1+k)},
\eeas
together with $\mathsf{g}_{k,1}(s)=\mathsf{f}_{k,1}(s)+\sum_{n=1}^{k-1}\mathsf{f}_{k-n,1}(s)\binom{-s/2}{n}\beta^{n}$, 
$\mathsf{g}_{k,2}(s)=\mathsf{f}_{k,2}(s)+\sum_{n=1}^{k-1}\mathsf{f}_{k-n,2}(s)\binom{-s/2}{n}\beta^{n}$. When $2s$ is even, $\mathsf{f}_{k,1}(s) = \mathsf{f}_{k,2}(s)=0$, so the claim in this remark holds for both even and odd $2s$.
Finally, when $s=0$, $m(\zeta x^{-2})\equiv 0$ and $(\zeta+\beta x^2)^{0} = 1$, so we have exact equality
$
W(\zeta;0) = \e^{\im \frac{\pi}{4}\sigma_3} J(\zeta;0), \delta_1 \le |\zeta| \le \delta_2.
$
\end{proof}
\subsection{The final ratio transformation} We take $0<x\ll 1$ and $0<\epsilon<1$ so that $-\beta x^2\in\mathbb{D}_{\epsilon}(0)$. Now compare $J(\zeta;2s),W(\zeta)$ in \eqref{A11},\eqref{f7},\eqref{f8} to $H(\zeta)$ from RHP \ref{traf2} by defining
\begin{equation}\label{f11}
	Z(\zeta;x,s):=H(\zeta;x,s)\begin{cases}W(\zeta;x,s)^{-1},&\zeta\in\mathbb{D}_{\epsilon}(0)\smallskip\\ J(\zeta;2s)^{-1}\e^{-\im\frac{\pi}{4}\sigma_3},&\zeta\in\mathbb{C}\setminus(\Sigma_H\cup\overline{\mathbb{D}_{\epsilon}(0)})\end{cases},
\end{equation}
and collect the ratio function's properties below.
\begin{problem}[Small norm problem]\label{SN3} Let $s>-\frac{1}{2}$ and $0<x\ll 1$. The function $Z(\zeta)=Z(\zeta;x,s)\in\mathbb{C}^{2\times 2}$ defined in \eqref{f11} has the following properties:
\begin{enumerate}
	\item[(1)] $\zeta\mapsto Z(\zeta)$ is analytic for $\zeta\in\mathbb{C}\setminus\Sigma_Z$ with $\Sigma_Z=\partial\mathbb{D}_{\epsilon}(0)$ shown in Figure \ref{fig9x}. On $\Sigma_Z$, $Z(\zeta)$ admits continuous boundary values $Z_{\pm}(\zeta)$ as one approaches $\Sigma_Z$ from either side of $\mathbb{C}\setminus\Sigma_Z$ non-tangentially.
	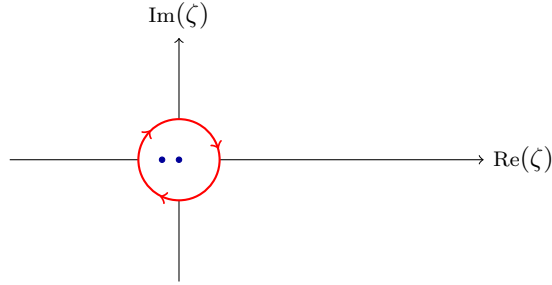
\begin{figure}[tbh]
	\begin{tikzpicture}[xscale=0.9,yscale=0.9]
	        \draw [->] (-5,0) -- (2,0) node [right] {$\footnotesize{\textnormal{Re}(\zeta)}$};
	  \draw [->] (-2.5,-1.8) -- (-2.5,1.8) node [above] {$\footnotesize{\textnormal{Im}(\zeta)}$};
	\draw [thick, color=red, fill=white, decoration={markings, mark=at position 0.07 with {\arrow{<}}}, decoration={markings, mark=at position 0.4 with {\arrow{<}}}, decoration={markings, mark=at position 0.7 with {\arrow{<}}},postaction={decorate}] (-2.5,0) circle [radius=0.6];
	\draw[fill, color=blue!60!black](-2.5,0) circle [radius=0.04];
	\draw[fill, color=blue!60!black](-2.75,0) circle [radius=0.04];
\end{tikzpicture}
\caption{The oriented jump contour $\Sigma_Z$, in the complex $\zeta$-plane drawn in \textcolor{red}{red}, for $Z(\zeta)$. The points $\zeta=-\beta x^2,0$ are colored in \textcolor{blue!60!black}{blue}.}
\label{fig9x}
\end{figure}
	\item[(2)] On $\Sigma_Z$ we record the jump condition $Z_+(\zeta)=Z_-(\zeta)G_Z(\zeta)$ with $G_Z(\zeta)=G_Z(\zeta;x,s)$ given as
	\begin{equation*}
		G_Z(\zeta)=W(\zeta)J(\zeta;2s)^{-1}\e^{-\im\frac{\pi}{4}\sigma_3},\ \ \ \zeta\in\Sigma_Z.
	\end{equation*}
	By construction, $Z(\zeta)$ has no jumps on $(\Gamma_2\cup\Gamma_3\cup\Gamma_4)\cap\{|\zeta|\geq\epsilon\}$
	\item[(3)] As $\zeta\rightarrow\infty$, $Z(\zeta)\rightarrow I$.
\end{enumerate}
\end{problem}

The fact that RHP \ref{SN3} constitutes a small norm problem, as $x\downarrow 0$, is made precisely below. It follows from RHP \ref{SN3}, condition $(2)$ and Remark \ref{structure}.
\begin{prop}\label{smallL2norm} Fix $s>-\frac{1}{2}$. There exists $x_0\in(0,1), c=c(s,x_{0})$ so that
\begin{equation*}
	\|G_Z(\cdot;x,s)-I\|_{L^2\cap L^{\infty}(\Sigma_Z)}\leq c\begin{cases}\max\{x^2,x^{2+4s}\},&2s\notin\mathbb{Z}_{\geq 0}\\ x^2,&2s\in\mathbb{Z}_{\geq 1},\ 2s\ \textnormal{even}\\ \max\{x^2,x^{2+4s}|\ln(x^2)|\},&2s\in\mathbb{Z}_{\geq 1},\ 2s\ \textnormal{odd}\end{cases}\ \ \ \ \ \ \  \forall\,\,0<x\leq x_0<1.
\end{equation*}
Moreover, $\|G_Z(\cdot;x,0)-I\|_{L^2\cap L^{\infty}(\Sigma_Z)}=0$.
\end{prop}
Consequently, RHP \ref{SN3} is solvable for $0<x\leq x_0$, cf. \cite{DZ}, and we summarise our findings below.
\begin{theo} For every $s>-\frac{1}{2}$, there exist $c=c(s)>0$ and $x_0=x_0(s)\in(0,1)$ so that RHP \ref{SN3} is uniquely solvable for all $0<x\leq x_0$ and any $\epsilon\in(0,1)$. Its solution admits the integral representation
\begin{equation}\label{z1}
	Z(\zeta)=I+\frac{1}{2\pi\im}\int_{\Sigma_Z}Z_-(\mu)\big(G_Z(\mu)-I\big)\frac{\d\mu}{\mu-\zeta},\ \ \ \ \zeta\in\mathbb{C}\setminus\Sigma_Z,
\end{equation}
where
\begin{equation*}
	\|Z_-(\cdot;x,s)-I\|_{L^2(\Sigma_Z)}\leq c\begin{cases}\max\{x^2,x^{2+4s}\},&2s\notin\mathbb{Z}_{\geq 0}\\ x^2,&2s\in\mathbb{Z}_{\geq 1},\ 2s\ \textnormal{even}\\ \max\{x^2,x^{2+4s}|\ln(x^2)|\},&2s\in\mathbb{Z}_{\geq 1},\ 2s\ \textnormal{odd}\end{cases}\ \ \ \ \ \forall\,\,0<x\leq x_0<1.
\end{equation*}
In addition, if $s=0$, then $Z(\zeta)\equiv I$.
\end{theo}
\begin{proof}
It follows from the singular integral equation approach of \cite{DZ} for solving Riemann–Hilbert problems that the solution admits the integral representation \eqref{z1}. Furthermore, the same method yields a Neumann series representation of $Z_{-}(\zeta;x,s)$ for $\zeta \in \Sigma_{Z}$. For completeness and for ease of reference in subsequent results, we give the derivation of this expansion below (we did not do the same in the context of Theorem \ref{newby1} and \ref{newby2}, but do it now). By Proposition \ref{smallL2norm}, there is an $x_{0}=x_{0}(s)\in (0,1)$ and $c=c(s)>0$ 
such that for $0<x\le x_0$, 
\bea\label{less1}
\|G_Z(\cdot;x,s)-I\|_{L^2\cap L^\infty(\Sigma_Z)} \leq cb(x)<1;\ \ \ b(x):=\begin{cases}
\max\{x^2, x^{2+4s}\}, & 2s\notin \mathbb{Z}_{\ge 0},\\
x^2, & 2s\in \mathbb{Z}_{\ge 1}, \ 2s\text{ even}\\
\max\{x^2, x^{2+4s}|\ln(x^2)|\}, & 2s\in \mathbb{Z}_{\ge 1}, \ 2s\text{ odd}
\end{cases}.
\eea
Let $\mathcal{C}_{\pm}: L^2(\Sigma_{Z})\rightarrow  L^2(\Sigma_{Z})$ be the Cauchy operators, which are given by 
\beas
\mathcal{C}_{\pm}[f](\zeta)
  = \lim_{\substack{\zeta'\to\zeta\\ \zeta'\in\text{side }\pm}}
    \frac{1}{2\pi \im}\int_{\Sigma_Z} f(\mu)\frac{\d\mu}{\mu-\zeta'}
\quad \zeta\in\Sigma_Z,
\eeas
where ``side $+$" is outside of $\Sigma_Z$ and ``side $-$" is inside of $\Sigma_{Z}$. By the Sokhotski-Plemelj theorem, 
\beas
\mathcal{C}_{+}[f]-\mathcal{C}_{-}[f]=f\ \ \ \ \textnormal{a.e. on}\ \Sigma_Z,
\eeas
and it is known that $\mathcal{C}_{\pm}$ are bounded linear operators on $L^2(\Sigma_Z)$.
By equation (\ref{z1}), 
\beas
Z_{-}=I+\mathcal{C}_{-}[Z_{-}(G_{Z}-I)].
\eeas
For convenience, we introduce the notation.
\bea\label{defofcg}
\mathcal{C}_{G}[f]:=\mathcal{C}_{-}[f(G_{Z}-I)],
\eea
and observe that $\mathcal{C}_{G}$ is a bounded operator on $L^{2}(\Sigma_{Z})$. In fact, in operator norm on $L^2(\Sigma_Z)$,
\bea\label{normofcg}
\|\mathcal{C}_{G}\|\leq \|\mathcal{C}_{-}\|\|G_{Z}(\cdot;x,s)-I\|_{L^{\infty}(\Sigma_{Z})}\stackrel{\eqref{less1}}{\leq}\frac{1}{2},
\eea
where the last inequality holds for $x$ sufficiently small. Consequently, for sufficiently small $x$,
\begin{align}\label{neumann}
Z_{-}=(I-\mathcal{C}_{G})^{-1}[I]=\sum_{n=0}^{\infty}\mathcal{C}_{G}^{n}[I],
\end{align}
where $\mathcal{C}_{G}^{n}$ denotes the $n$-fold composition of the operator $\mathcal{C}_{G}$ with itself, and the series \eqref{neumann} converges in norm on $L^2(\Sigma_{Z})$. What results by \eqref{normofcg} is therefore, with some $c=c(\Sigma_Z)>0$,
\beas
\|Z_-(\cdot;x,s)-I\|_{L^2(\Sigma_Z)}=\left\|\sum_{n=1}^{\infty}\mathcal{C}_{G}^{n}[I](\cdot;x,s)\right\|_{L^2(\Sigma_Z)}\leq c\,\|G_Z(\cdot;x,s)-I\|_{ L^\infty(\Sigma_Z)},
\eeas
for sufficiently small $x$. Then by \eqref{less1}, we obtain the claim in this theorem. In addition, if $s=0$, then $G_{Z}(\mu)\equiv I$ for $\mu\in \Sigma_{Z}$ and thus $Z(\zeta)\equiv I$ when $\zeta\in \mathbb{C}\setminus \Sigma_Z$.
\end{proof}

We are now left to extract the sought after small $x$-asymptotics.
\subsection{Distilling asymptotics} Simply trace back \eqref{f6},\eqref{f11}, with $\nu=2s$,
\begin{align}
	&x^{\frac{1}{2}\sigma_3}Q_1(x,s)x^{-\frac{1}{2}\sigma_3}=\lim_{\zeta\rightarrow\infty}x^{\frac{1}{2}\sigma_3}\Bigg\{Q(\zeta;x,s)\e^{\im\varpi(\zeta;x)\sigma_3}\e^{\im\frac{\pi}{4}\sigma_3}\frac{1}{\sqrt{2}}\begin{bmatrix}1 & -1\\ 1 & 1\end{bmatrix}\zeta^{-\frac{1}{4}\sigma_3}-I\Bigg\}x^{-\frac{1}{2}\sigma_3}\nonumber\\
	=&\,\,\frac{1}{x^2}\Bigg\{\frac{\im}{2\pi}\int_{\Sigma_Z}Z_-(\mu)\big(G_Z(\mu)-I\big)\d\mu+R(\nu)\Bigg\}\label{expressionofQ1},
\end{align}
where
\begin{align}\label{defofRnu}
R(\nu)=
\begin{bmatrix}\frac{1}{128}(4\nu^2-1)(4\nu^2-9) & \frac{\im}{1536}(4\nu^2-1)(4\nu^2-9)(4\nu^2-13)\smallskip\\ \frac{\im}{8}(4\nu^2-1) & -\frac{1}{128}(4\nu^2-1)(4\nu^2-9)\end{bmatrix}\Bigg\}.
\end{align}
Here, the first equality follows from \eqref{A4}, while the second follows from \eqref{f6} and \eqref{f11}. Next, by an application of the Cauchy-Schwarz inequality, as $x\downarrow 0$, 
\begin{equation*}
	\frac{\im}{2\pi}\int_{\Sigma_Z}Z_-(\mu)\big(G_Z(\mu)-I\big)\d\mu=\frac{\im}{2\pi}\oint_{\partial\mathbb{D}_{\epsilon}(0)}\big(G_Z(\mu)-I\big)\d\mu+\begin{cases}\mathcal{O}(\max\{x^4,x^{4+8s}\}),&2s\notin\mathbb{Z}_{\geq 0}\\ \mathcal{O}(x^4),&2s\ \textnormal{even}\\ \mathcal{O}(\max\{x^4,x^{4+8s}|\ln(x^2)|^2\}),&2s\ \textnormal{odd}\end{cases},
\end{equation*}
where the remaining integral vanishes identically for $s=0$. We are thus left to evaluate a contour integral.
\begin{prop}\label{prop68} 
Let $s>-\frac{1}{2}$. Assuming $s\neq 0$, as $x\downarrow 0$,
\begin{align}
	\frac{\im}{2\pi}\oint_{\partial\mathbb{D}_{\epsilon}(0)}\big(G_Z(\mu)-I\big)\d\mu
	=&\,\,x^2\begin{bmatrix}E_{11} & E_{12}\\ E_{21} & E_{22}\end{bmatrix}+\begin{cases}\mathcal{O}(\max\{x^4,x^{4+4s}\}),&2s\notin\mathbb{Z}_{\geq 0}\\ \mathcal{O}(x^4),&2s\ \textnormal{even}\\ \mathcal{O}(\max\{x^4,x^{4+4s}|\ln(x^2)|\}),&2s\ \textnormal{odd}\end{cases},\label{f12}
\end{align}
with
\begin{equation*}
	E_{21}=\frac{\im\beta}{4}+\frac{\pi}{4s^2}\frac{\e^{-\im\pi s}\mathsf{H}_{s}(x)}{(4^{s}\Gamma(2s))^2},
\ \ \ \    \mathsf{H}_{s}(x)=\begin{cases} 
    \mathsf{h}_{s}(x),&2s\notin\mathbb{Z}_{\geq 0}\\ \tilde{\mathsf{h}}_{s}(x),&2s\in\mathbb{Z}_{\geq 1}\end{cases}
\ \ \ \ \nu=2s\neq 0.
\end{equation*}
The entries $E_{11}=-E_{22}$ and $E_{12}$ are irrelevant for us and $\mathsf{h}_{s}(x)$, $\tilde{\mathsf{h}}_{s}(x)$ appeared in \eqref{f9},\eqref{f10}.
\end{prop}
\begin{proof} From Remark \ref{nearzero}, with $\nu=2s\notin\mathbb{Z}_{\geq 0}$, as $\zeta\rightarrow 0$,
\begin{equation*}
	E_W(\zeta)=\e^{\im\frac{\pi}{4}\sigma_3}\begin{bmatrix}\frac{1}{8}(4\nu^2+3) & 1\\ -1 & 0\end{bmatrix}\sqrt{\pi}\,\e^{-\im\frac{\pi}{4}}\Bigg\{\begin{bmatrix} I_{11}(0;\nu) & -\frac{\im}{2}\frac{1}{\sin(\pi\nu)}I_{11}(0;-\nu)\smallskip\\ I_{21}(0;\nu) & -\frac{\im}{2}\frac{1}{\sin(\pi\nu)}I_{21}(0;-\nu)\end{bmatrix}+\mathcal{O}\big(\zeta\big)\Bigg\}\e^{-\im\frac{\pi}{2}s\sigma_3},
\end{equation*}
where $I_{11}(0;\nu)=2^{-\nu}\Gamma(1+\nu)^{-1}$ and $I_{21}(0;\nu)= 2^{-\nu}\Gamma(\nu)^{-1}$. If $\nu=2s\in\mathbb{Z}_{\geq 1}$, then, as $\zeta\rightarrow 0$,
\begin{align*}
	E_W(\zeta)=&\,\e^{\im\frac{\pi}{4}\sigma_3}\begin{bmatrix}\frac{1}{8}(4\nu^2+3) & 1\\ -1 & 0\end{bmatrix}\sqrt{\pi}\,\e^{-\im\frac{\pi}{4}}\Bigg\{\begin{bmatrix}I_{11}(0;\nu) & I_{12}(0;\nu)\smallskip\\ I_{21}(0;\nu) & I_{22}(0;\nu)\end{bmatrix}+\mathcal{O}(\zeta)\Bigg\}\e^{-\im\frac{\pi}{2}s\sigma_3},
\end{align*}
with $I_{11}(0;\nu)=2^{-\nu}\Gamma(1+\nu)^{-1}, I_{21}(0;\nu)=2^{-\nu}\Gamma(\nu)^{-1},I_{12}(0;\nu)=-\frac{\im}{\pi}2^{\nu-1}\Gamma(\nu)$ and $I_{22}(0;\nu)=\frac{\im}{\pi}2^{\nu-1}\Gamma(1+\nu)$. In turn, by residue theorem,
\begin{equation*}
	\frac{\im}{2\pi}\oint_{\partial\mathbb{D}_{\epsilon}(0)}W_x(\mu;s)\d\mu=x^2E_W(0)\begin{bmatrix}\frac{\beta s}{2} & \mathsf{H}_{s}(x)\smallskip\\ 0 & -\frac{\beta s}{2}\end{bmatrix}E_W(0)^{-1},
\end{equation*}
which yields our result, after simplification, compare \eqref{f9} and \eqref{f10}. The proof is complete.
\end{proof}
\begin{cor} 
Let $s>-\frac{1}{2}$. Let $R(2s)$ be given as in (\ref{defofRnu}). Assuming $s\neq 0$, as $x\downarrow 0$,
\begin{align}
	x^{\frac{1}{2}\sigma_3}Q_1(x,s)&\,x^{-\frac{1}{2}\sigma_3}=\frac{R(2s)}{x^2}+\begin{bmatrix}E_{11}&E_{12}\\ E_{21}& E_{22}\end{bmatrix}+\frac{1}{x^2}\begin{cases}\mathcal{O}(\max\{x^4,x^{4+8s}\}),&2s\notin\mathbb{Z}_{\geq 0}\\ \mathcal{O}(x^4),&2s\ \textnormal{even}\\ \mathcal{O}(\max\{x^4,x^{4+4s}|\ln(x^2)|\}),&2s\ \textnormal{odd}\end{cases},
	\label{f13}
\end{align}
where $E_{jk}$ are as in \eqref{f12}. If $s=0$, then
\begin{equation*}
	x^{\frac{1}{2}\sigma_3}Q_1(x,0)x^{-\frac{1}{2}\sigma_3}=\frac{1}{8x^2}\begin{bmatrix}\frac{9}{16} & \frac{39\im}{64}\smallskip\\ -\im & -\frac{9}{16}\end{bmatrix}.
\end{equation*}
\end{cor}

\begin{cor}\label{nzero} With $v=v(z;s)$ as in \eqref{e73}, we obtain from \eqref{f13}, as $z\downarrow 0$, provided $s\neq 0$,
\begin{align*}
	v(z;s)=&\,-\frac{z^{1+2s}}{\Gamma(2+2s)}\bigg(\frac{\Gamma(1+s)}{\Gamma(1+2s)}\bigg)^2\frac{1}{2\cos(\pi s)}
	+\mathcal{O}\big(\max\{z^2,z^{2+4s}\}\big),\ \ \ \ \ \ \ \ \ \ 2s\notin\mathbb{Z}_{\geq 0},\\
	v(z;s)=&\,-\frac{z^{1+2s}}{\Gamma(2+2s)}\bigg(\frac{\Gamma(1+s)}{\Gamma(1+2s)}\bigg)^2\frac{\e^{2\pi\im s}}{\pi}\Big[\psi(1+s)-\psi(2+2s)+\ln(2z)\Big]\sin(\pi s)\\
	&\hspace{5.65cm}+\mathcal{O}(\max\{z^2,z^{2+2s}|\ln z|\}),\ \ \ \ 2s\ \textnormal{odd},
\end{align*}
and $v(z;s)=\mathcal{O}(z^2)$ for even $2s\in\mathbb{Z}_{\geq 1}$. Here, $\psi(z)$ is the digamma function. Additionally, for any $z>0$,
\begin{equation}\label{f14}
	v(z;0)=-\frac{z}{2}.
\end{equation}
\end{cor}
\begin{proof} These follow from \eqref{e73} and \eqref{f13}. Namely, for $s\neq 0$, with $\mathcal{E}(x)$ shorthand for the error term in \eqref{f13}, we have
\begin{equation}\label{expressionforv}
v(z;s)\stackrel{\eqref{e73}}{=}-\im\sqrt{\frac{z}{\beta}}Q_1^{21}\Big(2\sqrt{\frac{z}{\beta}},s\Big)+\frac{1}{16}-s^2-\frac{z}{2}
\stackrel{\eqref{f13}}{=}-2\pi\im\frac{z}{\beta}\frac{\e^{-\im\pi s}\,\mathsf{H}_{s}(2\sqrt{z/\beta})}{2^{4s}\Gamma^2(1+2s)}+\mathcal{E}(2\sqrt{z/\beta}),
\end{equation}
with $\mathsf{H}_{s}(x)\in \{\mathsf{h}_{s}(x),\tilde{\mathsf{h}}_{s}(x)\}$ given in Proposition \ref{prop68}. We  now evaluate
\begin{equation*}
	\frac{2\pi\im}{\beta 2^{4s}}\e^{-\im\pi s}\,\mathsf{h}_{s}(2\sqrt{z/\beta})\stackrel{\eqref{f9}}{=}\frac{z^{2s}}{2\cos(\pi s)}\frac{\Gamma^2(1+s)}{\Gamma(2+2s)},\ \ \ \ 2s\notin\mathbb{Z}_{\geq 0},
\end{equation*}
and
\begin{align*}
	\frac{2\pi\im}{\beta 2^{4s}}\e^{-\im\pi s}\,\tilde{\mathsf{h}}_{s}(2\sqrt{z/\beta})&\,\stackrel{\eqref{f10}}{=}z^{2s}\frac{\e^{2\pi\im s}}{\pi}\sin(\pi s)\frac{\Gamma^2(1+s)}{\Gamma(2+2s)}\big(\psi(1+s)-\psi(2+2s)+\ln(2z)\big),\ \ \ 2s\in\mathbb{Z}_{\geq 1}.
\end{align*}
The derivation of \eqref{f14} was already achieved in Remark \ref{sospecial}, using $Q_1^{21}(x,0)=-\frac{\im}{8x}$. This completes the proof of the Corollary.
\end{proof}
The expansions for $v(z;s)$ near $z=0$ in Corollary \ref{nzero} are crude, except for \eqref{f14}. One can compute higher order terms by Remark \ref{structure} and iteration of the Neumann series underwriting \eqref{z1}. Specifically, we have the following asymptotic expansion of $v(z;s)$ as $z\downarrow 0$ with certain  coefficients depending on $s$. 
\begin{prop}\label{loopy} Let $2s>-1$. Then there exists $z_{0}=z_{0}(s)\in(0,1)$ such that for $0<z\leq z_{0}<1$ and any integer $K\geq 0$, we have the following asymptotic expansions: for even $2s\in\mathbb{Z}_{\geq 1}$, 
\begin{equation*}
	v(z;s)= z\bigg(\sum_{j=1}^s\mathsf{d}_{2j}(s)z^{2j-1}+\sum_{j=2s}^{K}\mathsf{d}_{j+1}(s)z^j\bigg)+\mathcal{O}_{K,s}\left(z^{K+2}\right),
\end{equation*}
and for odd $2s\in\mathbb{Z}_{\geq 1}$,
\begin{align*}
v(z;s) =&\,\,z\bigg(\sum_{j=1}^{s+\frac{1}{2}}\mathsf{d}_{2j}(s)z^{2j-1}+\sum_{j=2s+1}^{K}\mathsf{d}_{j+1}(s)z^j+\sum_{\ell=1}^{K+1}\sum_{j=\ell-1}^{K}\mathsf{d}_{j+1,\ell}(s)(\ln z)^{\ell}z^{j+2\ell s}\bigg)\\
&\hspace{1.5cm}+\mathcal{O}_{K,s}\left(\max\Big\{|\ln z|z^{2s+K+2},z^{K+2}\Big\}\right),
\end{align*}
and for $2s\notin\mathbb{Z}_{\geq 0}$, provided that $s\in(m,m+1],m\in\mathbb{Z}_{\geq 0}$ or $s\in(-\frac{1}{2},0]$, 
\begin{align*}
v(z;s)=&\,\,z\left(\sum_{j=1}^{m+1}\mathsf{d}_{2j}(s)z^{2j-1}+\sum_{j=2m+2}^{K}\mathsf{d}_{j+1}(s)z^j+\sum_{\ell=1}^{K+1}\sum_{j=\ell-1}^{K}\mathsf{d}_{j+1,\ell}(s)z^{j+2\ell s}\right)\\
&\hspace{1.5cm}+\mathcal{O}_{K,s}\left(\max\Big\{z^{(1+2s)(K+2)},z^{K+2}\Big\}\right),
\end{align*}
where we set $m=0$ when $s\in (-\frac{1}{2},0]$, and the coefficients $ \mathsf{d}_{j}(s),\mathsf{d}_{j,\ell}(s)$ depend only on $s$.
\end{prop}

\begin{proof}
Abbreviate
\beas
M(x,s)=\frac{\im}{2\pi}\int_{\Sigma_Z}Z_-(\mu)\big(G_Z(\mu)-I\big)\d\mu.
\eeas
By \eqref{e73},\eqref{expressionofQ1} and \eqref{expressionforv}, we have 
\bea\label{expressionofv}
v(z;s)=-\frac{\im}{2}M^{21}\left(2\sqrt{\frac{z}{\beta}},s\right)-\frac{z}{2}.
\eea
In the following, we only prove the conclusion for the case $2s \notin \mathbb{Z}_{\ge 0}$,  the other two cases follow from almost identical arguments. Choose $x_{0}=x_{0}(s)\in (0,1)$, such that $\beta x_{0}^2<\epsilon/2$ and \eqref{less1} holds. Firstly, we claim that
\bea\label{truncation}
M(x,s)=\sum_{n=0}^{N}\frac{\im}{2\pi}\int_{\Sigma_Z}\mathcal{C}_{G}^{n}[I](\mu)\big(G_Z(\mu)-I\big)\d\mu+\mathcal{O}_{s}\left(\max\Big\{x^{2(N+2)},x^{(2+4s)(N+2)}\Big\}\right)
\eea
holds for any $x\in [0,x_{0})$ and $N\in\mathbb{Z}_{\geq 0}$. Indeed, by \eqref{normofcg},\eqref{neumann} and Cauchy-Schwarz inequality, 
\begin{align*}
&M(x,s)-\sum_{n=0}^{N}\frac{\im}{2\pi}\int_{\Sigma_Z}\mathcal{C}_{G}^{n}[I](\mu)\big(G_Z(\mu)-I\big)\d\mu
\leq \frac{1}{2\pi}\left\|\sum_{n=N+1}^{\infty}\mathcal{C}_{G}^{n}[I](\cdot;x,s)\right\|_{L^2(\Sigma_Z)}\|G_Z(\cdot;x,s)-I\|_{ L^2(\Sigma_Z)}\\
&\hspace{2cm}\leq c\,\|G_Z(\cdot;x,s)-I\|_{L^{\infty}(\Sigma_Z)}^{N+1}\|G_Z(\cdot;x,s)-I\|_{ L^2(\Sigma_Z)},\ \ \ \ \ c=c(\Sigma_Z)>0.
\end{align*}
Then (\ref{truncation}) follows from \eqref{less1}. Secondly, we claim that for each $n\in\mathbb{Z}_{\geq 0}$,
\begin{align}\label{seriesforM}
\frac{\im}{2\pi}\int_{\Sigma_Z}\mathcal{C}_{G}^{n}[I](\mu)\big(G_Z(\mu)-I\big)\d\mu
=
\sum_{k=n+1}^\infty \mathsf{d}_{k}^{(n)}(s)
\left(\beta x^2\right)^k+\sum_{\ell=1}^{n+1}\sum_{k=n+1}^\infty \mathsf{d}_{k,\ell}^{(n)}(s)\left(\beta x^2\right)^{2\ell s+k},
\end{align}
where $\mathsf{d}_{k}^{(n)}(s), \mathsf{d}_{k,\ell}^{(n)}(s)=\mathcal{O}_n\left(k^{(n+1)|s|/2+2n}\epsilon^{-k}\right)$, and
by the choice of $x_{0}$, the series is absolutely and uniformly convergent for $x \in [0, x_{0}]$. Indeed, by \eqref{explicitfornonint}, we know that each entry of $G(\zeta;x,s)-I$ has the following form, 
\begin{align}\label{entryform}
\sum_{k=1}^\infty \mathsf{a}_{k}(s)\frac{f(\zeta)}{\zeta^{k}}\left(\beta x^2\right)^k+x^{4s}\sum_{k=1}^\infty \mathsf{b}_{k}(s)\frac{g(\zeta)}{\zeta^{k}}\left(\beta x^2\right)^k,
\end{align}
where $
\mathsf{a}_{k}(s),\mathsf{b}_{k}(s)=\mathcal{O}_{s}(k^{|s|/2})$,
and $f(\zeta)$, $g(\zeta)$ are analytic in  $\mathbb{D}_{\epsilon}(0)$ and given in terms of the entries of $E_{W}(\zeta)$ 
and $E_{W}^{-1}(\zeta)$. Moreover, the series \eqref{entryform} is uniformly convergent for $\zeta \in \Sigma_{Z}$ for any $x\in (0,x_{0}]$. Next, by the definition of $\mathcal{C}_{G}$ as in \eqref{defofcg}, for $\zeta\in \Sigma_{Z}$, 
\begin{align}
\lim_{\substack{\tau\to\zeta\\ \tau\in\text{side}\,-}}
    \frac{1}{2\pi \im}\int_{\Sigma_Z} \frac{\eqref{entryform}}{\mu - \tau}\d\mu
=\lim_{\substack{\tau\to\zeta\\ \tau\in\text{side}\,-}}\frac{1}{2\pi \im}\sum_{k=1}^{\infty}\left(\mathsf{a}_{k}(s)\int_{\Sigma_Z} \frac{f(\mu)}{\mu^{k}(\mu - \tau)}\d\mu+x^{4s}\,\mathsf{b}_{k}(s)\int_{\Sigma_Z} \frac{g(\mu)}{\mu^{k}(\mu - \tau)}\d\mu\right)(\beta x^2)^{k}\label{Cauchyintegral}.
\end{align}
For fixed $\tau$ inside of $\Sigma_{Z}$, noting that $\Sigma_{Z}=\partial\mathbb{D}_{\epsilon}(0)$ is oriented clockwise, Cauchy's integral formula yields,
\begin{align*}
\frac{1}{2\pi \im}\int_{\Sigma_Z}\frac{f(\mu)}{\mu^{k}(\mu - \tau)}\d\mu=-\frac{f(\tau)}{\tau^k}-\frac{1}{(k-1)!}\frac{\d^{k-1}}{\d\mu^{k-1}}\left(\frac{f(\mu)}{\mu-\tau}\right)\Big|_{\mu=0}=\sum_{j=0}^{k-1}\frac{f^{(k-1-j)}(0)}{(k-1-j)!}\frac{1}{\tau^{j+1}}-\frac{f(\tau)}{\tau^k},
\end{align*}
and consequently
\begin{align*}
\eqref{Cauchyintegral}=\lim_{\substack{\tau\to\zeta\\ \tau\in\text{side}\,-}}\sum_{k=1}^{\infty}\left(\mathsf{a}_k(s)\left[\sum_{j=1}^{k}\frac{f^{(k-j)}(0)}{(k-j)!}\frac{1}{\tau^{j}}-\frac{f(\tau)}{\tau^k}\right]+x^{4s}\,\mathsf{b}_k(s)\left[\sum_{j=1}^{k}\frac{g^{(k-j)}(0)}{(k-j)!}\frac{1}{\tau^{j}}-\frac{g(\tau)}{\tau^k}\right]\right)(\beta x^2)^{k}.
\end{align*}
Since $f(\mu)$ and $g(\mu)$ are analytic in $\mathbb{D}_{\epsilon}(0)$, $|f^{(j)}(0)/j!|$, $|g^{(j)}(0)/j!|=\mathcal{O}(\epsilon^{-j})$ for any $j\geq 0$ and $|f(\tau)|,|g(\tau)|=\mathcal{O}(1)$. According to our choice, we have $\beta x^2/\epsilon < 1/2$ for any $x \in (0, x_0]$. Hence, the limit with respect to $\tau$ can be interchanged with the infinite series in $k$, and so each entry in $\mathcal{C}_{G}[I](\zeta;x,s)$ is of the form
\begin{align}\label{entryform1}
\sum_{k=1}^\infty \left(\sum_{j=1}^{k}\mathsf{a}_{k,j}^{(1)}(s)\frac{1}{\zeta^{j}}-\mathsf{a}_k(s)\frac{f(\zeta)}{\zeta^k}\right)\left(\beta x^2\right)^k+x^{4s}\sum_{k=1}^\infty \left(\sum_{j=1}^{k}\mathsf{b}_{k,1,j}^{(1)}(s)\frac{1}{\zeta^{j}}-\mathsf{b}_k(s)\frac{g(\zeta)}{\zeta^k}\right)\left(\beta x^2\right)^k
\end{align}
where $\mathsf{a}_{k,j}^{(1)}(s), \mathsf{b}_{k,1,j}^{(1)}(s)=\mathcal{O}\left(k^{|s|/2}\epsilon^{-(k-j)}\right)$. Note that for each $n\geq 2$,
\begin{align*}
\mathcal{C}_{G}^{n}[I](\zeta;x,s)=\lim_{\substack{\tau\to\zeta\\ \tau\in\text{side}\,-}}
    \frac{1}{2\pi \im}\int_{\Sigma_Z}\mathcal{C}_{G}^{n-1}[I](\mu;x,s)\big(G(\mu;x,s)-I\big)\frac{\d\mu}{\mu-\tau}.
\end{align*}
By an argument similar to that of (\ref{entryform1}), and using induction, we obtain that each entry of $\mathcal{C}_{G}^{n}[I](\zeta;x,s)$ can be written as a sum of $\mathcal{O}_n(1)$ many terms of the form
\begin{align}\label{entryformforgeneraln}
\sum_{k=n}^\infty \sum_{j=1}^{k}\left(\mathsf{a}_{k,j}^{(n)}(s)-\tilde{\mathsf{a}}_{k,j}^{(n)}(s)f_n(\zeta)\right)\frac{(\beta x^2)^k}{\zeta^{j}}+\sum_{\ell=1}^{n}\sum_{k=n}^\infty \sum_{j=1}^{k}\left(\mathsf{b}_{k,\ell,j}^{(n)}(s)-\tilde{\mathsf{b}}_{k,\ell,j}^{(n)}(s)g_n(\zeta)\right)\frac{(\beta x^2)^{2\ell s+k}}{\zeta^{j}},
\end{align}
where 
\begin{equation*}
	\mathsf{a}_{k,j}^{(n)}(s), \mathsf{b}_{k,\ell,j}^{(n)}(s),\tilde{\mathsf{a}}_{k,j}^{(n)}(s),\tilde{\mathsf{b}}_{k,\ell,j}^{(n)}(s)=\mathcal{O}_n\left(k^{n|s|/2+2(n-1)}\epsilon^{-(k-j)}\right),
\end{equation*}
and $f_n(\zeta),g_n(\zeta)$ are certain analytic functions in $\mathbb{D}_{\epsilon}(0)$. By our choice of $x_{0}$ such that $\beta x_{0}^2/\epsilon < 1/2$, the series expansion of $\mathcal{C}_{G}^{n}[I](\zeta;x,s)$ in powers of $x$ is uniformly convergent with respect to $\zeta \in \Sigma_{Z}$ for any $x \in (0, x_{0}]$. Then, after multiplying by $G(\zeta; x, s) - I$, whose entries are of the form (\ref{entryform}), we can rewrite each entry of $\mathcal{C}_{G}^{n}[I](\zeta; x, s)(G(\zeta; x, s) - I)$  as a new series having a form similar to \eqref{entryformforgeneraln} with $n$ replaced by $n+1$. Applying the Cauchy integral formula to each term in this series, we obtain (\ref{seriesforM}), as claimed. Thirdly, by \eqref{truncation} and \eqref{seriesforM}, we obtain that for $N\in\mathbb{Z}_{\geq 0}$, and $x\in (0, x_{0}]$,
\begin{align}\label{seriesexpansionforM}
M(x,s)=&\sum_{k=1}^{N+1}\tilde{\mathsf{d}}_{k}(s)
\left(\beta x^2\right)^k+\sum_{\ell=1}^{N+1}\sum_{k=\ell}^{N+1}\tilde{\mathsf{d}}_{k,\ell}(s)\left(\beta x^2\right)^{2\ell s+k}
+\mathcal{O}_{N,s}\left(\max\Big\{x^{(2+4s)(N+2)},x^{2(N+2)}\Big\}\right),
\end{align}
where $\tilde{\mathsf{d}}_{k}(s)=\sum_{n=0}^{k-1}\mathsf{d}_{k}^{(n)}(s)$ and  $\tilde{\mathsf{d}}_{k,\ell}(s)=\sum_{n=\ell-1}^{k-1}\mathsf{d}_{k,\ell}^{(n)}(s)$ with $\mathsf{d}_{k}^{(n)}(s),\mathsf{d}_{k,\ell}^{(n)}(s)$ as in \eqref{seriesforM}. Finally, we return to the asymptotic expansion of $v(z, s)$ as $z \downarrow 0$, as claimed in the Proposition. By \eqref{expressionofv} and \eqref{seriesexpansionforM}, for $2s\notin \mathbb{Z}_{\geq 0}$,
\begin{align}\label{seriesexpansionforv}
v(z,s)=&\sum_{k=1}^{N+1}\mathsf{d}_{k}(s)
z^k+\sum_{\ell=1}^{N+1}\sum_{k=\ell}^{N+1}\mathsf{d}_{k,\ell}(s)z^{2\ell s+k}+\mathcal{O}_{N,s}\left(\max\Big\{z^{(1+2s)(N+2)},z^{N+2}\Big\}\right).
\end{align}
By Corollary \ref{nzero}, we have $\mathsf{d}_{1}(s)=0$. Choose $N=K$. Set $z_{0}=\beta x_{0}^2/4$. Note that $v(z,s)$ satisfies the $\sigma$-Painlev\'e III$'$ equation (\ref{e42}) by Corollary \ref{cheat}. Substituting the expansion (\ref{seriesexpansionforv}) into (\ref{e42}) together with the initial condition $\mathsf{d}_{1}(s)=0$, we obtain the gap structure $\mathsf{d}_{2j-1}(s)=0$ for any $2\leq j<s+1$. Therefore, we derive the asymptotic formula for $v(z,s)$ as claimed above. This completes the proof.
 \end{proof}

At the end of this section, we state the following result, which is a direct consequence of Proposition \ref{loopy} and will be used in the proof of Theorem \ref{maintheorem1}.
\begin{cor}\label{afterintegral} Let $2s>-1$. Then there exists $z_{0}=z_{0}(s)\in(0,1)$ such that for $0<z\leq z_{0}<1$, we have the following asymptotic expansions: for fixed even $2s\in\mathbb{Z}_{\geq 1}$, 
\begin{equation*}
	\int_{0}^{z}v(x;s)\frac{\d x}{x}= \sum_{k=1}^s\mathsf{d}_{2k}(s)\frac{z^{2k}}{2k}+\mathcal{O}(z^{2s+1}),
\end{equation*}
and for fixed odd $2s\in\mathbb{Z}_{\geq 1}$,
\begin{align*}
\int_{0}^{z}v(x;s)\frac{\d x}{x} =&\sum_{k=1}^{s+\frac{1}{2}}\mathsf{d}_{2k}(s)\frac{z^{2k}}{2k}+\mathcal{O}\left(\int_{0}^{z}t^{2s}\ln t\,\d t\right).
\end{align*}
Lastly for fixed $2s\notin\mathbb{Z}_{\geq 0}$, provided that $s\in(m,m+1],m\in\mathbb{Z}_{\geq 0}$ or $2s\in(-1,0)$, 
\begin{align*}
\int_{0}^{z}v(x;s)\frac{\d x}{x}=\sum_{k=1}^{m+1}\mathsf{d}_{2k}(s)\frac{z^{2k}}{2k}+\mathcal{O}\left(z^{2s+1}\right),
\end{align*}
where we set $m=0$ when $2s\in (-1,0)$.
\end{cor}

\section{Proofs of Theorem \ref{expf} and Proposition \ref{1103addprop}}\label{sectiononexpf}\label{sectionofprovingTheoremexpf}

Before proving Theorem \ref{maintheorem1}, we introduce certain connections between the joint moments $F_{N}(s,h)$ and the family of real-valued random variables $\mathsf{X}(s)$ defined as in Definition \ref{defofxs}. Let $X_N$ be a random matrix in $\mathbf{H}(N)$, the set of all $N\times N$ complex Hermitian matrices, with 
\begin{equation*}
	\textnormal{Law}(X_N) = m^{(s)}_N,
\end{equation*}
the finite dimensional Hua-Pickrell measure. The distribution of the eigenvalues of this ensemble is given by the following probability measure on $\mathbb{R}^N/S_N$, with $S_N$ the symmetric group of degree $N$, 
\begin{equation}\label{HP0}
\frac{1}{\mathrm{D}_N^{(s)}}\prod_{1\leq j\leq k\leq N}|x_{k}-x_{j}|^2
\prod_{j=1}^N (1+x_j^2)^{-s-N}\d x_{j},
\end{equation}
where $\mathrm{D}_{N}^{(s)}$ is the normalising constant, which is given explicitly by
\begin{equation*}
 	\mathrm{D}_N^{(s)} = {\pi^N2^{-N(N+2s -1)}} \prod_{j=0}^{N-1} \frac{j!\,\Gamma(2s + N-j)}{\Gamma(s+N-j)^2}.
\end{equation*} 
It was proven by Borodin and Olshanski \cite[Theorem 5.1]{Borodin_2001} that ${\rm Tr}(X_N)/N$ converges almost surely and then in distribution, as $N\rightarrow\infty$, to a random variable, denoted by $\mathsf{X}(s)$ for any $s$ with $s>-1/2$. The
explicit identification of $\mathsf{X}(s)$ as (\ref{defofran}) comes from Qiu \cite{Qiu}; see Theorems 1.2, 1.3, and in particular Theorem 2.3 therein. For fixed $s$, we denote by $\mathbb{E}_{N}^{(s)}[\cdot]$ the expectation with respect to $m_{N}^{(s)}$, and by $\mathbb{E}[\cdot]$ the expectation with respect to the Hua-Pickrell measure on $\mathbf{H}(\infty)$, the space of infinite-dimensional Hermitian matrices. Note that the Hua-Pickrell measure is defined as the projective limit of $m_{N}^{(s)}$ on $\mathbf{H}(N)$, so 
\begin{equation*}
	\mathbb{E}\left[\e^{\frac{\im t}{2N}\rm Tr(X_{N})}\right]=\mathbb{E}_{N}^{(s)}\left[\e^{\frac{\im t}{2N}\rm Tr(X_{N})}\right],\ \ \ t\in\mathbb{R}.
\end{equation*}
 What follows is
\begin{equation}\label{the convergence for the expectation}
\lim_{N\rightarrow \infty}\mathbb{E}_{N}^{(s)}\left[\e^{\frac{\im t}{2N} \rm{Tr}(X_{N})}\right]=\mathbb{E}\left[\e^{\frac{\im t}{2}\mathsf{X}(s)}\right],
\end{equation}
uniformly on compact subsets of $\mathbb{R}\ni t$. There is a connection between the joint moments $F_{N}(s,h)$ and the expectation of the trace of random Hermitian matrices with respect to the finite-dimensional Hua-Pickrell measure. Specifically, we have the below identity, given in \cite[Proposition 2.7]{AKW}:
\begin{equation}\label{alternativeexpression}
\frac{F_N(s,h)}{N^{s^2+2h}}
= \frac{F_N(s,0)}{N^{s^2}} \, 2^{-2h} \,
\mathbb{E}_{N}^{(s)}\!\left[\left|\frac{{\rm Tr}(X_{N})}{N}\right|^{2h}\right].
\end{equation}
When taking the limit as $N \to \infty$ in \eqref{alternativeexpression}, we have the following result, which naturally extends Lemma \ref{convergenceat0} to the complex exponents $h\in \mathbb{C}$.
\begin{prop}\label{exofh}
Let $s>-\frac{1}{2}$. Let $h\in \mathbb{C}$ with $0\leq \textnormal{Re}(h)<s+\frac{1}{2}$. Then
\begin{equation*}
\lim_{N\rightarrow\infty}\frac{F_N(s,h)}{N^{s^2+2h}}=\frac{G(s+1)^2}{G(2s+1)}2^{-2h}\,\mathbb{E}\left[\left|\mathsf{X}(s)\right|^{2h}\right]<\infty.
\end{equation*}
\end{prop}

\begin{proof}
Note that for $s>-\frac{1}{2}$, by \cite{keating2000random}, 
\beas
\lim_{N\rightarrow \infty}\frac{F_{N}(s,0)}{N^{s^2}}=\frac{G(s+1)^2}{G(2s+1)}.
\eeas
So by (\ref{alternativeexpression}), to prove the claim of the current Proposition, it is sufficient to show that for $s>-\frac{1}{2}$ and $0\leq \textnormal{Re}(2h)<2s+1$,
\bea\label{limitexisits}
\lim_{N\rightarrow \infty}\mathbb{E}_{N}^{(s)}\!\left[\mathfrak{y}_N\right]=\mathbb{E}\left[\mathfrak{y}\right],\ \ \ \ \mathfrak{y}_N:=\frac{1}{N}|\textnormal{Tr}(X_N)|^{2h},\ \ \ \mathfrak{y}:=|\mathsf{X}(s)|^{2h}.
\eea
Firstly, we prove that $(\ref{limitexisits})$ holds for $s>0$ and $h\in \mathbb{C}$ with $0\leq \textnormal{Re}(h)<s+1/2$. By the continuous mapping theorem, $\mathfrak{y}_{N}$ converges in law, induced by the infinite-dimensional Hua-Pickrell measure, to $\mathfrak{y}$. For fixed $h$, choose a suitable $\delta>0$ such that $2\textnormal{Re}(h)(1+\delta)<2s+1$, then by \cite[(17)]{AKW}, we have
\begin{equation*}
	\sup_{N\geq 1}\mathbb{E}_{N}^{(s)}\Big[\left|\mathfrak{y}_{N}\right|^{1+\delta}\Big]<\infty,
\end{equation*}
so the sequence $\{\mathfrak{y}_{N}\}_{N\geq 1}$ is uniformly integrable, and thus $\mathbb{E}[\mathfrak{y}_{N}]\rightarrow\mathbb{E}[\mathfrak{y}]$ as $N\rightarrow \infty$. So $(\ref{limitexisits})$ holds for $s>0$ and $h\in \mathbb{C}$ with $0\leq \textnormal{Re}(h)<s+1/2$. 

We next show that \eqref{limitexisits} also holds for $-1/2<s\leq 0$, $h\in \mathbb{C}$ with $0\leq \textnormal{Re}(h)<s+1/2$. To that end we recall Scheff\'e's lemma, which states that if a sequence of the random variables $\{\mathfrak{z}_{n}\}_{n=1}^{\infty}$ with finite expectations converges almost surely to another random variable $\mathfrak{z}$ with a finite expectation, and $\lim_{n\rightarrow \infty}\mathsf{E}[|\mathfrak{z}_{n}|]=\mathsf{E}[|\mathfrak{z}|]$, where $\mathsf{E}[\cdot]$ is the expectation, then $\lim_{n\rightarrow \infty}\mathsf{E}[\mathfrak{z}_{n}]=\mathsf{E}[\mathfrak{z}]$. Note that ${\rm Tr}(X_N)/N$ converges almost surely to $\mathsf{X}(s)$, as $N\rightarrow\infty$, see \cite[Theorem 5.1]{Borodin_2001}. By the continuous mapping theorem, we have that $\mathfrak{y}_{N}$ convergence almost surely to $\mathfrak{y}$. Note that by the proof of \cite[Theorem 1.2]{AKW}, we have 
\begin{equation*}
	\mathbb{E}[|\mathfrak{y}_{N}|]=\mathbb{E}\left[\left|\textnormal{Tr}(X_{N})/N\right|^{2\textnormal{Re} (h)}\right]=\mathbb{E}_{N}^{(s)}\left[\left|\textnormal{Tr}(X_{N})/N\right|^{2\textnormal{Re} (h)}\right]<\infty,
\end{equation*}
and $\mathbb{E}[|\mathfrak{y}|]=\mathbb{E}[|\mathsf{X}(s)|^{2\textnormal{Re}(h)}]<\infty$, as well as $\lim_{N\rightarrow \infty}\mathbb{E}_{N}^{(s)}[|\mathfrak{y}_{N}|]=\mathbb{E}[|\mathfrak{y}|]$. So, by Scheff\'e's theorem, 
\begin{equation*}
	\lim_{N\rightarrow \infty}\mathbb{E}_{N}^{(s)}[\mathfrak{y}_{N}]=\mathbb{E}[\mathfrak{y}].
\end{equation*}
So \eqref{limitexisits} also holds for $-1/2<s\leq 0$, $h\in \mathbb{C}$ with $0\leq \textnormal{Re}(h)<s+1/2$. The proof of the Proposition is now complete.
\end{proof}
Next, we connect $\mathbb{E}[\e^{\frac{\im t}{2}\mathsf{X}(s)}]$ to $v(t;s)$, defined as in (\ref{e73}). Specifically, this connection is given in (\ref{1008e6}). We begin with the following Lemma.

\begin{lem}\cite[Proposition 2.4]{ABGS}\label{theanalyticfunction}
Let $s>-\frac{1}{2}$. For $z\in \mathbb{H}_{+}:=\{z\in\mathbb{C}:\,\textnormal{Re}(z)>0\}$, define
\begin{equation}\label{definitionofFN}
\mathcal{F}_{N}(z):=\frac{\e^{-\frac{z}{2}}}{N!\prod_{j=1}^{N}\Gamma(j)\Gamma(2s+j)}\int_{0}^{\infty}\cdots\int_{0}^{\infty}\Delta^2(y_{1},\ldots,y_{N})\prod_{j=1}^{N}\left(y_j+\frac{z}{N}\right)^{s}y_j^{s}\e^{-y_j}\d y_j,
\end{equation}
where $\Delta(y_1,\ldots,y_N)$ is the Vandermonde determinant. Then $\mathcal{F}_{N}(z)$ is analytic on $\mathbb{H}_{+}$ and there is an analytic function $\mathcal{F}(z)$ on
$\mathbb{H}_{+}$, such that
$\mathcal{F}(0)=1$ and for $t\in (0,\infty)$,
\begin{equation}
\mathcal{F}(t)=\mathbb{E}\left[\e^{\frac{\im t}{2}\mathsf{X}(s)}\right].
\end{equation}
Moreover, for $z\in \mathbb{H}_{+}$ and $m=0,1,2,\ldots$, we have
\begin{equation*}
\lim_{N\rightarrow \infty}\frac{\d^{m}}{\d z^{m}}\mathcal{F}_{N}(z)=\frac{\d^{m}}{\d z^{m}}\mathcal{F}(z),
\end{equation*}
uniformly on compact subsets of $\mathbb{H}_{+}$.
\end{lem}
\begin{prop}\label{smooth}
Let $s>-\frac{1}{2}$. Let $v_{N}(t;s)$ and $v(t;s)$ be as in \eqref{defofvN} and \eqref{e73}. Then $t\mapsto v_{N}(t;s)$ and $t\mapsto v(t;s)$ are smooth on $(0,+\infty)\subset\mathbb{R}$, and we have
\begin{equation}\label{convergenceofthederivative}
\lim_{N\rightarrow \infty}\frac{\d^{m}}{\d t^{m}}v_{N}(t;s)=\frac{\d^{m}}{\d t^{m}}v(t;s).
\end{equation}
pointwise in $t\in(0,+\infty)$ and $m\in\mathbb{Z}_{\geq 0}$.
Moreover, for $t\geq 0$,
\begin{equation}\label{1008e6}
\mathbb{E}\Big[\e^{\frac{\im t}{2}\mathsf{X}(s)}\Big]=\exp\left[\int_0^{t}v(x;s)\frac{\d x}{x}\right].
\end{equation}
\end{prop}
\begin{proof}
Let $\mathcal{F}_{N}(z)$ be as in \eqref{definitionofFN}. Note that $v_N(z;s)=u_N(z/N;s)-z/2$. By the definition of $\mathbb{E}_{N}^{(s)}[\cdot]$, 
\bea\label{expressionforEN}
\mathbb{E}_N^{(s)}\left[\e ^{\frac{\im t}{2N} {\rm Tr}(X_N)}\right]
=J_{N}(0,s)^{-1}J_{N}\left(\frac{t}{2N},s\right).
\eea
with $J_{N}(\xi,s)$ as in \eqref{defofJ}. Then by \eqref{i13},\eqref{1102defu} and \eqref{expressionforEN}, we have for $t\in[0,\infty)$,
\begin{equation*}\label{1008e1}
\mathcal{F}_{N}(t)=\exp\left(\int_{0}^{t}v_{N}(x;s)\frac{dx}{x}\right)=\mathbb{E}_N^{(s)}\left[\e^{\frac{\im t}{2N} {\rm Tr}(X_N)}\right].
\end{equation*}
By the proof workings of \cite[Theorem 1.2]{ABGS}, $t\mapsto\mathbb{E}_N^{(s)}[\e ^{\frac{\im t}{2N} {\rm Tr}(X_N)}]$ and $t\mapsto\mathbb{E}[\e^{\frac{\im t}{2}\mathsf{X}(s)}]$ are strictly positive on $\mathbb{R}$. Consequently, for $t\in(0,+\infty)$, we have
\begin{equation*}\label{1008e2}
v_{N}(t;s)=t\frac{\d}{\d t}\ln \mathcal{F}_{N}(t).
\end{equation*}
By (\ref{e73x}) and Lemma \ref{theanalyticfunction}, for $t\in(0,+\infty)$,
\begin{equation}\label{1008e3}
v(t;s)=t\frac{\d}{\d t} \ln \mathbb{E}\left[\e^{\frac{\im t}{2}\mathsf{X}(s)}\right].
\end{equation}
By Corollaries \ref{cor5555} and \ref{nzero}, see also Remark \ref{mero},
\begin{equation*}
	t\mapsto\exp\left[\int_0^{t}v(x;s)\frac{\d x}{x}\right]
\end{equation*}
is well defined and it tends to $1$ as $t\downarrow0$. Note that $\mathbb{E}[\e^{\frac{\im t}{2}\mathsf{X}(s)}]|_{t=0}=1$, so by \eqref{1008e3}, we have \eqref{1008e6}. The smoothness of $ v_{N}(t; s)$ and $v(t; s)$, together with \eqref{convergenceofthederivative}, is a consequence of Lemma \ref{theanalyticfunction}.
\end{proof}

Now we are ready to prove Theorem \ref{expf} and then Proposition \ref{1103addprop}. 
\begin{proof}[Proof of Theorem \ref{expf}]
By (\ref{1008e6}) and Corollary \ref{cor5555}, we obtain that $\mathbb{E}[\e^{\frac{\im t}{2}\mathsf{X}(s)}]$ decays exponentially fast when $t\rightarrow+\infty$. Let
\begin{equation}\label{defof}
\rho^{(s)}(x):=\frac{1}{2\pi}\int_{-\infty}^{\infty}\e^{\im xt}\,\mathbb{E}\Big[\e^{\im t\mathsf{X}(s)}\Big]\d t=\frac{1}{\pi}\,\textnormal{Re}\left(\int_{0}^{\infty}\e^{\im xt}\,\mathbb{E}[\e^{\im t\mathsf{X}(s)}]\d t\right),
\end{equation}
where the second equality follows from the fact that $t\mapsto \mathbb{E}[\e^{\im t\mathsf{X}(s)}]$ is even. By the exponential decay and continuity of $t\mapsto\mathbb{E}[\e^{\im t\mathsf{X}(s)}]$, \eqref{defof} is well defined for any $x\in \mathbb{R}$, that is, the distribution of $\mathsf{X}(s)$ has a bounded and continuous probability density function with respect to the Lebesgue measure on $\mathbb{R}$. Moreover, the same density is smooth.
\end{proof}

\begin{proof}[Proof of Proposition \ref{1103addprop}]
It follows from (\ref{defof}) and (\ref{1008e6}).
\end{proof}

From Proposition \ref{exofh}, for \(h \in \mathbb{C}\) with \(0 \leq \mathrm{Re}(h) < \frac{1}{2} + s\), we see that the moment \(\mathbb{E}[|\mathsf{X}(s)|^{2h}]\) is finite. In fact, by Theorem \ref{expf}, the finiteness can be extended to the range $s>-\frac{1}{2}$ and \(-\frac{1}{2} < \mathrm{Re}(h) < \frac{1}{2} + s\). This result was previously known for non-negative integer values of \(s\) in \cite{ABGS}.

\begin{prop}\label{propexpectation}
Let $s>-\frac{1}{2}$. Let $\rho^{(s)}(x)$ be given as in \eqref{defof}. Let $h\in\mathbb{C}$ with $-\frac{1}{2}<\textnormal{Re}(h)<\frac{1}{2}+s$. Then
\begin{equation}\label{expectationexpression}
\mathbb{E}\Big[|\mathsf{X}(s)|^{2h}\Big]=\int_{-\infty}^{\infty}|x|^{2h}\rho^{(s)}(x)\d x<\infty.
\end{equation}
\end{prop}
\begin{proof}
The first equality in \eqref{expectationexpression} follows from Theorem \ref{expf}. By Proposition \ref{exofh}, for \(h \in \mathbb{C}\) with \(0 \leq \mathrm{Re}(2h) < 2s+1\), we have \(\mathbb{E}[|\mathsf{X}(s)|^{2h}] < \infty\). For \(h \in \mathbb{C}\) with \(-1 < \mathrm{Re}(2h) < 0\), we use the following two facts: 
\begin{equation*}
	0 \leq \rho^{(s)}(x) \leq 1\ \ \ \forall\,x\in\mathbb{R}\hspace{1cm}\textnormal{and}\hspace{1cm}\rho^{(s)}(x) = \mathcal{O}\!\left(x^{-1}\right)\ \ \textnormal{as}\ |x|\rightarrow\infty.
\end{equation*}	
The decay follows from integration by parts in \eqref{defof} and from the exponential decay of \(\frac{\d}{\d t}\mathbb{E}[\mathrm{e}^{\im t \mathsf{X}(s)}]\) as \(|t| \to \infty\), established via \eqref{1008e6} and Corollary \ref{cor5555}.
\end{proof}

\begin{rem}\label{analyticfunction}
By Proposition \ref{propexpectation}, applying an argument similar to that in \cite[(71)]{ABGS}, we have that $h\mapsto \mathbb{E}[|\mathsf{X}(s)|^{2h}]$ is analytic on $
\{ h \in \mathbb{C} : -1 < \mathrm{Re}(2h) < 2s + 1 \}$, and so is $h\mapsto\lim_{N\rightarrow\infty}F_N(s,h)/N^{s^2+2h}$ by Proposition \ref{exofh}.
\end{rem}

\section{Proofs of Theorem \ref{maintheorem1}, Propositions \ref{1103prop} and \ref{rule}}

We begin with an integral representation of the moments of the random variable \( \mathsf{X}(s) \) expressed in terms of the Painlev\'e function \eqref{e42},\eqref{e73repeat}.

\begin{lem}\label{fractionalmoments}
Let $s>-\frac{1}{2}$ and $h\in\mathbb{C}$ with $-\frac{1}{2}<\textnormal{Re}(h)<\frac{1}{2}+s$. 
Then
\begin{equation}\label{1008even1}
\mathbb{E}[|\mathsf{X}(s)|^{2h}]=2\lim_{\epsilon\downarrow 0}\int_{0}^{\infty}\exp\left[\int_0^{2t}v(x;s)\frac{\d x}{x}\right]K_{2h}^{\epsilon}(t)\d t,
\end{equation}
with $v(x;s)$ and $K_{2h}^{\epsilon}(t)$ as in \eqref{e73} and \eqref{i5}, respectively.
\end{lem}

\begin{proof}
Let $\rho^{(s)}(x)$ be the density of the random variable $\mathsf{X}(s)$, given in \eqref{defof}. Note that $\rho^{(s)}(x)$ is non-negative and by Proposition \ref{propexpectation},
\begin{equation*}
\int_{-\infty}^{\infty}|x|^{2\textnormal{Re}(h)}\rho^{(s)}(x)\d x<+\infty.
\end{equation*}
So, by the dominated convergence theorem,  for $h\in\mathbb{C}$ with $-\frac{1}{2}<\textnormal{Re}(h)<\frac{1}{2}+s$,
\begin{equation*}
\mathbb{E}\Big[|\mathsf{X}(s)|^{2h}\Big]=\lim_{\epsilon\downarrow 0}\int_{-\infty}^{\infty}|x|^{2h}\e^{-\epsilon|x|}\rho^{(s)}(x)\d x.
\end{equation*}
By the definition of $\rho^{(s)}(x)$,
\begin{equation*}
\mathbb{E}[|\mathsf{X}(s)|^{2h}]=\lim_{\epsilon\downarrow 0}\frac{1}{2\pi}\int_{-\infty}^{\infty}|x|^{2h}\e^{-\epsilon|x|}\left(\int_{-\infty}^{\infty}\e^{\im xt}\,\mathbb{E}\Big[\e^{\im t\mathsf{X}(s)}\Big]\d t\right)\d x.
\end{equation*}
From the exponential decay of $\mathbb{E}[\mathrm{e}^{\mathrm{i} t \mathsf{X}(s)}]$, obtained by \eqref{1008e6} and Corollary \ref{cor5555}, we have
\begin{equation*}
\int_{-\infty}^{\infty}\int_{-\infty}^{\infty}|x|^{2\textnormal{Re}(h)}\e^{-\epsilon|x|}\Big|\mathbb{E}\Big[\e^{\im t\mathsf{X}(s)}\Big]\Big|\d x\d t<\infty.
\end{equation*}
Thus, by Fubini's theorem,
\begin{equation*}
\mathbb{E}\Big[|\mathsf{X}(s)|^{2h}\Big]=\lim_{\epsilon\downarrow 0}\int_{-\infty}^{\infty}\mathbb{E}\Big[\e^{\im t\mathsf{X}(s)}\Big]\left(\frac{1}{2\pi}\int_{-\infty}^{\infty}|x|^{2h}\e^{-\epsilon|x|}\e^{\im xt}\d x\right)\d t\stackrel{\eqref{i5}}{=}\lim_{\epsilon\downarrow 0}\int_{-\infty}^{\infty}\mathbb{E}\Big[\e^{\im t\mathsf{X}(s)}\Big]K_{2h}^{\epsilon}(t)\d t.
\end{equation*}
Note that $\mathbb{E}[\e^{\im t\mathsf{X}(s)}]$ and $K_{2h}^{\epsilon}(t)$ are even functions of $t$, combined with (\ref{1008e6}), we conclude the claim in this lemma.
\end{proof}

\begin{remark}\label{rangeofFN}
For fixed $N\in\mathbb{Z}_{\geq 1}$, since \(\mathrm{Tr}(X_{N})\) is also a random variable when \(X_{N}\) is chosen randomly from
\begin{equation*} 
	\Big(\mathbf{H}(N), m_{N}^{(s)}\Big),
\end{equation*}
we can, by an argument similar to that of Lemma \ref{fractionalmoments}, obtain for \(s > -\tfrac{1}{2}\) and \(h \in \mathbb{C}\) with \(-\tfrac{1}{2} < \mathrm{Re}(h) < \tfrac{1}{2} + s\) that,
\begin{equation}\label{1105even1}
\mathbb{E}_{N}^{(s)}\!\Big[|\mathrm{Tr}(X_{N})|^{2h}\Big]
= 2 \lim_{\epsilon \downarrow 0} \int_{0}^{\infty}
\exp\!\left[\int_{0}^{2t}\!\!\left(u_{N}(x;s) - \frac{N x}{2}\right)\!\frac{\mathrm{d}x}{x}\right]
K_{2h}^{\epsilon}(t)\,\mathrm{d}t,
\end{equation}
where \(u_{N}(x;s)\) and \(K_{2h}^{\epsilon}(t)\) are as defined in \eqref{1102defu} and \eqref{i5}, respectively. In fact, one may perform a large-\(x\) asymptotic analysis on the corresponding Riemann--Hilbert problem for \(u_{N}(x;s)\), as was done in Corollary \ref{cor5555} for \(v(x;s)\).
Alternatively, since \(u_{N}(x;s)\) admits an explicit representation in terms of a Hankel determinant involving \(U(a,b,x)\) as defined in \eqref{1102defu}, one may use the known asymptotic expansion \cite[\S 13.7.3]{NIST}
\[
U(a,b,t) \sim t^{-a}\sum_{m=0}^{\infty}(-1)^{m}\frac{(a)_{m}(a-b+1)_{m}}{m!\,t^{m}}, \qquad t \to +\infty.
\]
This yields
\begin{equation}\label{asymptoticforu}
u_{N}(t;s) = N s + \mathcal{O}(t^{-1}), \qquad t \to \infty.
\end{equation}
Consequently, from \eqref{asymptoticforu} and \eqref{1008e1} we have
\[
\mathbb{E}_{N}^{(s)}\left[\e^{\im t\mathrm{Tr}(X_{N})}\right]
= \exp\!\left[\int_{0}^{2t}\!\!\left(u_{N}(x;t) - \frac{N x}{2}\right)\!\frac{\mathrm{d}x}{x}\right],\ \ \ t>0,
\]
which decays exponentially as \(t \to +\infty\). Hence, combining this with the facts that 
\begin{equation*}
	t\mapsto\mathbb{E}_{N}^{(s)}\left[\e^{\im t\, \mathrm{Tr}(X_{N})}\right]
\end{equation*}
is an even function which equals $1$ at $t=0$, we get \(t\mapsto\mathbb{E}_{N}^{(s)}\left[\e^{\im t\mathrm{Tr}(X_{N})}\right] \in L^{1}(\mathbb{R})\). Therefore, by \eqref{alternativeexpression}, \eqref{1105even1}, \eqref{expressionforEN}, and Proposition \ref{1104inprop}, we conclude that \eqref{1104i7} holds for \(2s > -1\) and \(h \in \mathbb{C}\) with \(-1 < \mathrm{Re}(2h) < 2s+1\).
\end{remark}

As a corollary of Lemma~\ref{fractionalmoments}, we now prove Proposition~\ref{1103prop}.
\begin{proof}[Proof of Proposition~\ref{1103prop}]
The result follows directly from Proposition~\ref{exofh} and Lemma~\ref{fractionalmoments}.
\end{proof}

We now prove Theorem~\ref{maintheorem1}. Before doing so, we need some preparations. For $n \in \mathbb{N}$ and $t>0$, define
\begin{equation}\label{defofu}
u^{(n)}(t) = u^{(n)}(t; s)
:= \frac{\mathrm{d}^n}{\mathrm{d}t^n}
\exp\!\left[
\int_{0}^{2t} v(x; s)\, \frac{\mathrm{d}x}{x}
\right].
\end{equation}

We now compute the right-hand side of (\ref{1008even1}) for $2h \in \mathbb{Z}_{\geq 1}$. 
\begin{prop}\label{lim1} Let $s>-\frac{1}{2}$ and $h\in\mathbb{C}:-\frac{1}{2}<\textnormal{Re}(h)<\frac{1}{2}+\textnormal{Re}(s)$ with $2h\in\mathbb{Z}_{\geq 1}$. Then we have
\begin{equation*}
	\lim_{\epsilon\downarrow 0}\int_0^{\infty}K_{2h}^{\epsilon}(t)\exp\left[\int_0^{2t}v(x;s)\frac{\d x}{x}\right]\d t=\begin{cases}\displaystyle \frac{1}{2}(-1)^h u^{(2h)}(0),&2h\ \textnormal{even}\bigskip\\ \displaystyle\frac{1}{\pi}(-1)^{h+\frac{1}{2}}\int_0^{\infty}u^{(2h)}(\lambda)\frac{\d\lambda}{\lambda},&2h\ \textnormal{odd}\end{cases},
\end{equation*}
with $u^{(n)}(t)$ as in \eqref{defofu}.
\end{prop}
\begin{proof} 
By Corollary \ref{cor5555} and Corollary \ref{afterintegral}, we have 
\begin{equation}\label{central}
	L^1\cap L^{\infty}(0,\infty)\ni t^{m}\frac{\d^n}{\d t^n}u(t)=\begin{cases}\mathcal{O}(1),&t\downarrow 0\smallskip\\ \mathcal{O}(\exp(-(1-\delta)t)),&t\rightarrow+\infty\end{cases}\ \ \ \forall\,n,m\in\mathbb{Z}_{\geq 0}:\,n\leq\textnormal{Re}(2s+1),
\end{equation}
where $\delta$ is a sufficiently small positive number.
Hence, by the dominated convergence theorem, applied twice to swap integration with differentiation,
\begin{align}
	\int_0^{\infty}K_{2h}^{\epsilon}(t)u(t)\d t\stackrel{\eqref{e38}}{=}\frac{(-1)^{2h}}{\pi}\frac{\partial^{2h}}{\partial\epsilon^{2h}}\int_0^{\infty}\frac{\epsilon u(t)}{\epsilon^2+t^2}\d t=&\,\,\frac{(-1)^{2h}}{\pi}\frac{\partial^{2h}}{\partial\epsilon^{2h}}\int_0^{\infty}\frac{u(\epsilon\lambda)}{1+\lambda^2}\d\lambda\nonumber\\
	=&\,\,\frac{(-1)^{2h}}{\pi}\int_0^{\infty}u^{(2h)}(\epsilon\lambda)\frac{\lambda^{2h}}{1+\lambda^2}\d\lambda.\label{e40}
\end{align}
However, by geometric progression,
\begin{equation*}
	\lambda^{2h}=(\lambda\mp\im)\sum_{k=0}^{2h-1}\lambda^k(\pm\im)^{2h-1-k}+(\pm\im)^{2h},\ \ \ 2h\in\mathbb{Z}_{\geq 1},
\end{equation*}
so that in turn for the integral in the right hand side of \eqref{e40},
\begin{align}
	\int_0^{\infty}&\,u^{(2h)}(\epsilon\lambda)\bigg(\frac{1}{\lambda-\im}-\frac{1}{\lambda+\im}\bigg)\lambda^{2h}\d\lambda\nonumber\\
	=&\,\,\sum_{k=0}^{2h-1}\big(\im^{2h-1-k}-(-\im)^{2h-1-k}\big)\int_0^{\infty}\lambda^ku^{(2h)}(\epsilon\lambda)\d\lambda+\im^{2h}\int_0^{\infty}\bigg(\frac{1}{\lambda-\im}-\frac{(-1)^{2h}}{\lambda+\im}\bigg)u^{(2h)}(\epsilon\lambda)\d\lambda.\label{e41}
\end{align}
But, integrating by parts $k\in\{0,\ldots,2h-1\}$ times,
\begin{equation*}
	\int_0^{\infty}\lambda^ku^{(2h)}(\epsilon\lambda)\d\lambda=(-1)^k\epsilon^{-k}k!\int_0^{\infty}u^{(2h-k)}(\epsilon\lambda)\d\lambda=(-1)^{k+1}\epsilon^{-k-1}k!\,u^{(2h-k-1)}(0),
\end{equation*}
where $u^{(2h-k-1)}(0)<\infty$ for all $k\in\{0,\ldots,2h-1\}$ by (\ref{central}) since $2h<\textnormal{Re}(2s+1)$ by assumption. Also,
\begin{equation*}
	u^{(2h)}(t)=\begin{cases}\mathcal{O}(1),& 2h\ \textnormal{even}\\
	\mathcal{O}(t^{\delta}),& 2h\ \textnormal{odd}
	\end{cases}\ \ \ \ \ \textnormal{as}\ \ t\downarrow 0,\ \textnormal{with some}\ \delta>0,
\end{equation*}
by Corollary \ref{afterintegral} and so dominated convergence yields
\begin{align*}
	\int_0^{\infty}\bigg(\frac{1}{\lambda-\im}-\frac{(-1)^{2h}}{\lambda+\im}\bigg)u^{(2h)}(\epsilon\lambda)\d\lambda\stackrel{\epsilon\downarrow 0}{\longrightarrow}\begin{cases}\im\pi u^{(2h)}(0),&2h\ \textnormal{even}\smallskip\\
	2\int_0^{\infty}u^{(2h)}(\lambda)\frac{\d\lambda}{\lambda},&2h\ \textnormal{odd}\end{cases}.
\end{align*}
What results from the above is, seeing that the limit under investigation necessarily exists by \eqref{i7},
\begin{equation*}
	\lim_{\epsilon\downarrow 0}\int_0^{\infty}K_{2h}^{\epsilon}(t)u(t)\d t=\begin{cases}\displaystyle \frac{1}{2}(-1)^h u^{(2h)}(0),&2h\ \textnormal{even}\bigskip\\ \displaystyle\frac{1}{\pi}(-1)^{h+\frac{1}{2}}\int_0^{\infty}u^{(2h)}(\lambda)\frac{\d\lambda}{\lambda},&2h\ \textnormal{odd}\end{cases},
\end{equation*}
as claimed before. This completes our proof.
\end{proof}
\begin{remark}
We remark that, when $2h$ is even, the above argument coincides with that in \cite[Section 4]{Basor_2019}.
\end{remark}
We next compute the right-hand side of (\ref{1008even1}) for $2h \notin \mathbb{Z}_{\geq 1}$.

\begin{prop}\label{lim2} Let $s>-\frac{1}{2}$ and $h\in\mathbb{C}:-\frac{1}{2}<\textnormal{Re}(h)<\frac{1}{2}+\textnormal{Re}(s)$ with $2h\notin\mathbb{Z}_{\geq 1}$,
\begin{equation*}
	\lim_{\epsilon\downarrow 0}\int_0^{\infty}K_{2h}^{\epsilon}(t)\exp\left[\int_0^{2t}v(x;s)\frac{\d x}{x}\right]\d t=-\frac{\sin(\pi h)}{\pi}\,\Gamma(2h-2M)\int_0^{\infty}u^{(2M+1)}(t)\frac{\d t}{t^{2h-2M}}
\end{equation*}
with $u^{(n)}(t)$ as in \eqref{defofu}. Here, $\textnormal{Re}(h)\in(M,M+1)$ with $M\in\mathbb{Z}_{\geq 0}$, and we set $M=-\frac{1}{2}$ when $\textnormal{Re}(h)\in(-\frac{1}{2},0)$.
\end{prop}
\begin{proof} Suppose, convergence questions aside for a moment, we integrate by parts $L$-times so that
\begin{equation*}
	\int_0^{\infty}\frac{u(t)\d t}{(\epsilon\pm\im t)^{2h+1}}=\sum_{k=0}^{L-1}\frac{u^{(k)}(t)}{\epsilon^{2h-k}}\bigg\{\prod_{\ell=0}^k\frac{(2h-\ell)^{-1}}{\pm\im}\bigg\}\Bigg|_{t=0}^{\infty}+\bigg\{\prod_{\ell=0}^{L-1}\frac{(2h-\ell)^{-1}}{\pm\im}\bigg\}\int_0^{\infty}\frac{u^{(L)}(t)\d t}{(\epsilon\pm\im t)^{2h+1-L}}.
\end{equation*}
The concrete choice of $L\in\mathbb{Z}_{\geq 1}$ is now motivated by the requirements
\begin{equation}\label{j1}
	u^{(k)}(t)=\mathcal{O}(1),\ \ t\downarrow 0\ \ \forall\,k\in\{0,1,\ldots,L-1\}\ \ \ \ \ \ \textnormal{and}\ \ \ \ \ \ L^1(0,\infty)\ni t\mapsto u^{(L)}(t) t^{-2h-1+L},
\end{equation}
seeing that $u^{(k)}(t)$ decays exponentially fast as $t$ tends to $+\infty$ for any $k\in\mathbb{Z}_{\geq 0}$. For \eqref{j1} to hold, we recall Corollary \ref{afterintegral} and the overall constraint $0<\textnormal{Re}(2h)<1+\textnormal{Re}(2s)$. In turn, if $\textnormal{Re}(2h)\in (2M,2M+2)$ with $M\in\mathbb{Z}_{\geq 0}$, we necessarily must have $2M<\textnormal{Re}(2s)+1$ and we take $L=2M+1$. In turn,
\begin{equation*}
	u^{(k)}(t)=\mathcal{O}(1),\ t\downarrow 0\ \ \ \forall\,k\in\{0,1,\ldots,2M\}\ \ \ \ \ \textnormal{as well as}\ \ \ \ \ L^1(0,\infty)\ni t\mapsto u^{(2M+1)}(t)t^{-2h+2M},
\end{equation*}
where we also use that $L=2M+1$ is an odd natural number. Consequently, the aforementioned $L$-fold integration by parts is legitimate, and the limit under consideration necessarily exists by Proposition \ref{propexpectation} and (\ref{1008even1}).
We deduce
\begin{equation}\label{directlyusethedef}
	\lim_{\epsilon\downarrow 0}\int_0^{\infty}K_{2h}^{\epsilon}(t)u(t)\d t=\frac{\Gamma(2h-2M)}{2\pi\,\im^{2M+1}}\lim_{\epsilon\downarrow 0}\int_0^{\infty}\bigg(\frac{1}{(\epsilon+\im t)^{2h-2M}}-\frac{1}{(\epsilon-\im t)^{2h-2M}}\bigg)u^{(2M+1)}(t)\d t,
\end{equation}
and hence, by the dominated convergence theorem,
\begin{equation*}
	\lim_{\epsilon\downarrow 0}\int_0^{\infty}K_{2h}^{\epsilon}(t)u(t)\d t=-\frac{\sin(\pi h)}{\pi}\Gamma(2h-2M)\int_0^{\infty}u^{(2M+1)}(t)\frac{\d t}{t^{2h-2M}},
\end{equation*}
as claimed in the statement. In particular, when $-\frac{1}{2}<\textnormal{Re}(h)<0$, we directly use (\ref{directlyusethedef}) with $M=-\frac{1}{2}$ by the definition of $K_{2h}^{\epsilon}(t)u(t)\d t$. Then by the dominated convergence theorem, we have the conclusion in this proposition. This completes our proof. 
%
\end{proof}

\begin{remark} Observe how Proposition \ref{lim2} links back to Proposition \ref{lim1}: if on one hand $2h\rightarrow 2M+1\in(2M,2M+2)$ approaches an odd natural number then Proposition \ref{lim2} yields at once
\begin{equation*}
	-\frac{\sin(\pi h)}{\pi}\Gamma(2h-2M)\int_0^{\infty}u^{(2M+1)}(t)\frac{\d t}{t^{2h-2M}}\longrightarrow\frac{1}{\pi}(-1)^{M+1}\int_0^{\infty}u^{(2M+1)}(t)\frac{\d t}{t},
\end{equation*}
and if on the other hand $2h\rightarrow 2M$ approaches an even natural number then Proposition \ref{lim2} says
\begin{equation*}
	-\frac{\sin(\pi h)}{\pi}\Gamma(2h-2M)\int_0^{\infty}u^{(2M+1)}(t)\frac{\d t}{t^{2h-2M}}\longrightarrow-\frac{(-1)^M}{2}\int_0^{\infty}u^{(2M+1)}(t)\d t=\frac{1}{2}(-1)^M u^{(2M)}(0).
\end{equation*}
The so-obtained results match onto the ones stated in Proposition \ref{lim1}.
\end{remark}

Summarizing the above, we are now ready to prove Theorem \ref{maintheorem1}.

\begin{proof}[Proof of Theorem \ref{maintheorem1}]
Using Corollary \ref{cheat}, we obtain the explicit form of $v(z;s)$ given in (\ref{e73repeat}). By (\ref{1008even1}) and Proposition \ref{exofh},
\begin{equation}\label{1008even2}
\lim_{N\rightarrow\infty}\frac{1}{N^{s^2+2h}}F_N(s,h)=\frac{G(s+1)^2}{G(2s+1)}2^{1-2h}\lim_{\epsilon\downarrow 0}\int_{0}^{+\infty}\exp\left[\int_0^{2t}v(x;s)\frac{\d x}{x}\right]K_{2h}^{\epsilon}(t)\d t
\end{equation}
holds for $s>-\frac{1}{2}$ and $h\in\mathbb{C}$ with $0\leq \textnormal{Re}(h)<\frac{1}{2}+s$. Now the conclusion of Theorem \ref{maintheorem1} follows from Propositions \ref{lim1} and \ref{lim2}.
\end{proof}

We conclude this section by providing a detailed proof of Proposition~\ref{rule}.
\begin{proof}[Proof of Proposition \ref{rule}]
Let $s=m$. When $h\in (m,m+\frac{1}{2})$, by Proposition \ref{lim2},
\bea\label{1028formula1}
\lim_{N\rightarrow\infty}\frac{F_N(m,h)}{N^{m^2+2h}}=-\frac{G^2(1+m)}{G(1+2m)}\displaystyle\frac{2^{1-2h}}{\pi}\sin(\pi h)\Gamma(2h-2m)\int_0^{\infty}u^{(2m+1)}(\lambda)\frac{\d\lambda}{\lambda^{2h-2m}}
\eea
Recall that 
$
u(\lambda)=\exp\left[\int_0^{2\lambda}v(x;m)\frac{\d x}{x}\right].
$
Denote 
$
f(\lambda)=\int_0^{2\lambda}v(x;m)\frac{\d x}{x}.
$
By Fa\`{a} di Bruno's formula, 
\beas
u^{(n)}(\lambda)
=u(\lambda)\sum_{l_1+2l_2+\cdots+nl_n =n} \frac{n!}{l_1!\cdots l_n!} \prod_{j=1}^n \left(\frac{f^{(j)}(\lambda)}{j!}\right)^{l_j}.
\eeas
Substitute this into (\ref{1028formula1}), 
\begin{align*}
\lim_{N\rightarrow\infty}\frac{F_N(m,h)}{N^{m^2+2h}}=-\frac{G^2(1+m)}{G(1+2m)}\displaystyle\frac{2^{1-2h}}{\pi}\sin(\pi h)\Gamma(2h-2m)\left(\mathcal{G}_{1}+\mathcal{G}_{2}\right),
\end{align*}
where
\begin{equation}\label{zzz66}
\mathcal{G}_{1}:=\sum_{l_1+2l_2+\cdots+2ml_{2m} =2m+1} \frac{(2m+1)!}{l_1!\cdots l_{2m}!} \int_0^{\infty}u(\lambda)\prod_{j=1}^{2m} \left(\frac{f^{(j)}(\lambda)}{j!}\right)^{l_j}\frac{\d\lambda}{\lambda^{2h-2m}},
\end{equation}
and 
\beas
\mathcal{G}_{2}:=\int_0^{\infty}u(\lambda)f^{(2m+1)}(\lambda)\frac{\d\lambda}{\lambda^{2h-2m}}.
\eeas
For $\mathcal{G}_{1}$, observe that there exists an odd $j$ such that $l_{j}\neq 0$, since $l_{1}+2l_{2}+\cdots+2ml_{2m}$ is odd. By the gap structure in Corollary \ref{afterintegral}, we have 
\begin{equation}\label{zz66}
u(\lambda)\prod_{j=1}^{2m} \left(\frac{f^{(j)}(\lambda)}{j!}\right)^{l_j} = \mathcal{O}_{m}(\lambda) 
\quad \text{as} \quad \lambda \downarrow 0.
\end{equation}
Combining \eqref{zz66} with the exponential rate of decay of the function in \eqref{zz66} as $\lambda \to \infty$, deduced from Corollary \ref{cor5555}, and applying the dominated convergence theorem, we can move the limit $h \to m+\frac{1}{2}$ in \eqref{remark2.8} inside the integral in \eqref{zzz66}. Hence, 
$
\lim_{h \to m+\frac{1}{2}}\mathcal{G}_{1} =\mathcal{O}_{m}(1).
$
For $\mathcal{G}_{2}$, by Proposition \ref{loopy} with $K=2m+1$, 
\beas
u(\lambda)f^{(2m+1)}(\lambda)=\mathsf{d}_{2m+1}(m)2^{2m+1}(2m)!+\mathcal{O}_{m}(\lambda), \quad \text{as} \quad \lambda \downarrow 0.
\eeas
Again by the  exponential rate of decay of $\lambda\mapsto u(\lambda)f^{(2m+1)}(\lambda)$ as $\lambda \to \infty$, then 
\beas
\lim_{h \to m+\frac{1}{2}}\mathcal{G}_{2}=\mathsf{d}_{2m+1}(m)2^{2m+1}(2m)!\lim_{h \to m+\frac{1}{2}}\int_0^{1}\frac{\d\lambda}{\lambda^{2h-2m}}+\mathcal{O}_{m}(1).
\eeas
Note that $\lim_{h \rightarrow m+\frac{1}{2}}\int_0^{1}\frac{\d\lambda}{\lambda^{2h-2m}}=+\infty$. So we conclude the claim in Proposition \ref{rule}.
\end{proof}

\begin{appendix}
\section{The Painlev\'e-III model problem}
The following RHP is closely related to the one in \cite[Section $1.2$]{XZ}. Throughout $\beta>0$ is a fixed parameter. 
\begin{problem}\label{PIIImodel} Let $x\in(0,\infty)$ and $s>-\frac{1}{2}$. Find $Q(\zeta)=Q(\zeta;x,s)\in\mathbb{C}^{2\times 2}$ with the following properties:
\begin{enumerate}
	\item[(1)] $\zeta\mapsto Q(\zeta)$ is analytic for $\zeta\in\mathbb{C}\setminus\Sigma_Q$ where $\Sigma_Q=\bigcup_{j=1}^4\Sigma_j$ consists of the four oriented rays
	\begin{equation*}
		\Sigma_1:=(-\beta,0)\subset\mathbb{R},\ \ \ \ \ \Sigma_2:=\e^{\im\frac{5\pi}{3}}(0,\infty),\ \ \ \ \ \Sigma_3:=(0,\infty)\subset\mathbb{R},\ \ \ \ \ \Sigma_4:=\e^{\im\frac{\pi}{3}}(0,\infty),
	\end{equation*}
	and is shown in Figure \ref{figa}. On $\Sigma_Q\setminus\{-\beta,0\}$, $Q(\zeta)$ admits continuous limiting values $Q_{\pm}(\zeta)$ as we approach $\Sigma_Q\setminus\{-\beta,0\}$ from either side of $\mathbb{C}\setminus\Sigma_Q$.
	\item[(2)] The limiting values $Q_{\pm}(\zeta)$ on $\Sigma_Q\setminus\{-\beta,0\}$ satisfy $Q_+(\zeta)=Q_-(\zeta)G_Q(\zeta)$ with $G_Q(\zeta)=G_Q(\zeta;s)$ given by
	\begin{align*}
		G_Q(\zeta)=\e^{-\im\pi s\sigma_3},\ \ \zeta\in\Sigma_1;&\hspace{1cm} G_Q(\zeta)=\begin{bmatrix}0&1\\ -1&0\end{bmatrix},\ \ \zeta\in\Sigma_3;\\
		G_Q(\zeta)=\begin{bmatrix}1 & 0 \\ \e^{2\pi\im s}& 1\end{bmatrix},\ \ \zeta\in\Sigma_2;&\hspace{1cm} G_Q(\zeta)=\begin{bmatrix}1&0\\ \e^{-2\pi\im s} & 1\end{bmatrix},\ \ \zeta\in\Sigma_4.
	\end{align*}
	\begin{figure}[tbh]
	\begin{tikzpicture}[xscale=0.9,yscale=0.9]
	\draw [thick,red,decoration={markings,mark= at position 0.6 with {\arrow{>}}},postaction={decorate}] (0,0) -- (1,1.732050);
	\draw [thick,red,decoration={markings,mark= at position 0.6 with {\arrow{>}}},postaction={decorate}] (0,0) -- (1,-1.732050);
	 \draw [->] (-2,0) -- (3.5,0) node [right] {$\footnotesize{\textnormal{Re}(\zeta)}$};
        \draw [thick,red,decoration={markings,mark= at position 0.5 with {\arrow{>}}},postaction={decorate}] (0,0) -- (3,0);
  \draw [->] (0,-2) -- (0,2) node [above] {$\footnotesize{\textnormal{Im}(\zeta)}$};
  	\draw [thick,red,decoration={markings,mark= at position 0.6 with {\arrow{>}}},postaction={decorate}] (-1,0) -- (0,0);
  \draw [fill, color=blue!60!black] (-1,0) circle [radius=0.04];
  \draw [fill, color=blue!60!black] (0,0) circle [radius=0.04];
		\node [right] at (1,1.72) {{\footnotesize $\Sigma_4$}};
		\node [right] at (1,-1.72) {{\footnotesize $\Sigma_2$}};
		\node [above] at (2,0.1) {{\footnotesize $\Sigma_3$}};
		\node [above] at (-0.65,0.1) {{\footnotesize $\Sigma_1$}};
		\node [below] at (-0.5,-0.75) {{\footnotesize\textcolor{cyan}{ $\Omega_1$}}};
		\node [above] at (-0.5,0.75) {{\footnotesize\textcolor{cyan}{ $\Omega_4$}}};
		\node [above] at (2.5,0.75) {{\footnotesize\textcolor{cyan}{ $\Omega_3$}}};
		\node [below] at (2.5,-0.75) {{\footnotesize\textcolor{cyan}{ $\Omega_2$}}};
\end{tikzpicture}
\caption{The oriented jump contour $\Sigma_Q$ in the complex $\zeta$-plane. The (end)points $\zeta=0,-\beta$ are colored in \textcolor{blue!60!black}{blue}. The sectors in between in \textcolor{cyan}{cyan}.}
\label{figa}
\end{figure}
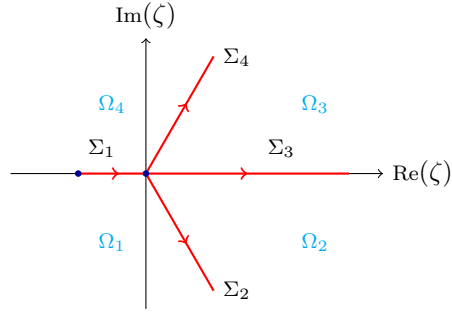
	
	\item[(3)] Near $\zeta=-\beta$, $\zeta\mapsto Q(\zeta)$ is weakly singular in that, with some $\zeta\mapsto\widehat{Q}(\zeta)=\widehat{Q}(\zeta;x,s)$ analytic at $\zeta=-\beta$ and non-vanishing, $Q(\zeta)$ is of the form
	\begin{equation}\label{A1}
		Q(\zeta)=\widehat{Q}(\zeta)(\zeta+\beta)^{\frac{s}{2}\sigma_3},\ \ \zeta\in\mathbb{D}_{\epsilon}(-\beta),\ \ \ \ \ \ \textnormal{arg}(\zeta+\beta)\in(0,2\pi).
	\end{equation}
	\item[(4)] Near $\zeta=0$, $\zeta\mapsto Q(\zeta)$ is weakly singular in that, with some $\zeta\mapsto\widecheck{Q}(\zeta)=\widecheck{Q}(\zeta;x,s)$ analytic at $\zeta=0$ and non-vanishing, $Q(\zeta)$ is of the form
	\begin{equation}\label{A2}
		Q(\zeta)=\widecheck{Q}(\zeta)\zeta^{\frac{s}{2}\sigma_3}\mathcal{M}_j(\zeta),\ \ \zeta\in\Omega_j\cap\mathbb{D}_{\epsilon}(0),\ \ \ \ \ \ \textnormal{arg}\,\zeta\in(0,2\pi),
	\end{equation}
	with the piecewise defined, for $s\notin\mathbb{Z}_{\geq 0}$,
	\begin{align*}
		\mathcal{M}_4(\zeta)=&\,\,\begin{bmatrix}1&\frac{\im}{2}\frac{\e^{\im\pi s}}{\sin(\pi s)}\\ 0 & 1\end{bmatrix},\ \ \ \ \ \mathcal{M}_3(\zeta)=\mathcal{M}_4(\zeta)\begin{bmatrix}1&0\\ -\e^{-2\pi\im s}& 1\end{bmatrix},\\
		 \mathcal{M}_1(\zeta)=&\,\,\mathcal{M}_4(\zeta)\e^{\im\pi s\sigma_3},\ \ \ \ \ \mathcal{M}_2(\zeta)=\mathcal{M}_1(\zeta)\begin{bmatrix}1&0\\ \e^{2\pi\im s}& 1\end{bmatrix}.
	\end{align*}
	If $s\in\mathbb{Z}_{\geq 0}$, then \eqref{A2} holds
	with the piecewise defined
	\begin{align*}
		\mathcal{M}_4(\zeta)=&\,\,\begin{bmatrix}1 & \frac{\im}{2\pi}\ln\zeta\\ 0 & 1\end{bmatrix},\ \ \ \ \ \mathcal{M}_3(\zeta)=\mathcal{M}_4(\zeta)\begin{bmatrix}1&0\\ -\e^{-2\pi\im s}& 1\end{bmatrix},\\
		 \mathcal{M}_1(\zeta)=&\,\,\mathcal{M}_4(\zeta)\e^{\im\pi s\sigma_3},\ \ \ \ \ \mathcal{M}_2(\zeta)=\mathcal{M}_1(\zeta)\begin{bmatrix}1&0\\ \e^{2\pi\im s}& 1\end{bmatrix}.
	\end{align*}
	\item[(5)] As $\zeta\rightarrow\infty$ and $\zeta\notin\Sigma_Q$, $Q(\zeta)$ is normalized as follows,
	\begin{equation}\label{A4}
		Q(\zeta)=\bigg\{I+\sum_{k=1}^2Q_k(x,s)\zeta^{-k}+\mathcal{O}\big(\zeta^{-3}\big)\bigg\}\zeta^{\frac{1}{4}\sigma_3}\frac{1}{\sqrt{2}}\begin{bmatrix}1&1\\ -1&1\end{bmatrix}\e^{-\im\frac{\pi}{4}\sigma_3}\e^{-\im\varpi(\zeta;x)\sigma_3},
	\end{equation}
	with $Q_k(x,s)$ independent of $\zeta$ and the function $\varpi(\zeta;x):=x\zeta^{\frac{1}{2}}$, as well as all other fractional powers, defined with cut on $[0,\infty)\subset\mathbb{R}$ such that $\textnormal{arg}\,\zeta\in(0,2\pi)$.
\end{enumerate}	
\end{problem}

\begin{prop}\label{Fredana} Let $x\in(0,\infty)$ and $s>-\frac{1}{2}$. Then RHP \ref{PIIImodel} admits a unique solution $Q(\zeta)=Q(\zeta;x,s)\in\mathbb{C}^{2\times 2}$.
\end{prop}
\begin{proof} Unicity follows directly from RHP \ref{PIIImodel} noting that any solution of the same RHP is unimodular. For existence we invoke Zhou's vanishing lemma \cite{Z} and first consider the homogeneous variant of RHP \ref{PIIImodel} in which conditions $(1)-(4)$ are as is, but condition $(5)$ is replaced by
\begin{equation*}
	Q(\zeta)=\mathcal{O}\big(\zeta^{-1}\big)\zeta^{\frac{1}{4}\sigma_3}\frac{1}{\sqrt{2}}\begin{bmatrix}1&1\\ -1&1\end{bmatrix}\e^{-\im\frac{\pi}{4}\sigma_3}\e^{-\im\varpi(\zeta;x)\sigma_3},\ \ \ \zeta\rightarrow\infty,\ \ \textnormal{arg}\,\zeta\in(0,2\pi),\ \ \zeta\notin\Sigma_Q.
\end{equation*}
Then, collapsing $\Sigma_4,\Sigma_2$ with $\Sigma_3$ and normalizing at infinity,
\begin{equation*}
	Z(\zeta):=Q(\zeta)\begin{cases}\begin{bmatrix}1&0\\ \e^{-2\pi\im s} & 1\end{bmatrix}\e^{\im\varpi(\zeta;x)\sigma_3},&\zeta\in\Omega_3\smallskip\\
	\begin{bmatrix}1&0\\ -\e^{2\pi\im s} & 1\end{bmatrix}\e^{\im\varpi(\zeta;x)\sigma_3},&\zeta\in\Omega_2\smallskip\\ \e^{\im\varpi(\zeta;x)\sigma_3},&\zeta\in\Omega_1\cup\Omega_4
	\end{cases},
\end{equation*}
solves a RHP with jump on $\Sigma_1\cup\Sigma_3\subset\mathbb{R}$ only,
\begin{equation*}
	Z_+(\zeta)=Z_-(\zeta)\e^{-\im\pi s\sigma_3},\ \ \zeta\in\Sigma_1;\ \ \ \ Z_+(\zeta)=Z_-(\zeta)\begin{bmatrix}\e^{-2\pi\im s+2\im x\sqrt{\zeta}} & 1\\ 0 & \e^{2\pi\im s-2\im x\sqrt{\zeta}}\end{bmatrix},\ \ \zeta\in\Sigma_3,
\end{equation*}
with normalization $Z(\zeta)=\mathcal{O}(\zeta^{-\frac{3}{4}})$ as $\zeta\rightarrow\infty$ with $\textnormal{arg}\,\zeta\in(0,2\pi)$ and singular behaviours at $\zeta=-\beta,0$ of the type
\begin{equation*}
	Z(\zeta)=\mathcal{O}(1)(\zeta+\beta)^{\frac{s}{2}\sigma_3},\ \textnormal{arg}(\zeta+\beta)\in(0,2\pi);\hspace{0.25cm} Z(\zeta)=\mathcal{O}(1)\zeta^{\frac{s}{2}\sigma_3}\begin{bmatrix}\mathcal{O}(1)& \mathcal{O}(1+\tau\ln\zeta)\\ 0 &\mathcal{O}(1)\end{bmatrix},\ \textnormal{arg}(\zeta)\in(0,2\pi);
\end{equation*}
with $\tau=0$ if $s\notin\mathbb{Z}_{\geq 0}$. From hereon the argument proceeds as in \cite[page $1651-1653$]{XZ}, utilizing $2s\in(-1,\infty)\subset\mathbb{R}$ and $x>0$: first one introduces
\begin{equation*}
	W(\zeta):=Z(\zeta)\begin{cases}\begin{bmatrix}0 & -1\\ 1 & 0\end{bmatrix},&\textnormal{Im}(\zeta)>0\smallskip\\ I,&\textnormal{Im}(\zeta)<0\end{cases},
\end{equation*}
which has jumps
\begin{equation*}
	W_+(\zeta)=W_-(\zeta)\begin{bmatrix}0&-1\\ 1&0\end{bmatrix},\ \ \ \zeta\in(-\infty,-\beta);\hspace{1.5cm} W_+(\zeta)=W_-(\zeta)\begin{bmatrix}0 & -\e^{-\im\pi s}\\ \e^{\im\pi s} & 0\end{bmatrix},\ \ \ \zeta\in(-\beta,0);
\end{equation*}
\begin{equation*}
	W_+(\zeta)=W_-(\zeta)\begin{bmatrix}1 & -\e^{-2\pi\im s+2\im x\sqrt{\zeta}}\\ \e^{2\pi\im s-2\im x\sqrt{\zeta}} & 0\end{bmatrix},\ \ \ \zeta\in(0,\infty);
\end{equation*}
normalization at infinity $W(\zeta)=\mathcal{O}(\zeta^{-\frac{3}{4}})$ off $\mathbb{R}$ and properly adjusted singular behaviours near $\zeta=-\beta,0$. In turn the auxiliary function
\begin{equation*}
	H(\zeta):=W(\zeta)\big(W(\overline{\zeta})\big)^{\dagger},\ \ \ \ \zeta\in\mathbb{C}\setminus\mathbb{R},
\end{equation*}
is analytic off $\mathbb{R}$, it has $H(\zeta)=\mathcal{O}(1)$ as $\zeta\rightarrow-\beta$, $H(\zeta)=\mathcal{O}(\zeta^s\ln\zeta)$ as $\zeta\rightarrow 0$ and $H(\zeta)=\mathcal{O}(\zeta^{-\frac{3}{2}})$ as $\zeta\rightarrow\infty$, all for $\zeta\notin\mathbb{R}$. Consequently, by Cauchy's theorem $\int_{-\infty}^{\infty}H_+(\zeta)\d\zeta=0_{2\times 2}$ and adding to that its Hermitian conjugate, we derive
\begin{equation*}
	\begin{bmatrix}0 & 0\\ 0 & 0\end{bmatrix}=\int_{-\infty}^{\infty}\Big[H_+(\zeta)+\big(H_+(\zeta)\big)^{\dagger}\Big]\d\zeta=2\int_0^{\infty}W_-(\zeta)\begin{bmatrix}1 & 0\\ 0 & 0\end{bmatrix}\big(W_-(\zeta)\big)^{\dagger}\d\zeta,
\end{equation*}
and so $W_-^{11}(\zeta)=W_-^{21}(\zeta)\equiv 0$ for all $\zeta\in(0,\infty)$. Consequently also $W_+^{12}(\zeta)=W_+^{22}(\zeta)\equiv 0$ for $\zeta\in(0,\infty)$ in light of the jump of $\zeta\mapsto W(\zeta)$ on $(0,\infty)$. Moreover, since the jump matrix on $(0,\infty)$ extends analytically to all of $\mathbb{C}\setminus(-\infty,0]$, we can extend $W(\zeta)$ from $\Im(\zeta)>0$ analytically to $\mathbb{C}\setminus(-\infty,0]$ such that $W^{12}(\zeta)=W^{22}(\zeta)\equiv 0$ for $\zeta\in(0,\infty)$. In turn, $W^{12}(\zeta)=W^{22}(\zeta)\equiv 0$ for $\textnormal{Im}(\zeta)>0$, and, by similar reasoning, $W^{11}(\zeta)=W^{21}(\zeta)\equiv 0$ for $\textnormal{Im}(\zeta)<0$. Next, for $k\in\{1,2\}$ one assembles the scalar-valued functions
\begin{equation*}
	h_k(\zeta):=\begin{cases}W^{k1}(\zeta),&\textnormal{Im}(\zeta)>0\smallskip\\
	W^{k2}(\zeta),&\textnormal{Im}(\zeta)<0\end{cases},
\end{equation*}
and realizes that $\zeta\mapsto h_k(\zeta)$ are analytic on $\mathbb{C}\setminus[-\beta,\infty)$ with jumps given by
\begin{equation*}
	h_{k+}(\zeta)=h_{k-}(\zeta)\e^{\im\pi s},\ \ \zeta\in(-\beta,0);\hspace{1.25cm}h_{k+}(\zeta)=h_{k-}(\zeta)\e^{2\pi\im s-2\im x\sqrt{\zeta}},\ \ \zeta\in(0,\infty),
\end{equation*}
with weak singularities near $\zeta=-\beta,0$ and normalisation $h_k(\zeta)=\mathcal{O}(\zeta^{-\frac{3}{4}})$ at $\zeta=\infty$ off $\mathbb{R}$. After a rotation in the $\zeta$-variable, and adapting \cite[$(2.68)$]{XZ} to our needs, one is finally in a situation to apply Carlson's theorem from which it follows $h_k(\zeta)\equiv 0$ for $\zeta\notin\mathbb{R}$. This translates into $W(\zeta)\equiv 0_{2\times 2}$ for $\zeta\notin\mathbb{R}$, and thus necessarily $W(\zeta)\equiv 0_{2\times 2}$ on all of $\mathbb{C}$. Tracing back, $Q(\zeta)\equiv 0_{2\times 2}$ on $\mathbb{C}$, so the homogeneous variant of RHP \ref{PIIImodel} has only the trivial solution. That is sufficient to deduce solvability of RHP \ref{PIIImodel} by Zhou's vanishing lemma, for any $z>0$ and $s>-\frac{1}{2}$. Our proof is complete.
\end{proof}

\begin{rem}\label{mero} A more refined analysis, invoking the analytic Fredholm alternative of \cite{Z}, reveals that RHP \ref{PIIImodel}, given $2s>-1$, is solvable for all $\textnormal{Re}(x)>0$ away from a discrete set that is disjoint from $(0,\infty)\subset\mathbb{R}$. In particular $Q(\zeta;x,s)$ is meromorphic in $x$ in the half-plane $\textnormal{Re}(x)>0$.

\end{rem}

Unicity of the solution of RHP \ref{PIIImodel} imposes symmetry constraints on the same problem and below we list one of those.
\begin{cor} Suppose $Q(\zeta)=Q(\zeta;x,s)\in\mathbb{C}^{2\times 2}$ solves RHP \ref{PIIImodel} for $x\in(0,\infty)$ and $s>-\frac{1}{2}$. Then
\begin{equation*}
	Q(\zeta;x,s)=\sigma_3\overline{Q(\bar{\zeta};x,s)}\sigma_3,\ \ \ \zeta\in\mathbb{C}\setminus\Sigma_Q,
\end{equation*}
which yields $Q_1^{11}(x,s),Q_1^{22}(x,s)\in\mathbb{R}$ and $Q_1^{21}(x,s),Q_1^{12}(x,s)\in\im\mathbb{R}$, in particular. Moreover it yields
\begin{equation*}
	\widehat{Q}(-\beta)=\sigma_3\overline{\widehat{Q}(-\beta)}\e^{-\im\pi s\sigma_3}\sigma_3.
\end{equation*}
\end{cor}
Seeing that $G_Q(\zeta;s)$ is piecewise constant in $\zeta$ and $x$-independent the below Lax system \eqref{A5} is a natural consequence of RHP \ref{PIIImodel}.
\begin{prop} Suppose $Q(\zeta)=Q(\zeta;x,s)\in\mathbb{C}^{2\times 2}$ solves RHP \ref{PIIImodel} for $x\in(0,\infty)$ and $s>-\frac{1}{2}$. Then
\begin{equation}\label{A5}
	\frac{\partial Q}{\partial\zeta}(\zeta)=\bigg\{xF_{\infty}+\frac{F_0}{\zeta}+\frac{F_1}{\zeta+\beta}\bigg\}Q(\zeta),\ \ \ \ \ \frac{\partial Q}{\partial x}(\zeta)=\Big\{\zeta G_{\infty}+G_0\Big\}Q(\zeta),	
\end{equation}
with $\zeta$-independent coefficients
\begin{equation*}
	F_{\infty}=\frac{1}{2}G_{\infty}=\frac{\im}{2}\sigma_+,\ \ \ \ \ F_0=\frac{1}{2}\begin{cases}s\,\widecheck{Q}(0)\sigma_3\widecheck{Q}(0)^{-1},&s\neq 0\smallskip\\
	\frac{\im}{\pi}\widecheck{Q}(0)\sigma_+\widecheck{Q}(0)^{-1},&s=0\end{cases},\ \ \ \ \ F_1=\frac{s}{2}\widehat{Q}(-\beta)\sigma_3\widehat{Q}(-\beta)^{-1},
\end{equation*}
\begin{equation*}
	G_0=\im\sigma_-+\im\big[Q_1,\sigma_+\big];\hspace{1.5cm}\sigma_-=\begin{bmatrix}0&0\\ 1 & 0\end{bmatrix},\ \ \sigma_+=\begin{bmatrix}0&1\\ 0 & 0\end{bmatrix},
\end{equation*}
that are expressed in terms of the data occurring in conditions $(3),(4)$ and $(5)$ of RHP \ref{PIIImodel}. The same coefficients satisfy the constraints
\begin{equation}\label{A6}
	F_0+F_1=\frac{1}{4}\sigma_3+\frac{x}{2}G_0,\ \ \ \ \ \ \  \det F_0=\det F_1=-\frac{1}{4}s^2.
\end{equation}
Moreover, parametrising the entries of $G_0$ and $F_1$ in the following $(x,s)$-dependent fashion,
\begin{equation*}
	a:=-\im Q_1^{21}=G_0^{11}=-G_0^{22}\in\mathbb{R},\hspace{1cm} b:=\im\big(Q_1^{11}-Q_1^{22}\big)=G_0^{12}\in\im\mathbb{R},
\end{equation*}
\begin{equation*}
	p:=F_1^{11}=-F_1^{22}\in\mathbb{R},\hspace{1cm}q:=F_1^{12}\in\im\mathbb{R},\hspace{1cm}r:=F_1^{21}\in\im\mathbb{R},
\end{equation*}
we also record the following constraints
\begin{equation}\label{A7}
	a^2+\frac{\d a}{\d x}=-\im b,\ \ \ \ \ \beta(-\im r)=\frac{1}{2}\frac{\d}{\d x}(xa),\ \ \ \ \ \ 4\big(p^2+qr\big)=s^2.
\end{equation}
\end{prop}

\begin{cor}\label{PcorIII} Compatibility of system \eqref{A5} is equivalent to the coupled ODE system
\begin{equation*}
	\frac{\d p}{\d x}=-\im(\beta r+q)+br,\ \ \ \ \ \frac{\d q}{\d x}=2\im\beta p+2(aq-bp),\ \ \ \ \ \frac{\d r}{\d x}=2(\im p-ar),
\end{equation*}
\begin{equation*}
	\frac{\d}{\d x}(xa)=-2\im\beta r,\ \ \ \ \ \ \frac{\d}{\d x}(xb)=4\im\beta p-b.
\end{equation*}

%
%
%
%
%
%
%
%
\end{cor}

\section{The Airy model problem}
Suppose $\textnormal{Ai}(\zeta)$ denotes the Airy function and $\textnormal{Ai}'(\zeta)$ its derivative. We construct
\begin{equation}\label{A10}
	A(\zeta):=\sqrt{2\pi}\,\e^{-\im\frac{\pi}{4}}\begin{bmatrix}\textnormal{Ai}(\zeta)&\e^{\im\frac{\pi}{3}}\textnormal{Ai}(\e^{-\im\frac{2\pi}{3}}\zeta)\smallskip\\
	\textnormal{Ai}'(\zeta) & \e^{-\im\frac{\pi}{3}}\textnormal{Ai}'(\e^{-\im\frac{2\pi}{3}}\zeta)\end{bmatrix}
	\begin{cases}I,&\textnormal{arg}\,\zeta\in(0,\frac{2\pi}{3})\smallskip\\ \bigl[\begin{smallmatrix}1&0\\ -1 & 1\end{smallmatrix}\bigr],&\textnormal{arg}\,\zeta\in(\frac{2\pi}{3},\pi)\smallskip\\ \bigl[\begin{smallmatrix}1&-1\\ 0 & 1\end{smallmatrix}\bigr],&\textnormal{arg}\,\zeta\in(-\frac{2\pi}{3},0)\smallskip\\ \bigl[\begin{smallmatrix}1&-1\\ 0 & 1\end{smallmatrix}\bigr]\bigl[\begin{smallmatrix}1&0\\ 1 & 1\end{smallmatrix}\bigr],&\textnormal{arg}\,\zeta\in(-\pi,-\frac{2\pi}{3})\end{cases},
\end{equation}
and realise that $A(\zeta)$ is uniquely characterised by the following four properties:
\begin{problem}\label{Airymodel} The model function $A(\zeta)\in\mathbb{C}^{2\times 2}$ defined in \eqref{A10} is such that:
\begin{enumerate}
	\item[(1)] $\zeta\mapsto A(\zeta)$ is analytic for $\zeta\in\mathbb{C}\setminus\Sigma_A$ with $\Sigma_A$ shown in Figure \ref{figb}. On $\Sigma_A$, $A(\zeta)$ admits continuous limiting values $A_{\pm}(\zeta)$ as we approach $\Sigma_A$ from either side of $\mathbb{C}\setminus\Sigma_A$.
	\item[(2)] The limiting values $A_{\pm}(\zeta)$ on $\Sigma_A\ni\zeta$ satisfy
	\begin{equation*}
		A_+(\zeta)=A_-(\zeta)\begin{bmatrix}0&1\\ -1 & 0\end{bmatrix},\ \ \zeta<0;\ \ \ \ \ \ \ \ \ A_+(\zeta)=A_-(\zeta)\begin{bmatrix}1 & 1\\ 0 & 1\end{bmatrix},\ \ \zeta>0;
	\end{equation*}
	\begin{equation*}
		A_+(\zeta)=A_-(\zeta)\begin{bmatrix}1&0\\ 1 & 1\end{bmatrix},\ \ \textnormal{arg}\,\zeta=\pm\frac{2\pi}{3}.
	\end{equation*}
	\item[(3)] $\zeta\mapsto A(\zeta)$ is bounded near $\zeta=0$.
	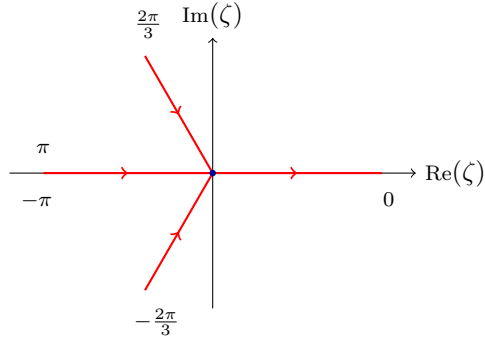
\begin{figure}[tbh]
	\begin{tikzpicture}[xscale=0.9,yscale=0.9]
	 \draw [->] (-3,0) -- (3,0) node [right] {$\footnotesize{\textnormal{Re}(\zeta)}$};
	   \draw [->] (0,-2) -- (0,2) node [above] {$\footnotesize{\textnormal{Im}(\zeta)}$};
	\draw [thick, color=red, decoration={markings, mark=at position 0.25 with {\arrow{>}}}, decoration={markings, mark=at position 0.75 with {\arrow{>}}}, postaction={decorate}] (-2.5,0) -- (2.5,0);
	\draw [thick, color=red, decoration={markings, mark=at position 0.5 with {\arrow{>}}}, postaction={decorate}] (-1,1.73205) -- (0,0);
	\draw [thick, color=red, decoration={markings, mark=at position 0.5 with {\arrow{>}}}, postaction={decorate}] (-1,-1.73205) -- (0,0);
	\node [below] at (2.6,-0.15) {{\footnotesize $0$}};
	\node [below] at (-2.6,-0.15) {{\footnotesize $-\pi$}};
	\node [above] at (-2.5,0.15) {{\footnotesize $\pi$}};
	\node [right] at (-1.3,2.2) {{\footnotesize $\frac{2\pi}{3}$}};
	\node [right] at (-1.3,-2.2) {{\footnotesize $-\frac{2\pi}{3}$}};
	\draw [fill, color=blue!60!black] (0,0) circle [radius=0.04];
\end{tikzpicture}
\caption{The oriented jump contour $\Sigma_A$, shown in red, for the model function $A(\zeta)$ in the complex $\zeta$-plane, with indicated values of $\textnormal{arg}\,\zeta$ on the four rays.}
\label{figb}
\end{figure}

	\item[(4)] As $\zeta\rightarrow\infty$ and $\zeta\notin\Sigma_A$, $A(\zeta)$ is normalised as follows,
	\begin{equation*}
		A(\zeta)=\Bigg\{I-\frac{7}{48\zeta}\begin{bmatrix}0 & 0\\ 1 &0\end{bmatrix}+\frac{5}{48\zeta^2}\begin{bmatrix}0&1\\0&0\end{bmatrix}+\mathcal{O}\big(\zeta^{-3}\big)\Bigg\}\zeta^{-\frac{1}{4}\sigma_3}\frac{1}{\sqrt{2}}\begin{bmatrix}1&1\\ -1 & 1\end{bmatrix}\e^{-\im\frac{\pi}{4}\sigma_3}\e^{-\frac{2}{3}\zeta^{\frac{3}{2}}\sigma_3},
	\end{equation*}
	where all fractional powers are principal ones with cut on $(-\infty,0]\subset\mathbb{R}$ such that $\textnormal{arg}\,\zeta\in(-\pi,\pi)$.
\end{enumerate}
\end{problem}

\section{The Bessel model problem}

Suppose $I_{\nu}(\zeta),K_{\nu}(\zeta)$ denote the principal branch modified Bessel functions of index $\nu>-1$ and $I_{\nu}'(\zeta),K_{\nu}'(\zeta)$ their derivatives. We construct, where $\zeta^{\alpha}:\mathbb{C}\setminus(-\infty,0]\rightarrow\mathbb{C}$ is defined with its principal branch,
\begin{equation}\label{A11}
	J(\zeta;\nu):=\begin{bmatrix}\frac{1}{8}(4\nu^2+3)&1\\ -1 & 0\end{bmatrix}\sqrt{\pi}\,\e^{-\im\frac{\pi}{4}}\begin{bmatrix}I_{\nu}((\e^{-\im\pi}\zeta)^{\frac{1}{2}}) & -\frac{\im}{\pi}K_{\nu}((\e^{-\im\pi}\zeta)^{\frac{1}{2}})\smallskip\\
	(\e^{-\im\pi}\zeta)^{\frac{1}{2}}I_{\nu}'((\e^{-\im\pi}\zeta)^{\frac{1}{2}}) & -\frac{\im}{\pi}(\e^{-\im\pi}\zeta)^{\frac{1}{2}}K_{\nu}'((\e^{-\im\pi}\zeta)^{\frac{1}{2}})\end{bmatrix}\mathcal{N}(\zeta;\nu),
\end{equation}
in terms of the piecewise constant multiplier
\begin{equation*}
	\mathcal{N}(\zeta;\nu):=\begin{cases}\bigl[\begin{smallmatrix}1&0\\ -\e^{-\im\pi\nu}&1\end{smallmatrix}\bigr],&\textnormal{arg}\,\zeta\in(0,\frac{\pi}{3})\smallskip\\ I,&\textnormal{arg}\,\zeta\in(\frac{\pi}{3},\frac{5\pi}{3})\smallskip\\
	\bigl[\begin{smallmatrix}1&0\\ \e^{\im\pi\nu} & 1\end{smallmatrix}\bigr],&\textnormal{arg}\,\zeta\in(\frac{5\pi}{3},2\pi)\end{cases},
\end{equation*}
and realise that $J(\zeta)=J(\zeta;\nu)$ is uniquely characterised by the following four properties:
\begin{problem}\label{BessRHP} The model function $J(\zeta)=J(\zeta;\nu)\in\mathbb{C}^{2\times 2}$ defined in \eqref{A11} is such that:
\begin{enumerate}
	\item[(1)] $\zeta\mapsto J(\zeta)$ is analytic for $\zeta\in\mathbb{C}\setminus\Sigma_J$ with $\Sigma_J$ shown in Figure \ref{figc}. On $\Sigma_J\setminus\{0\}$, $J(\zeta)$ admits continuous limiting values $J_{\pm}(\zeta)$ as we approach $\Sigma_J$ from either side of $\mathbb{C}\setminus\Sigma_J$.
	\item[(2)] The limiting values $J_{\pm}(\zeta)$ on $(\Sigma_J\setminus\{0\})\ni\zeta$ satisfy
	\begin{equation*}
		J_+(\zeta)=J_-(\zeta)\begin{bmatrix}1 & 0\\ \e^{-\im\pi\nu}&1\end{bmatrix},\ \ \textnormal{arg}\,\zeta=\frac{\pi}{3};\hspace{1cm}J_+(\zeta)=J_-(\zeta)\begin{bmatrix}1 & 0\\ \e^{\im\pi\nu}&1\end{bmatrix},\ \ \textnormal{arg}\,\zeta=\frac{5\pi}{3};
	\end{equation*}
	\begin{equation*}
		J_+(\zeta)=J_-(\zeta)\begin{bmatrix}0&1\\ -1 & 0\end{bmatrix},\ \ \zeta>0.
	\end{equation*}
	\item[(3)] Near $\zeta=0$, $\zeta\mapsto J(\zeta)$ is weakly singular in that, with some $\zeta\mapsto\widecheck{J}(\zeta)=\widecheck{J}(\zeta;\nu)$ analytic at $\zeta=0$ and non-vanishing, $J(\zeta)$ is of the form
	\begin{equation*}
		J(\zeta)=\widecheck{J}(\zeta)\big(\e^{-\im\pi}\zeta\big)^{\frac{\nu}{2}\sigma_3}\begin{bmatrix}1 & \frac{\im}{2}\frac{1}{\sin(\pi\nu)}\\ 0 & 1\end{bmatrix}\mathcal{N}(\zeta;\nu),
	\end{equation*}
	provided $\nu\notin\mathbb{Z}_{\geq 0}$ and if $\nu\in\mathbb{Z}_{\geq 0}$ then it is of the form
	\begin{equation*}
		J(\zeta)=\widecheck{J}(\zeta)\big(\e^{-\im\pi}\zeta\big)^{\frac{\nu}{2}\sigma_3}\begin{bmatrix}1 & -\frac{\e^{\im\pi\nu}}{2\pi\im}\ln(\e^{-\im\pi}\zeta)\\ 0 & 1\end{bmatrix}\mathcal{N}(\zeta;\nu)
	\end{equation*}
	Throughout, principal branches with cuts on $(-\infty,0]\subset\mathbb{R}$ are employed.
	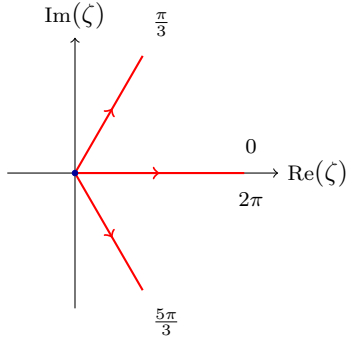
\begin{figure}[tbh]
	\begin{tikzpicture}[xscale=0.9,yscale=0.9]
	 \draw [->] (-1,0) -- (3,0) node [right] {$\footnotesize{\textnormal{Re}(\zeta)}$};
	   \draw [->] (0,-2) -- (0,2) node [above] {$\footnotesize{\textnormal{Im}(\zeta)}$};
	\draw [thick, color=red, decoration={markings, mark=at position 0.5 with {\arrow{>}}},  postaction={decorate}] (0,0) -- (2.5,0);
	\draw [thick, color=red, decoration={markings, mark=at position 0.5 with {\arrow{<}}}, postaction={decorate}] (1,1.73205) -- (0,0);
	\draw [thick, color=red, decoration={markings, mark=at position 0.5 with {\arrow{<}}}, postaction={decorate}] (1,-1.73205) -- (0,0);
	\node [below] at (2.6,-0.15) {{\footnotesize $2\pi$}};
	\node [above] at (2.6,0.15) {{\footnotesize $0$}};
	\node [right] at (1,2.2) {{\footnotesize $\frac{\pi}{3}$}};
	\node [right] at (1,-2.2) {{\footnotesize $\frac{5\pi}{3}$}};
	\draw [fill, color=blue!60!black] (0,0) circle [radius=0.04];
\end{tikzpicture}
\caption{The oriented jump contour $\Sigma_J$, shown in red, for the model function $J(\zeta)$ in the complex $\zeta$-plane, with indicated values of $\textnormal{arg}\,\zeta$ on the three rays.}
\label{figc}
\end{figure}

	\item[(4)] As $\zeta\rightarrow\infty$ and $\zeta\notin\Sigma_J$, $J(\zeta)$ is normalised as follows,
	\begin{align*}
		J(\zeta)=&\,\Bigg\{I-\frac{1}{\zeta}\begin{bmatrix}-\frac{1}{128}(4\nu^2-1)(4\nu^2-9)&-\frac{1}{1536}(4\nu^2-1)(4\nu^2-9)(4\nu^2-13)\smallskip\\ \frac{1}{8}(4\nu^2-1) & \frac{1}{128}(4\nu^2-1)(4\nu^2-9)\end{bmatrix}+\mathcal{O}\big(\zeta^{-2}\big)\Bigg\}\\
		&\hspace{3cm}\times\big(\e^{-\im\pi}\zeta\big)^{\frac{1}{4}\sigma_3}\frac{1}{\sqrt{2}}\begin{bmatrix}1&1\\ -1 & 1\end{bmatrix}\e^{-\im\frac{\pi}{4}\sigma_3}\exp\Big[\big(\e^{-\im\pi}\zeta\big)^{\frac{1}{2}}\sigma_3\Big],
	\end{align*}
	with $\textnormal{arg}\,\zeta\in(0,2\pi)$ throughout and the fractional power with cut on $(-\infty,0]\subset\mathbb{R}$
\end{enumerate}
\end{problem}
\begin{rem}\label{nearzero} Utilising \cite[$\S 10.25,10.27$]{NIST} we record the following expressions for $\widecheck{J}(\zeta)=\widecheck{J}(\zeta;\nu)$ in condition $(3)$ of RHP \ref{BessRHP},
\begin{equation*}
	\widecheck{J}(\zeta)=\begin{bmatrix}\frac{1}{8}(4\nu^2+3) & 1\\ -1 & 0\end{bmatrix}\sqrt{\pi}\,\e^{-\im\frac{\pi}{4}}\begin{bmatrix}I_{11}(\zeta;\nu) & -\frac{\im}{2}\frac{1}{\sin(\pi\nu)}I_{11}(\zeta;-\nu)\smallskip\\ I_{21}(\zeta;\nu) & -\frac{\im}{2}\frac{1}{\sin(\pi\nu)}I_{21}(\zeta;-\nu)\end{bmatrix},\ \ \nu\notin\mathbb{Z}_{\geq 0},
\end{equation*}
in terms of the entire functions
\begin{equation*}
	I_{11}(\zeta;\nu):=\frac{1}{2^{\nu}}\sum_{k=0}^{\infty}\frac{(-\frac{\zeta}{4})^k}{k!\,\Gamma(1+\nu+k)},\ \ \ \ \ \ \ \ I_{21}(\zeta;\nu):=\frac{1}{2^{\nu}}\sum_{k=0}^{\infty}\frac{(\nu+2k)(-\frac{\zeta}{4})^k}{k!\,\Gamma(1+\nu+k)}.
\end{equation*}
In the non-generic case $\nu\in\mathbb{Z}_{\geq 0}$ we have instead
\begin{align*}
	\widecheck{J}(\zeta;\nu)=\begin{bmatrix}\frac{1}{8}(4\nu^2+3) & 1\\ -1 & 0\end{bmatrix}\sqrt{\pi}\,\e^{-\im\frac{\pi}{4}}\begin{bmatrix}I_{11}(\zeta;\nu) & I_{12}(\zeta;\nu)\smallskip\\ I_{21}(\zeta;\nu) & I_{22}(\zeta;\nu)+\frac{\im}{\pi}I_{11}(\zeta;\nu)\zeta^{\nu}\end{bmatrix}\begin{bmatrix}1 & \frac{1}{\im\pi}\zeta^{\nu}\ln 2\\ 0 & 1\end{bmatrix},\ \ \nu\in\mathbb{Z}_{\geq 0},
\end{align*}
in terms of $I_{11}(\zeta;\nu),I_{21}(\zeta;\nu)$ as before and the entire functions
\begin{equation*}
	I_{12}(\zeta;\nu):=-\frac{\im}{2\pi}\Bigg[2^{\nu}\sum_{k=0}^{\nu-1}\frac{\Gamma(\nu-k)}{k!}\Big(\frac{\zeta}{4}\Big)^k+2^{-\nu}\zeta^{\nu}\sum_{k=0}^{\infty}\big(\psi(k+1)+\psi(\nu+k+1)\big)\frac{(-\frac{\zeta}{4})^k}{k!(\nu+k)!}\Bigg],
\end{equation*}
\begin{equation*}
	I_{22}(\zeta;\nu):=\frac{\im}{2\pi}\Bigg[2^{\nu}\sum_{k=0}^{\nu-1}\frac{\Gamma(\nu-k)}{k!}(\nu-2k)\Big(\frac{\zeta}{4}\Big)^k-2^{-\nu}\zeta^{\nu}\sum_{k=0}^{\infty}\big(\psi(k+1)+\psi(\nu+k+1)\big)(\nu+2k)\frac{(-\frac{\zeta}{4})^k}{k!(\nu+k)!}\Bigg],
\end{equation*}
using the convention $\sum_{k=0}^na_k=0$ for $n<0$.
\end{rem}
\begin{rem}\label{impcons} Utilising \eqref{A11} the unique solution of RHP \ref{PIIImodel} for $s=0$ is given by
\begin{equation*}
	Q(\zeta;x,0)=x^{-\frac{1}{2}\sigma_3}\e^{\im\frac{\pi}{4}\sigma_3}J\big(x^2\zeta;0\big),\ \ \ \zeta\in\mathbb{C}\setminus\Sigma_J,\ \ \ x>0,
\end{equation*}
so, in particular, $Q_1^{21}(x,0)=-\frac{\im}{8x}$.
\end{rem}
\end{appendix}


\begin{bibsection}
\begin{biblist}

\bib{ABGS}{article}{
   AUTHOR = {Assiotis, Theodoros},
   author={Bedert, Benjamin},
   author={Gunes, Mustafa Alper},
   author={Soor, Arun},
     TITLE = {On a distinguished family of random variables and {P}ainlev\'{e}
              equations},
   JOURNAL = {Probab. Math. Phys.},
  FJOURNAL = {Probability and Mathematical Physics},
    VOLUME = {2},
      YEAR = {2021},
    NUMBER = {3},
     PAGES = {613--642},
      ISSN = {2690-0998},
   MRCLASS = {60B20 (11M50 15B52 33C10 37J65)},
  MRNUMBER = {4408021},
MRREVIEWER = {J\"{o}rg Neunh\"{a}userer},
       DOI = {10.2140/pmp.2021.2.613},
       URL = {https://doi-org.bris.idm.oclc.org/10.2140/pmp.2021.2.613},
}

\bib{AGKW}{article}{
      title={Exchangeable arrays and integrable systems for characteristic polynomials of random matrices}, 
      author={Assiotis, Theodoros },
      author={Gunes, Mustafa Alper},
      author={Keating, Jonathan P.},
      author={Wei, Fei},
      year={2024},
      eprint={https://arxiv.org/abs/2407.19233},
      archivePrefix={arXiv},
      primaryClass={math.PR},
}

\bib{assiotis2025joint}{article}{
  title={Joint Moments of Characteristic Polynomials from the Orthogonal and Unitary Symplectic Groups},
  author={Assiotis, Theodoros },
      author={Gunes, Mustafa Alper},
      author={Keating, Jonathan P.},
      author={Wei, Fei},
  year={2025},
  eprint={https://arxiv.org/abs/2508.09910},
  archivePrefix={arXiv},
  primaryClass={math-ph}
}

\bib{AKW}{article}{
  AUTHOR = {Assiotis, Theodoros},
  author={Keating, Jonathan P.},
  author={Warren, Jon},
     TITLE = {On the joint moments of the characteristic polynomials of
              random unitary matrices},
   JOURNAL = {Int. Math. Res. Not. IMRN},
  FJOURNAL = {International Mathematics Research Notices. IMRN},
      YEAR = {2022},
    NUMBER = {18},
     PAGES = {14564--14603},
      ISSN = {1073-7928},
   MRCLASS = {60B20 (11M50 15B52 43A80)},
  MRNUMBER = {4485964},
MRREVIEWER = {Giorgio Cipolloni},
       DOI = {10.1093/imrn/rnab336},
       URL = {https://doi-org.bris.idm.oclc.org/10.1093/imrn/rnab336},
}

\bib{Bailey_2019}{article}{
   AUTHOR = {Bailey, E. C.},
   author={Bettin, S.},
   author={Blower, G.},
   author={Conrey, J. B.},
   author={Prokhorov, A.},
   author={Rubinstein, M. O.},
   author={Snaith, N. C.},
     TITLE = {Mixed moments of characteristic polynomials of random unitary
              matrices},
   JOURNAL = {J. Math. Phys.},
  FJOURNAL = {Journal of Mathematical Physics},
    VOLUME = {60},
      YEAR = {2019},
    NUMBER = {8},
     PAGES = {083509, 26},
      ISSN = {0022-2488},
   MRCLASS = {15B52 (11M06 11M50 15A15 34M55)},
  MRNUMBER = {3995715},
MRREVIEWER = {Santosh Kumar},
       DOI = {10.1063/1.5092780},
       URL = {https://doi-org.bris.idm.oclc.org/10.1063/1.5092780},
}

\bib{Basor_2019}{article}{
   AUTHOR = {Basor, Estelle},
   author={Bleher, Pavel},
   author={Buckingham, Robert},
   author={Grava, Tamara},
   author={Its, Alexander},
   author={Its, Elizabeth},
   author={Keating, Jonathan P.},
     TITLE = {A representation of joint moments of {CUE} characteristic
              polynomials in terms of {P}ainlev\'{e} functions},
   JOURNAL = {Nonlinearity},
  FJOURNAL = {Nonlinearity},
    VOLUME = {32},
      YEAR = {2019},
    NUMBER = {10},
     PAGES = {4033--4078},
      ISSN = {0951-7715},
   MRCLASS = {60B20 (11M50 30E25 33E17 34M55 35Q15)},
  MRNUMBER = {4012580},
MRREVIEWER = {Shuaixia Xu},
       DOI = {10.1088/1361-6544/ab28c7},
       URL = {https://doi-org.bris.idm.oclc.org/10.1088/1361-6544/ab28c7},
}

\bib{Borodin_2001}{article}{
   AUTHOR = {Borodin, Alexei},
   author={Olshanski, Grigori},
     TITLE = {Infinite random matrices and ergodic measures},
   JOURNAL = {Comm. Math. Phys.},
  FJOURNAL = {Communications in Mathematical Physics},
    VOLUME = {223},
      YEAR = {2001},
    NUMBER = {1},
     PAGES = {87--123},
      ISSN = {0010-3616},
   MRCLASS = {60B15 (15A52 28C10 82B31)},
  MRNUMBER = {1860761},
MRREVIEWER = {Steven B. Damelin},
       DOI = {10.1007/s002200100529},
       URL = {https://doi-org.bris.idm.oclc.org/10.1007/s002200100529},
}

\bib{BW2}{article}{
      title={On the joint moments of the characteristic polynomial for the Circular Jacobi Ensemble and Painlev\'e equations}, 
      author={Bothner, Thomas},
      author={Wei, Fei},
      year={2025},
journal={in preparation}, 
}

\bib{chen2012coulumb}{article}{
  AUTHOR = {Chen, Yang},
  author={McKay, Matthew R.},
     TITLE = {Coulumb fluid, {P}ainlev\'{e} transcendents, and the information
              theory of {MIMO} systems},
   JOURNAL = {IEEE Trans. Inform. Theory},
  FJOURNAL = {Institute of Electrical and Electronics Engineers.
              Transactions on Information Theory},
    VOLUME = {58},
      YEAR = {2012},
    NUMBER = {7},
     PAGES = {4594--4634},
      ISSN = {0018-9448},
   MRCLASS = {94A17 (15B52)},
  MRNUMBER = {2949840},
       DOI = {10.1109/TIT.2012.2195154},
       URL = {https://doi-org.bris.idm.oclc.org/10.1109/TIT.2012.2195154},
}

\bib{CFKRS1}{article}{
 AUTHOR = {Conrey, J. B.},
 author={Farmer, D. W.},
 author={Keating, J. P.},
 author={Rubinstein, M. O.},
 author={Snaith, N. C.},
     TITLE = {Autocorrelation of random matrix polynomials},
   JOURNAL = {Comm. Math. Phys.},
  FJOURNAL = {Communications in Mathematical Physics},
    VOLUME = {237},
      YEAR = {2003},
    NUMBER = {3},
     PAGES = {365--395},
      ISSN = {0010-3616},
   MRCLASS = {11M41 (11M06 60B15)},
  MRNUMBER = {1993332},
MRREVIEWER = {Ze\'{e}v Rudnick},
       DOI = {10.1007/s00220-003-0852-2},
       URL = {https://doi-org.bris.idm.oclc.org/10.1007/s00220-003-0852-2},
}

\bib{CFKRS2}{article}{
  AUTHOR = {Conrey, J. B.},
  author={Farmer, D. W.},
  author={Keating, J. P.},
  author={Rubinstein, M. O.},
  author={Snaith, N. C.},
     TITLE = {Integral moments of {$L$}-functions},
   JOURNAL = {Proc. London Math. Soc. (3)},
  FJOURNAL = {Proceedings of the London Mathematical Society. Third Series},
    VOLUME = {91},
      YEAR = {2005},
    NUMBER = {1},
     PAGES = {33--104},
      ISSN = {0024-6115},
   MRCLASS = {11M26},
  MRNUMBER = {2149530},
MRREVIEWER = {K. Soundararajan},
       DOI = {10.1112/S0024611504015175},
       URL = {https://doi-org.bris.idm.oclc.org/10.1112/S0024611504015175},
}

\bib{conreyetal}{article}{
	AUTHOR = {Conrey, J. B.},
	author={Rubinstein, M. O.},
	author={Snaith, N. C.},
     TITLE = {Moments of the derivative of characteristic polynomials with
              an application to the {R}iemann zeta function},
   JOURNAL = {Comm. Math. Phys.},
  FJOURNAL = {Communications in Mathematical Physics},
    VOLUME = {267},
      YEAR = {2006},
    NUMBER = {3},
     PAGES = {611--629},
      ISSN = {0010-3616},
   MRCLASS = {11M06},
  MRNUMBER = {2249784},
MRREVIEWER = {Tanguy Rivoal},
       DOI = {10.1007/s00220-006-0090-5},
       URL = {https://doi-org.bris.idm.oclc.org/10.1007/s00220-006-0090-5},
}

\bib{Dehaye2008}{article}{
AUTHOR = {Dehaye, Paul-Olivier},
     TITLE = {Joint moments of derivatives of characteristic polynomials},
   JOURNAL = {Algebra Number Theory},
  FJOURNAL = {Algebra \& Number Theory},
    VOLUME = {2},
      YEAR = {2008},
    NUMBER = {1},
     PAGES = {31--68},
      ISSN = {1937-0652},
   MRCLASS = {15A52 (05E10 33C80 60B15)},
  MRNUMBER = {2377362},
       DOI = {10.2140/ant.2008.2.31},
       URL = {https://doi-org.bris.idm.oclc.org/10.2140/ant.2008.2.31},
}


\bib{Deift1999}{book}{
  AUTHOR = {Deift, P. A.},
     TITLE = {Orthogonal polynomials and random matrices: a
              {R}iemann-{H}ilbert approach},
    SERIES = {Courant Lecture Notes in Mathematics},
    VOLUME = {3},
 PUBLISHER = {New York University, Courant Institute of Mathematical
              Sciences, New York; American Mathematical Society, Providence,
              RI},
      YEAR = {1999},
     PAGES = {viii+273},
      ISBN = {0-9658703-2-4; 0-8218-2695-6},
   MRCLASS = {47B80 (15A52 30E25 33D45 37K10 42C05 47B36 60F99)},
  MRNUMBER = {1677884},
MRREVIEWER = {Alexander Vladimirovich Kitaev},
}



 \bib{DKMVZ}{article}{
 AUTHOR = {Deift, P.},
 author={Kriecherbauer, T.},
 author={McLaughlin, K. T.-R.},
 author={Venakides, S.},
 author={Zhou, X.},
     TITLE = {Uniform asymptotics for polynomials orthogonal with respect to
              varying exponential weights and applications to universality
              questions in random matrix theory},
   JOURNAL = {Comm. Pure Appl. Math.},
  FJOURNAL = {Communications on Pure and Applied Mathematics},
    VOLUME = {52},
      YEAR = {1999},
    NUMBER = {11},
     PAGES = {1335--1425},
      ISSN = {0010-3640},
   MRCLASS = {42C05 (15A52 41A60 82B41)},
  MRNUMBER = {1702716},
MRREVIEWER = {D. S. Lubinsky},
       DOI =
              {10.1002/(SICI)1097-0312(199911)52:11<1335::AID-CPA1>3.0.CO;2-1},
       URL =
              {https://doi-org.bris.idm.oclc.org/10.1002/(SICI)1097-0312(199911)52:11<1335::AID-CPA1>3.0.CO;2-1},
}


\bib{DZ}{article}{
 AUTHOR = {Deift, P.},
 author={Zhou, X.},
     TITLE = {A steepest descent method for oscillatory {R}iemann-{H}ilbert
              problems. {A}symptotics for the {MK}d{V} equation},
   JOURNAL = {Ann. of Math. (2)},
  FJOURNAL = {Annals of Mathematics. Second Series},
    VOLUME = {137},
      YEAR = {1993},
    NUMBER = {2},
     PAGES = {295--368},
      ISSN = {0003-486X},
   MRCLASS = {35Q53 (34A55 34L25 35Q15 35Q55)},
  MRNUMBER = {1207209},
MRREVIEWER = {Alexey V. Samokhin},
       DOI = {10.2307/2946540},
       URL = {https://doi-org.bris.idm.oclc.org/10.2307/2946540},
}


\bib{FN}{article}{
    AUTHOR = {Flaschka, Hermann},
    author={Newell, Alan C.},
     TITLE = {Monodromy- and spectrum-preserving deformations. {I}},
   JOURNAL = {Comm. Math. Phys.},
  FJOURNAL = {Communications in Mathematical Physics},
    VOLUME = {76},
      YEAR = {1980},
    NUMBER = {1},
     PAGES = {65--116},
      ISSN = {0010-3616},
   MRCLASS = {35Q20 (14D05 58F07 81C05)},
  MRNUMBER = {588248},
MRREVIEWER = {H\'{e}l\`ene Airault},
       URL = {http://projecteuclid.org.bris.idm.oclc.org/euclid.cmp/1103908189},
}


\bib{FIK2}{article}{
  AUTHOR = {Fokas, A. S.},
  author={Its, A. R.},
  author={Kitaev, A. V.},
     TITLE = {Discrete {P}ainlev\'{e} equations and their appearance in quantum
              gravity},
   JOURNAL = {Comm. Math. Phys.},
  FJOURNAL = {Communications in Mathematical Physics},
    VOLUME = {142},
      YEAR = {1991},
    NUMBER = {2},
     PAGES = {313--344},
      ISSN = {0010-3616},
   MRCLASS = {58F07 (81T40)},
  MRNUMBER = {1137067},
MRREVIEWER = {Nikolai A. Kostov},
       URL = {http://projecteuclid.org.bris.idm.oclc.org/euclid.cmp/1104248588},
}

\bib{forrester2025higher}{article}{
  title={Higher order linear differential equations for unitary matrix integrals: applications and generalisations},
  author={Forrester, Peter J},
  author={Wei, Fei},
  year={2025}
  eprint={https://arxiv.org/abs/2508.20797},
  archivePrefix={arXiv},
  primaryClass={math-ph}
}


\bib{ForresterWittePainleve1}{article}{
    AUTHOR = {Forrester, P. J.},
    author={Witte, N. S.},
     TITLE = {Application of the {$\tau$}-function theory of {P}ainlev\'{e}
              equations to random matrices: {PIV}, {PII} and the {GUE}},
   JOURNAL = {Comm. Math. Phys.},
  FJOURNAL = {Communications in Mathematical Physics},
    VOLUME = {219},
      YEAR = {2001},
    NUMBER = {2},
     PAGES = {357--398},
      ISSN = {0010-3616},
   MRCLASS = {82B41 (15A52 34M55)},
  MRNUMBER = {1833807},
MRREVIEWER = {Oleksiy Khorunzhiy},
       DOI = {10.1007/s002200100422},
       URL = {https://doi-org.bris.idm.oclc.org/10.1007/s002200100422},
}

\bib{forrester2002application}{article}{
  AUTHOR = {Forrester, P. J.},
  author={Witte, N. S.},
     TITLE = {Application of the {$\tau$}-function theory of {P}ainlev\'{e}
              equations to random matrices: {$\rm P_V$}, {$\rm P_{III}$},
              the {LUE}, {JUE}, and {CUE}},
   JOURNAL = {Comm. Pure Appl. Math.},
  FJOURNAL = {Communications on Pure and Applied Mathematics},
    VOLUME = {55},
      YEAR = {2002},
    NUMBER = {6},
     PAGES = {679--727},
      ISSN = {0010-3640},
   MRCLASS = {33E17 (34M55 37K10 82B31)},
  MRNUMBER = {1885665},
MRREVIEWER = {Mark Adler},
       DOI = {10.1002/cpa.3021},
       URL = {https://doi-org.bris.idm.oclc.org/10.1002/cpa.3021},
}

\bib{Forrester_2006}{article}{
    AUTHOR = {Forrester, P. J.},
    author={Witte, N. S.},
     TITLE = {Boundary conditions associated with the {P}ainlev\'{e} {III{$'$}}
              and {V} evaluations of some random matrix averages},
   JOURNAL = {J. Phys. A},
  FJOURNAL = {Journal of Physics. A. Mathematical and General},
    VOLUME = {39},
      YEAR = {2006},
    NUMBER = {28},
     PAGES = {8983--8995},
      ISSN = {0305-4470},
   MRCLASS = {82B41 (11M06 11Z05 33E17 34M55)},
  MRNUMBER = {2240469},
MRREVIEWER = {David W. Farmer},
       DOI = {10.1088/0305-4470/39/28/S13},
       URL = {https://doi-org.bris.idm.oclc.org/10.1088/0305-4470/39/28/S13},
}

\bib{gonek2007hybrid}{article}{
  AUTHOR = {Gonek, S. M.},
  author={Hughes, C. P.},
  author={Keating, J. P.},
     TITLE = {A hybrid {E}uler-{H}adamard product for the {R}iemann zeta
              function},
   JOURNAL = {Duke Math. J.},
  FJOURNAL = {Duke Mathematical Journal},
    VOLUME = {136},
      YEAR = {2007},
    NUMBER = {3},
     PAGES = {507--549},
      ISSN = {0012-7094},
   MRCLASS = {11M26 (11M06 15A52)},
  MRNUMBER = {2309173},
MRREVIEWER = {Steven Joel Miller},
       DOI = {10.1215/S0012-7094-07-13634-2},
       URL = {https://doi-org.bris.idm.oclc.org/10.1215/S0012-7094-07-13634-2},
}

\bib{GR}{book}{
  AUTHOR = {Gradshteyn, I. S.},
  author={Ryzhik, I. M.},
     TITLE = {Table of integrals, series, and products},
   EDITION = {Seventh Edition},
      NOTE = {Translated from the Russian,
              Translation edited and with a preface by Alan Jeffrey and
              Daniel Zwillinger,
              With one CD-ROM (Windows, Macintosh and UNIX)},
 PUBLISHER = {Elsevier/Academic Press, Amsterdam},
      YEAR = {2007},
     PAGES = {xlviii+1171},
      ISBN = {978-0-12-373637-6; 0-12-373637-4},
   MRCLASS = {00A22 (33-00 65-00 65A05)},
  MRNUMBER = {2360010},
}

\bib{Hughes}{article}{
  title={On the characteristic polynomial of a random unitary matrix and the {R}iemann
zeta function},
  author={Hughes, Christopher P.},
  journal={PhD Thesis, University of Bristol},
  year={2001}
}



\bib{ItsAlexander}{article}{
    author = {Its, A.R.},
    isbn = {9780198744191},
    title = {Painlevé transcendents},
    booktitle = {The Oxford Handbook of Random Matrix Theory},
    publisher = {Oxford University Press},
    year = {2015},
    doi = {10.1093/oxfordhb/9780198744191.013.9},
    url = {https://doi.org/10.1093/oxfordhb/9780198744191.013.9},
}

\bib{ItsN}{book}{
AUTHOR = {Its, Alexander R.},
author={Novokshenov, Victor Yu.},
     TITLE = {The isomonodromic deformation method in the theory of
              {P}ainlev\'{e} equations},
    SERIES = {Lecture Notes in Mathematics},
    VOLUME = {1191},
 PUBLISHER = {Springer-Verlag, Berlin},
      YEAR = {1986},
     PAGES = {iv+313},
      ISBN = {3-540-16483-9},
   MRCLASS = {34A20 (35Q20 58F07)},
  MRNUMBER = {851569},
MRREVIEWER = {Alexander A. Pankov},
       DOI = {10.1007/BFb0076661},
       URL = {https://doi-org.bris.idm.oclc.org/10.1007/BFb0076661},
}


\bib{JM}{article}{
  AUTHOR = {Jimbo, Michio},
  author={Miwa, Tetsuji},
     TITLE = {Monodromy preserving deformation of linear ordinary
              differential equations with rational coefficients. {II}},
   JOURNAL = {Phys. D},
  FJOURNAL = {Physica D. Nonlinear Phenomena},
    VOLUME = {2},
      YEAR = {1981},
    NUMBER = {3},
     PAGES = {407--448},
      ISSN = {0167-2789},
   MRCLASS = {34A20 (14K25 58A15 58F07 81C05)},
  MRNUMBER = {625446},
MRREVIEWER = {V. A. Golubeva},
       DOI = {10.1016/0167-2789(81)90021-X},
       URL = {https://doi-org.bris.idm.oclc.org/10.1016/0167-2789(81)90021-X},
}




 \bib{keating2000random}{article}{
  AUTHOR = {Keating, J. P.},
  author={Snaith, N. C.},
     TITLE = {Random matrix theory and {$\zeta(1/2+it)$}},
   JOURNAL = {Comm. Math. Phys.},
  FJOURNAL = {Communications in Mathematical Physics},
    VOLUME = {214},
      YEAR = {2000},
    NUMBER = {1},
     PAGES = {57--89},
      ISSN = {0010-3616},
   MRCLASS = {11M26 (15A52 82B41)},
  MRNUMBER = {1794265},
MRREVIEWER = {Ze\'{e}v Rudnick},
       DOI = {10.1007/s002200000261},
       URL = {https://doi-org.bris.idm.oclc.org/10.1007/s002200000261},
}

\bib{keating-fei}{article}{
  AUTHOR = {Keating, Jonathan P.},
  author={Wei, Fei},
     TITLE = {Joint moments of higher order derivatives of {CUE}
              characteristic polynomials {II}: structures, recursive
              relations, and applications},
   JOURNAL = {Nonlinearity},
  FJOURNAL = {Nonlinearity},
    VOLUME = {37},
      YEAR = {2024},
    NUMBER = {8},
     PAGES = {Paper No. 085009, 54},
      ISSN = {0951-7715},
   MRCLASS = {11M50 (33C90 33E17 60B20)},
  MRNUMBER = {4767096},
MRREVIEWER = {Steven Joel Miller},
       DOI = {10.1088/1361-6544/ad5948},
       URL = {https://doi-org.bris.idm.oclc.org/10.1088/1361-6544/ad5948},
}

\bib{O}{article}{
  AUTHOR = {Okamoto, Kazuo},
     TITLE = {Studies on the {P}ainlev\'{e} equations. {IV}. {T}hird {P}ainlev\'{e}
              equation {$P_{{\rm III}}$}},
   JOURNAL = {Funkcial. Ekvac.},
  FJOURNAL = {Funkcialaj Ekvacioj. Serio Internacia},
    VOLUME = {30},
      YEAR = {1987},
    NUMBER = {2-3},
     PAGES = {305--332},
      ISSN = {0532-8721},
   MRCLASS = {58F05 (34C20 58F35)},
  MRNUMBER = {927186},
MRREVIEWER = {H\'{e}l\`ene Airault},
       URL = {http://www.math.kobe-u.ac.jp/~fe/xml/mr0927186.xml},
}

\bib{NIST}{book}{
 TITLE = {N{IST} handbook of mathematical functions},
    EDITOR = {Olver, Frank W. J.},
    editor={Lozier, Daniel W.},
    editor={Boisvert, Ronald F.},
    editor={Clark, Charles W.},
      NOTE = {With 1 CD-ROM (Windows, Macintosh and UNIX)},
 PUBLISHER = {U.S. Department of Commerce, National Institute of Standards
              and Technology, Washington, DC; Cambridge University Press,
              Cambridge},
      YEAR = {2010},
     PAGES = {xvi+951},
      ISBN = {978-0-521-14063-8},
   MRCLASS = {33-00 (00A20 65-00)},
  MRNUMBER = {2723248},
}

\bib{Qiu}{article}{
    AUTHOR = {Qiu, Yanqi},
     TITLE = {Infinite random matrices \& ergodic decomposition of finite
              and infinite {H}ua-{P}ickrell measures},
   JOURNAL = {Adv. Math.},
  FJOURNAL = {Advances in Mathematics},
    VOLUME = {308},
      YEAR = {2017},
     PAGES = {1209--1268},
      ISSN = {0001-8708},
   MRCLASS = {60B20 (22E66 28D05 33C45)},
  MRNUMBER = {3600086},
       DOI = {10.1016/j.aim.2017.01.003},
       URL = {https://doi-org.bris.idm.oclc.org/10.1016/j.aim.2017.01.003},
}


\bib{simm2024moments}{article}{
  title={On moments of the derivative of CUE characteristic polynomials and the Riemann zeta function},
  author={Simm, Nick},
  author={Wei, Fei},
  year={2024},
  eprint={https://arxiv.org/abs/2409.03687},
  archivePrefix={arXiv},
  primaryClass={math.PR}
}

\bib{tracy1994fredholm}{article}{
 AUTHOR = {Tracy, Craig A.},
 author={Widom, Harold},
     TITLE = {Fredholm determinants, differential equations and matrix
              models},
   JOURNAL = {Comm. Math. Phys.},
  FJOURNAL = {Communications in Mathematical Physics},
    VOLUME = {163},
      YEAR = {1994},
    NUMBER = {1},
     PAGES = {33--72},
      ISSN = {0010-3616},
   MRCLASS = {82B05 (33C90 47A75 47G10 47N55 82B10)},
  MRNUMBER = {1277933},
MRREVIEWER = {Peter J. Forrester},
       URL = {http://projecteuclid.org.bris.idm.oclc.org/euclid.cmp/1104270379},
}

\bib{V}{article}{
  AUTHOR = {Vanlessen, M.},
     TITLE = {Strong asymptotics of {L}aguerre-type orthogonal polynomials
              and applications in random matrix theory},
   JOURNAL = {Constr. Approx.},
  FJOURNAL = {Constructive Approximation. An International Journal for
              Approximations and Expansions},
    VOLUME = {25},
      YEAR = {2007},
    NUMBER = {2},
     PAGES = {125--175},
      ISSN = {0176-4276},
   MRCLASS = {15A52 (33C10 33C45 82B44)},
  MRNUMBER = {2283495},
MRREVIEWER = {Guangming Pan},
       DOI = {10.1007/s00365-005-0611-z},
       URL = {https://doi-org.bris.idm.oclc.org/10.1007/s00365-005-0611-z},
}

\bib{W}{article}{
 AUTHOR = {Winn, B.},
     TITLE = {Derivative moments for characteristic polynomials from the
              {CUE}},
   JOURNAL = {Comm. Math. Phys.},
  FJOURNAL = {Communications in Mathematical Physics},
    VOLUME = {315},
      YEAR = {2012},
    NUMBER = {2},
     PAGES = {531--562},
      ISSN = {0010-3616},
   MRCLASS = {60B20},
  MRNUMBER = {2971735},
MRREVIEWER = {VenKata K. B. Kota},
       DOI = {10.1007/s00220-012-1512-1},
       URL = {https://doi-org.bris.idm.oclc.org/10.1007/s00220-012-1512-1},
}

\bib{XZ}{article}{
  AUTHOR = {Xu, Shuai-Xia},
  author={Zhao, Yu-Qiu},
     TITLE = {Critical edge behavior in the modified {J}acobi ensemble and
              {P}ainlev\'{e} equations},
   JOURNAL = {Nonlinearity},
  FJOURNAL = {Nonlinearity},
    VOLUME = {28},
      YEAR = {2015},
    NUMBER = {6},
     PAGES = {1633--1674},
      ISSN = {0951-7715},
   MRCLASS = {60B20 (30E25 34M55)},
  MRNUMBER = {3350603},
       DOI = {10.1088/0951-7715/28/6/1633},
       URL = {https://doi-org.bris.idm.oclc.org/10.1088/0951-7715/28/6/1633},
}

\bib{Z}{article}{
 AUTHOR = {Zhou, Xin},
     TITLE = {The {R}iemann-{H}ilbert problem and inverse scattering},
   JOURNAL = {SIAM J. Math. Anal.},
  FJOURNAL = {SIAM Journal on Mathematical Analysis},
    VOLUME = {20},
      YEAR = {1989},
    NUMBER = {4},
     PAGES = {966--986},
      ISSN = {0036-1410},
   MRCLASS = {34B25 (35G15 45F15 45P05)},
  MRNUMBER = {1000732},
MRREVIEWER = {David J. Kaup},
       DOI = {10.1137/0520065},
       URL = {https://doi-org.bris.idm.oclc.org/10.1137/0520065},
}






%
%
























%
%
%
%














\end{biblist}
\end{bibsection}


\end{document}